\documentclass[11pt,letterpaper]{article}
\usepackage[utf8x]{inputenc}
\usepackage{float}
\usepackage{thm-restate}
\usepackage{caption}

\usepackage[linesnumbered,boxed]{algorithm2e}
\SetAlCapSkip{1em}
\SetAlCapNameFnt{\small}
\SetAlCapFnt{\small}

\usepackage{mathtools}

\def\showauthornotes{1}

\usepackage{xspace,enumerate}

\usepackage[dvipsnames]{xcolor}

\usepackage[T1]{fontenc}
\usepackage[full]{textcomp}

\usepackage[american]{babel}

\usepackage{mathtools}

\usepackage{hyperref}
\usepackage{amsthm}
\usepackage{cleveref}

\newtheorem{theorem}{Theorem}[section]
\newtheorem*{theorem*}{Theorem}

\newtheorem{proposition}[theorem]{Proposition}
\newtheorem*{proposition*}{Proposition}
\newtheorem{lemma}[theorem]{Lemma}
\newtheorem*{lemma*}{Lemma}

\newtheorem*{conjecture*}{Conjecture}
\newtheorem{fact}[theorem]{Fact}
\newtheorem*{fact*}{Fact}

\newtheorem*{hypothesis*}{Hypothesis}

\newtheorem{claim}[theorem]{Claim}
\newtheorem*{claim*}{Claim}

\theoremstyle{definition}
\newtheorem{definition}[theorem]{Definition}

\newtheorem{remark}[theorem]{Remark}
\newtheorem*{remark*}{Remark}

\newtheorem*{observation*}{Observation}

\usepackage[
letterpaper,
top=0.7in,
bottom=0.9in,
left=1in,
right=1in]{geometry}

\usepackage[varg]{pxfonts}
\usepackage{textcomp}
\usepackage[scr=rsfso]{mathalfa}
\usepackage{bm} 
\let\mathbb\varmathbb

\usepackage{hyperref}
\hypersetup{
colorlinks=true,
urlcolor=blue,
linkcolor=blue,
citecolor=OliveGreen,
}
\usepackage{prettyref}

\newcommand{\Abs}[1]{\left\lvert#1\right\rvert}

\newcommand{\Esymb}{\mathbb{E}}
\newcommand{\Psymb}{\mathbb{P}}

\DeclareMathOperator*{\E}{\Esymb}

\DeclareMathOperator*{\ProbOp}{\Psymb}

\renewcommand{\Pr}{\ProbOp}

\DeclareMathOperator{\Tr}{Tr}
\DeclareMathOperator{\poly}{poly}
\DeclareMathOperator{\negl}{negl}

\newcommand{\C}{\mathbb C}

\newcommand{\cC}{\mathcal C}

\newcommand{\cR}{\mathcal R}

\numberwithin{equation}{section}

\newcommand{\hf}{\hat{f}}
\newcommand{\binset}{\{0,1\}}
\newcommand{\lab}{\mathsf{lab}}
\newcommand{\off}{\mathsf{off}}

\newcommand{\secp}{\lambda}
\newcommand{\prg}{\mathsf{G}}

\newcommand{\hF}{\hat{F}}

\newcommand{\ncz}{\mathbf{NC}^0}
\newcommand{\qnczf}{\mathbf{QNC}^0_{\text{\rm f}}}
\newcommand{\acz}{\mathbf{AC}^0}
\newcommand{\bqp}{\mathbf{BQP}}
\newcommand{\szk}{\mathbf{SZK}}
\newcommand{\ph}{\mathbf{PH}}

\usepackage{relsize}
\usepackage[font=footnotesize]{caption}

\usepackage{appendix}
\usepackage{amssymb}
\usepackage{amsbsy}
\usepackage{amsfonts}
\usepackage{latexsym}
\usepackage{amsmath}
\usepackage{stmaryrd}

\usepackage{qcircuit}

\ifnum\showauthornotes=0
\newcommand{\hnote}[1]{\textcolor{blue}{\small {\textbf{(Henry:} #1\textbf{) }}}}
\newcommand{\znote}[1]{\textcolor{red}{\small {\textbf{(Z:} #1\textbf{) }}}}
\else
\newcommand{\hnote}[1]{}
\newcommand{\znote}[1]{}
\fi

\definecolor{DarkBlue}{rgb}{0,0,0.8}  
\definecolor{DarkOrange}{rgb}{0.8,0.4,0}  
\def\mylinkcolor{DarkBlue}
\def\mysubcolor{DarkOrange}

\hypersetup{
	colorlinks=true, urlcolor=\mylinkcolor,
  linkcolor=\mylinkcolor, citecolor=\mylinkcolor}

\newcommand{\ket}[1]{|#1\rangle}
\newcommand{\bra}[1]{\langle#1|}

\newcommand{\ketbra}[2]{|#1\rangle\! \langle #2|}

\newcommand{\reg}[1]{\mathsf{#1}}
\newcommand{\qrv}[1]{\ensuremath{\boldsymbol{#1}}}

\renewcommand{\cal}[1]{\mathcal{#1}}
\renewcommand{\sf}[1]{\mathsf{#1}}
\newcommand{\scr}[1]{\mathscr{#1}}

\newcommand{\epr}{\mathsf{EPR}}

\newcommand{\win}{1}
\newcommand{\wout}{2}

\newcommand{\puretomixed}[1]{ \left \llbracket #1 \right \rrbracket}
\newcommand{\ub}[2]{{\color{\mysubcolor} \underbracket{\textcolor{black}{#1}}_{\textcolor{\mysubcolor}{#2}}}}
\newcommand{\ubket}[2]{{\color{\mysubcolor}\underbracket{\textcolor{black}{\ket{#1}}}_{\textcolor{\mysubcolor}{#2}}}}
\newcommand{\wh}[1]{\widehat{#1}}
\newcommand{\gadget}[1]{\wh{\sf{#1}}}

\newcommand{\corr}{\mathrm{corr}}

\newcommand{\Sim}{\mathsf{Sim}}

\newcommand{\Enc}{\mathsf{Enc}}
\newcommand{\Dec}{\mathsf{Dec}}
\newcommand{\CRE}{\mathsf{CRE}}
\newcommand{\QRE}{\mathsf{QRE}}
\newcommand{\QGC}{\mathsf{QGC}}
\newcommand{\CEnc}{\mathsf{CEnc}}
\newcommand{\CDec}{\mathsf{CDec}}
\newcommand{\CSim}{\mathsf{CSim}}

\newcommand{\ggate}{\mathsf{GateEnc}}

\newcommand{\evg}{\mathsf{GateEval}}

\newcommand{\winwire}{{\mathtt{inwire}}}
\newcommand{\woutwire}{{\mathtt{outwire}}}
\newcommand{\xz}{\{x,z\}}

\newcommand{\qma}{\mathbf{QMA}}
\newcommand{\np}{\mathbf{NP}}

\begin{document}

\title{Quantum Garbled Circuits}
\author{Zvika Brakerski~\thanks{Weizmann Institute of Science, \texttt{zvika.brakerski@weizmann.ac.il}. Supported by the Binational Science Foundation (Grant No.\ 2016726), and by the European Union Horizon 2020 Research and Innovation Program via ERC Project REACT (Grant 756482) and via Project PROMETHEUS (Grant 780701).} \and
Henry Yuen~\thanks{University of Toronto. \texttt{hyuen@cs.toronto.edu}. Supported by a NSERC Discovery Grant.}}

\date{}
\maketitle

\begin{abstract}
We present a garbling scheme for quantum circuits, thus achieving a decomposable randomized encoding scheme for quantum computation. Specifically, we show how to compute an encoding of a given quantum circuit and quantum input, from which it is possible to derive the output of the computation and nothing else. 

In the classical setting, garbled circuits (and randomized encodings in general) are a versatile cryptographic tool with many applications such as secure multiparty computation, delegated computation, depth-reduction of cryptographic primitives, complexity lower-bounds, and more. However, a quantum analogue for garbling general circuits was not known prior to this work. We hope that our quantum randomized encoding scheme can similarly be useful for applications in quantum computing and cryptography.

The properties of our scheme are as follows:
\begin{itemize}
		\item Our scheme has perfect correctness, and has perfect information-theoretic security if we allow the encoding size to blow-up considerably (double-exponentially in the depth of the circuit in the worst-case). This blowup can be avoided via computational assumptions (specifically, the existence of quantum-secure pseudorandom generators). In the computational case, the size of the encoding is proportional to the size of the circuit being garbled, up to a polynomial in the security parameter.
		
	\item The encoding process is decomposable: each input qubit can be encoded independently, when given access to classical randomness and EPR pairs.
	
	\item The complexity of encoding essentially matches the size of its output and furthermore it can be computed via a constant-depth quantum circuit with bounded-arity gates as well as quantum fan-out gates (which come ``for free'' in the classical setting). Formally this is captured by the complexity class $\qnczf$.	
\end{itemize}

To illustrate the usefulness of quantum randomized encoding, we use it to design a conceptually-simple zero-knowledge (ZK) proof system for the complexity class $\qma$. Our protocol has the so-called $\Sigma$ format with a single-bit challenge, and allows the inputs to be delayed to the last round. The only previously-known ZK $\Sigma$-protocol for $\qma$ is due to Broadbent and Grilo (FOCS 2020), which does not have the aforementioned properties.

\end{abstract}
 
\newpage

\tableofcontents

\section{Introduction}

A \emph{randomized encoding} (RE) of a function $f$ is another function $\hat{f}$, computed probabilistically, such that on every input $x$, the output $f(x)$ can be recovered from $\hat{f}(x)$, and no other information about $f$ or $x$ is conveyed by $\hat{f}(x)$. A trivial example of a RE of a function $f$ is $f$ itself. Things become much more interesting when computing $\hat{f}(x)$ is simpler in some way than computing $f(x)$; for example, $\hat{f}(x)$ could be computed via a highly parallel process even if evaluating $f(x)$ itself requires a long sequential computation.

REs are central objects in cryptographic research and have proven useful in a multitude of settings: the most famous example of a RE is Yao's garbled circuits construction~\cite{Yao86}, but it was only until the work of Applebaum, Ishai and Kushilevitz in \cite{AIK04,AIK06} that the formal notion of randomized encodings was presented. Applications of RE range from secure multi-party computation, parallel cryptography, verifiable computation, software protection, functional encryption, key-dependent message security, program obfuscation and more. We refer the readers to an extensive survey by Applebaum \cite{Applebaum17} for additional details and references. Interestingly, REs have also proved useful in recent circuit lower bounds \cite{ChenR20}.

A useful feature of many randomized encodings is \emph{decomposability}: this is where a function $f$ and a sequence of inputs $(x_1,\ldots,x_n)$ can be encoded in such a way that $\hf(x_1,\ldots,x_n) = (\hf_\off,\hf_1,\ldots,\hf_n)$ where $\hf_\off$ (called the ``offline'' part of the encoding) depends only on $f$ and the randomness $r$ of the encoding, and $\hf_i$ (the ``online'' part) only depends on $x_i$ and the randomness $r$. Such randomized encodings are called \emph{decomposable}.

A good illustration of the usefulness of decomposable REs (DREs) is the task of ``private simultaneous messages'' (PSM) introduced by Feige, Kilian and Naor~\cite{FKN94}. In a PSM protocol for computing a function $f$, a set of $n$ separated players each have an input $x_i$ and send a message $m_i$ to a referee, who then computes the output value $y = f(x_1, \ldots, x_n)$. The messages $m_i$ cannot reveal any information about the $x_i$'s aside from the fact that $f(x_1,\ldots,x_n) = y$ (formally, the messages $e_i$ can be simulated given $y$). The parties share a common random string $r$ that is independent of their inputs and unknown to the referee, and the goal is to accomplish this task using minimal communication.

Using a DRE such as garbled circuits, the parties can simply send an encoding of the function $f$ and their respective inputs respectively to the referee. In particular, using the point-and-permute garbled circuits scheme of Beaver, Micali and Rogaway~\cite{BMR90,Rogawaythesis}, it is possible to construct DREs with perfect decoding correctness and perfect simulation security for any function $f$ with complexity that scales with the formula size of $f$. Assuming the adversaries are computationally bounded, it is possible to reduce this complexity to scale polynomially with the circuit size of $f$. Thus, the PSM task can be performed efficiently using DREs.

In some cases, even non-decomposable REs can be useful. However, some other non-degeneracy condition should be imposed, since (as mentioned before) every function $f$ is trivially a non-decomposable RE of itself. For example, if for some measure of complexity, the complexity of computing the encoding $\hf$ is lower than the complexity of computing $f$, then this can be leveraged for blind delegated computation: a verifier who wishes to compute $f(x)$ can first compute the encoding of a function $g(x)$ that outputs a random string $r_0$ if $f(x) = 0$, and otherwise outputs $r_1$. The server can evaluate the encoding $\hat{g}(x)$ to obtain either $r_0$ or $r_1$, which the verifier decodes to determine $f(x)$. As long as the complexity of encoding $g(x)$ is less than $f(x)$, this yields a non-trivial delegation scheme.

Given the richness and utility of randomized encodings in cryptography and theoretical computer science, it is very natural to ask whether there exists a \emph{quantum} analogue of randomized encodings. Despite its appeal, this question has remained open, and as far as we know the notion was not even formally defined in the literature before.

\subsection{Quantum Randomized Encodings} 

In this paper we introduce the notion of randomized encodings in the quantum setting, propose a construction, and analyze it. Our definition is an adaptation of the classical one: the \emph{quantum randomized encoding (QRE)} of a quantum operation $F$ (represented as a quantum circuit) and a quantum state $\qrv{x}$ is another quantum state $\hat{F}(\qrv{x})$ satisfying two properties:
\begin{enumerate}
	\item (\emph{Correctness}). The quantum state $F(\qrv{x})$ can be \emph{decoded} from $\hat{F}(\qrv{x})$. 
	\item (\emph{Privacy}). The encoding $\hat{F}(\qrv{x})$ reveals no information about $F$ or $\qrv{x}$ apart from the output $F(\qrv{x})$. 
\end{enumerate}
The privacy property is formalized by saying there is a simulator that, given $F(\qrv{x})$, can compute the encoding $\hat{F}(\qrv{x})$. We also refer to $\hat{F}$ as the encoding of $F$. 

Furthermore, we also define what it means for a QRE to be \emph{decomposable}: the encoding $\hat{F}(\qrv{x})$ can be computed in a way that each qubit of the input $\qrv{x}$ is encoded independently, and the encoding takes in as input $\qrv{x}$, a classical random string $r$, and a sequence of EPR pairs $\qrv{e}$. \footnote{We recall that an EPR pair is the maximally entangled state $\frac{1}{\sqrt{2}} \Big( \ket{00} + \ket{11} \Big)$, the quantum analogue of a pair of classically correlated bits. } (See \Cref{sec:qredef} for a formal definition of (decomposable) QRE.)

For comparison of the notion of QRE with other cryptographic notions such as MPC, FHE and program obfuscation see Appendix~\ref{sec:comparisonwithnotions}.

We then present a construction of a decomposable QRE, which we call the \emph{Quantum Garbled Circuits} scheme:
\begin{theorem}[Main result, informal]
\label{thm:main-informal}
Suppose $\CRE$ is a classical DRE scheme with perfect correctness, information-theoretic (resp. computational) privacy, and polynomial time decoding. 
Then there exists a decomposable QRE scheme $\QGC$ with the following properties: 
\begin{enumerate}
	\item $\QGC$ has perfect correctness and polynomial-time decoding.
	\item $\QGC$ uses $\CRE$ as a black box, and has information-theoretic (resp. computational) privacy.
	\item If the encoding procedure of $\CRE$ can be computed in $\ncz$, then the encoding procedure of $\QGC$ can be computed in $\qnczf$ (i.e. the class of constant-depth quantum circuits with unbounded fan-out gates). 
\end{enumerate}
\end{theorem}

(See \Cref{sec:qreresult} for a formal statement of our main result). We elaborate on the properties of the QRE scheme below. It assumes the existence of a classical DRE scheme $\CRE$ with specific correctness, privacy, and complexity properties; examples of such schemes can be found in \cite{BMR90,Rogawaythesis} (also see the survey in~\cite{Applebaum17}). In the case of computational privacy, we assume the existence of quantum-secure one-way functions. 

\paragraph{Correctness.} The correctness property asserts that from an encoding $\hat{F}(\qrv{x})$, it is possible to decode the output state $F(\qrv{x})$ with probability $1$. This is inherited from the perfect correctness property of $\CRE$. Furthermore, the decoding procedure preserves quantum correlations with side information: if $(\qrv{x},\qrv{y})$ denotes the joint state of the input $\qrv{x}$ and some auxiliary quantum state $\qrv{y}$ which may be \emph{entangled} with $\qrv{x}$, then the joint state of the output and side information after decoding is $(F(\qrv{x}),\qrv{y})$. 

The decoding procedure also takes polynomial time in the size of encoding. This inherits the polynomial-time decoding complexity of $\CRE$.

\paragraph{Privacy.} 
The privacy property implies that there exists a quantum algorithm $\Sim$ such that $\Sim(F(\qrv{x}))$ is indistinguishable from the encoding $\hat{F}(\qrv{x})$, and furthermore $\Sim$ runs in time polynomial in the size of the encoding $\hat{F}(\qrv{x})$. The quantum scheme inherits its privacy from the classical scheme in a black box way: if $\CRE$ is secure against all (quantum) distinguishers of size $S$, then $\QGC$ is secure against distinguishers of size at most $S - \Lambda$ where $\Lambda$ is the complexity of decoding $\hat{F}(\qrv{x})$ (which is polynomial in the size of $\hat{F}(\qrv{x})$). Note that perfect information-theoretic privacy corresponds to privacy against distinguishers of all sizes. 

The privacy property also holds even when considering entangled side information: the joint state $(\Sim(F(\qrv{x})),\qrv{y})$ (which is computed by applying $F$, then $\Sim$ to the input $\qrv{x}$) is indistinguishable from $(\hat{F}(\qrv{x}),\qrv{y})$.

\paragraph{Size of the Encoding.} 
The \emph{size} of the encoding of $\QGC$ is the number of qubits in  $\hat{F}(\qrv{x})$, which depends on the circuit size and depth of $F$, and also on the size of the classical encodings computed in the scheme $\CRE$. For example, the classical DRE schemes from \cite{BMR90,Rogawaythesis} with computational privacy have encoding size that scales polynomially with the circuit size of $f$. This translates to the size of $\hat{F}$ being polynomial in the circuit size of $F$. On the other hand, the known classical decomposable RE schemes with information-theoretic privacy all have encodings $\hat{f}$ that grow \emph{exponentially} with the circuit \emph{depth} of the function $f$. Using such a scheme as our $\CRE$, the size of the corresponding quantum encoding $\hat{F}$ in our construction might grow even doubly-exponentially with the circuit depth of the quantum operation $F$. Note that this is the worst-case, the exact growth depends on the composition of gates in the circuit (see technical overview).\footnote{\label{fn:bug}We note that an earlier version of this work claimed that the growth in the information-theoretic setting is only (single) exponential. However we discovered an error in our original proof, and the correct analysis turns out to imply the aforementioned parameters. We discuss the causes for this in our technical overview below.}
However we note that even in this case, the number of EPR pairs used in the encoding $\hat{F}(\qrv{x})$ remains linear in the circuit size of $F$.

\paragraph{Complexity and Locality of Encoding.} 
The decomposability of the encoding $\hat{F}(\qrv{x})$ is analogous to the decomposability property of classical DRE schemes, where the encoding can be expressed as the following concatenation:
\[
	\hat{F}(\qrv{x}; r, \qrv{e}) = (\hat{F}_\off(r, \qrv{e}), \hat{F}_1(\qrv{x}_1; r, \qrv{e}), \ldots, \hat{F}_n(\qrv{x}_n; r, \qrv{e}))
\]
where we indicate the dependency of the encoding on randomness string $r$ and a sequence of EPR pairs $\qrv{e}$. The state $\hat{F}_\off(r, \qrv{e})$ is called the ``offline'' part of the encoding that depends on $F$ but not $\qrv{x}$, and $\{ \hat{F}_j(\qrv{x}_j; r, \qrv{e}) \}_j$ forms the ``online'' part of the encoding where $\qrv{x}_j$ is the $j$-th qubit of the $n$-qubit state $\qrv{x}$. 

The encoding procedure of $\QGC$ is highly parallelizable. Suppose that the encoding procedure of $\CRE$ is computable in $\ncz$ (which is the case for the randomized encoding schemes of \cite{BMR90,Rogawaythesis}). Then the offline part $\hat{F}_\off(r;\qrv{e})$ can be computed by a $\qnczf$ circuit acting on $(r,\qrv{e})$; a $\qnczf$ circuit is a constant-depth circuit composed of single- and two-qubit gates, as well as \emph{fan-out gates} with unbounded arity, which implements the unitary $\ket{x,y_1,\ldots,y_n} \to \ket{x,x \oplus y_1,\ldots,x \oplus y_n}$. 
Similarly, $\hat{F}_j(\qrv{x}_j; r, \qrv{e})$ can be computed by a $\qnczf$ circuit acting on $r$, a constant number of qubits of $\qrv{e}$, and the qubit $\qrv{x}_j$ (but does not depend on the gates of $F$). 

We note that in the quantum setting it is not so clear what is the ``correct'' analogue for the complexity class $\ncz$, which is the class of functions that can be computed in constant-depth, or equivalently, functions where each output bit depends only on a constant number of input bits. This implicitly assumes that input bits can be replicated an arbitrary number of times (for example, all of the output bits may depend on one input bit). However, due to the No-Cloning Theorem we cannot assume that input qubits can be copied, and thus it seems reasonable to consider constant-depth quantum circuits augmented with fan-out gates. However, the fan-out gate does appear to yield unexpected power in the quantum setting: for example, the parity gate can be computed in $\qnczf$, while classically it is even outside $\acz$ (see e.g.~\cite{M99,HS05}). Nevertheless, $\qnczf$ circuits appear to be weaker than general polynomial-size quantum circuits and it may be reasonable to assume that it will be possible to implement the fan-out gate in ``constant depth'' in some quantum computing architectures (see \cite{HS05} for discussion).

\paragraph{Classical Encoding for Classical Inputs.} A desirable property that comes up in the quantum setting is to allow some of the parties to remain classical, even when performing a quantum task. In the RE setting, we would like to allow parties with a classical inputs to compute their encoding in a classical manner (and in particular with access only to the classical part of the randomness/EPR string). Our scheme indeed allows this type of functionality, and therefore allows applications such as quantum PSM (as discussed above) even when some of the parties are classical. Nevertheless, the encoding (and in particular the offline part that depends on the circuit) requires quantum computation.

Could the encoding of a quantum circuit and classical input be made entirely classical? As we will discuss later, this could be used to achieve general indistinguishability obfuscation for quantum circuits. However, there are certain complexity-theoretic constraints on this possibility: Applebaum showed that any language decidable by circuits that admit efficient RE with information-theoretic security falls into the class $\szk \subseteq \ph$~\cite{Applebaum14szk}. Therefore, if we could achieve QRE with statistical security for polynomial size quantum circuits, this would imply $\bqp \subseteq \szk$. On the other hand, the oracle separation between $\bqp$ and $\ph$ by Raz and Tal~\cite{RT19} suggests that this inclusion is unlikely.\footnote{We thank Vinod Vaikuntanathan for pointing this out to us.} As presented, our Quantum Garbled Circuits scheme achieves statistical security for all quantum circuits of depth $O(\log \log n)$, but with small modifications can handle an interesting subclass of quantum circuits of depth $O(\log n)$ that is not obviously classically simulable, and thus languages computed by this subclass are not obviously contained in $\szk$.\footnote{This subclass includes circuits such that the first $O(\log \log n)$ layers can have arbitrary $2$-qubit gates, and the remaining layers are all Clifford gates.} 
This suggests that obtaining entirely classical encoding of quantum circuits cannot be achieved with statistical security; whether it can be achieved with computational security remains an intriguing open problem.

\subsection{Other Related Work} 
We mention some related work on adapting the notion of randomized encodings/garbled circuits to the quantum setting. In \cite{KashefiW17}, Kashefi and Wallden present an interactive, multi-round protocol for verifiable, blind quantum computing, that is inspired by Yao's garbled circuits. The motivation for their protocol comes from wanting a protocol where a weak quantum client delegates a quantum computation to a powerful quantum server, while still maintaining verifiability.

In a recent paper \cite{Zhang20Succinct} (which builds on prior work \cite{Zhang19}), Zhang presents a blind delegated quantum computation protocol that is (partially) ``succinct'': it is an interactive protocol with an initial quantum phase whose complexity is independent of the computation being delegated, and the second phase is completely classical (with communication and round complexity that depends on the size of the computation). The security of the protocol is proved in the random oracle model. The construction and analysis appear to use ideas from classical garbled circuits. 

Both the work of \cite{KashefiW17} and \cite{Zhang20Succinct} focus on protocols for delegated quantum computation, and both protocols involve a large number of rounds of interaction that grow with the size of the computation being delegated. In contrast, the focus of our work is on studying the notion of quantum randomized encodings (in which the number of rounds of interaction is constant). 

Finally, we mention that while our notion of quantum randomized encoding has many similarities with other commonly studied cryptographic notions such as secure multiparty computation (MPC) and homomorphic encryption, QRE is a distinct notion with different goals. We provide a more detailed comparison in \Cref{sec:comparisonwithnotions}.

\subsection{Application: A New Zero-Knowlege $\Sigma$-Protocol for $\qma$}

To highlight the usefulness of the notion of QRE, we present an application to designing zero-knowledge (ZK) protocols for the complexity class $\qma$. Specifically we show how to easily obtain $3$-round ``sigma'' (abbreviated by $\Sigma$) protocols for $\qma$ using QRE as a black box, and in fact our construction achieves features that were not known before in the literature. We elaborate more below. 

Zero-knowledge proofs~\cite{goldwasser1989knowledge} is one of the most basic and useful notions in cryptography. Essentially, it is an interactive proof system where the verifier is guaranteed to learn nothing beyond the validity of the statement being proven. This is formalized by showing that for any accepting instance and any (possibly malicious) verifier, there exists a simulator which can generate a view which is indistinguishable from the actual view of the verifier in an interaction with an honest prover. The notion of indistinguishability depends on whether we are considering computational or statistical zero-knowledge.

In the classical setting, the canonical ZK protocol for $\np$ was presented by Goldreich, Micali, and Wigderson~\cite{goldreich1991proofs}. The protocol has a simple $3$-message structure (known as the $\Sigma$ format): the prover sends a message, the verifier sends a uniformly random challenge, the prover responds, and the verifier decides to accept or reject based on the transcript of the communication. Aside from their simplicity, $\Sigma$-protocols are also desirable because it one can then use the Fiat-Shamir heuristic~\cite{fiat1986prove} to make the protocol non-interactive, for example. 
In some cases, it is useful to have a $\Sigma$-protocol with a single challenge bit. In the classical case this implies a notion known as ``special soundness'' which is useful, for example, for constructing non-interactive zero-knowledge (NIZK) protocols~\cite{BFM88,FLS90}. Another useful feature is ``delayed input'', where the prover can produce the first message without any knowledge of the instance or the witness. 
This is useful for confirming well-formedness of the execution of a protocol while minimizing the number of extra rounds of communication. In the classical setting Blum's Graph-Hamiltonicity protocol \cite{Blum86,FLS90} has these properties and is thus often used.

Zero-knowledge proof systems for $\qma$, the quantum analogue of $\np$, have only been studied fairly recently, and known results are still few~\cite{broadbent2016zero,vidick2020classical,coladangelo2020non,broadbent2019zero,bitansky2020post}. Recently, Broadbent and Grilo~\cite{broadbent2019zero} presented the first ZK $\Sigma$-protocol for $\qma$, achieving constant soundness error.  
Their protocol relies on a reduction to a special variant of the local Hamiltonians problem. It requires multi-bit challenges and does not seem to support delayed inputs.

In this work, we show a simple approach for obtaining a $\Sigma$-protocol with a single-bit challenge and delayed-input functionality, using quantum randomized encodings. Like~\cite{broadbent2019zero}, our protocol also has constant soundness error. In contrast to \cite{broadbent2019zero}, we do not require a reduction to a specific $\qma$-complete problem. Our approach is similar to constructions of ZK protocols from randomized encoding in the classical setting \cite{HazayV16}.

Conceptually, the protocol is simple. Recall that a $\qma$ problem $L$ is defined by a (quantum polynomial time) verifier circuit $V$, which takes a classical instance $x$ and a quantum witness $\qrv{w}$ and decides (with all but negligible probability) whether $x$ is a yes or no instance. We generically create a zero-knowledge protocol, where the basic idea is as follows. The prover creates a QRE of $V$, and sends it to the verifier, together with commitments to the labels and the randomness used to generate the QRE. The verifier sends a challenge bit $b$. Now for $b=0$ simply all the commitments are opened and the verifier checks that indeed the proper circuit was encoded, if $b=1$ then only the labels corresponding to the actual $x,\qrv{w}$, and the verifier can thus check the value of $V$ on them. The actual protocol is slightly more complicated since the QRE only has ``labels'' for classical inputs, and $\qrv{w}$ needs to be hidden even given the labels. Therefore $\qrv{w}$ is treated slightly differently than described above (essentially ``teleported'' into the circuit).

\subsection{Future Directions and Open Problems}

We end this section with several examples of future directions and open problems.

\begin{enumerate}
	\item \textbf{Applications of QRE}. We presented one application in the form of a simple zero-knowledge protocol for $\qma$. Given the variety of applications of RE in classical cryptography, we anticipate that there is similarly many analogous applications in the quantum setting. We elaborate on several potential applications in \Cref{sec:otherapps}.
	
	\item \textbf{Obtain statistically-private QRE for all log-depth circuits}. Our information-theoretic QRE has overhead that is \emph{doubly-exponential} in the depth of the circuit being encoded (although as mentioned it should be possible to encode certain classes of circuits with ``only'' an exponential overhead). Thus there is a gap between what is achievable with classical RE (where it is possible to encode all log-depth circuits with statistical privacy). Can information-theoretic QRE be achieved for all log-depth circuits, or is the gap inherent? We note that it is not known whether statistically secure RE can be performed for all polynomial-size \emph{classical} circuits.
	
	\item \textbf{Completely classical encoding for quantum circuits}. Can the encoding of a quantum circuit be made completely classical? This would be very useful for obtaining obfuscation for quantum circuits, assuming classical obfuscation. 
\end{enumerate}

\subsection{Paper Organization}

Section~\ref{sec:overview} contains a technical overview of our contribution. Section~\ref{sec:prelims} contains notation and preliminaries about quantum computation and classical randomized encoding. In Section~\ref{sec:qredefresult} we define the notion of quantum randomized encoding, state some of its basic properties and state our main result. The details of our new zero-knowledge $\Sigma$-protocol appear in Section~\ref{sec:zk}. Section~\ref{sec:construction} contains our Quantum Garbled Circuits construction and Section~\ref{sec:analysis} contains proofs of correctness and privacy. We note that there is no dependence at all (or vice versa) between the last two sections and Section~\ref{sec:zk}, and the order of reading them should be up to the reader's preference.

\subsection*{Acknowledgments} We thank Sanjam Garg and Vinod Vaikuntanathan for insightful discussions. We thank Chinmay Nirkhe for lengthy discussions about quantum garbled circuits. We thank Nir Bitansky and Omri Shmueli for discussions on quantum zero-knowledge. We thank anonymous conference reviewers for their helpful feedback. We thank the Simons Institute for the Theory of Computing -- much of this work was performed while the authors were visiting the Institute as a part of the Summer Cluster on Quantum Computing (2018) and the Semester Programs on Quantum Computing and Lattice Cryptography (2020).

\section{Overview of Our Construction}
\label{sec:overview}

We provide an overview of our techniques, we refer to the technical sections for the formal presentation and proofs. In what follows, we use bolded variables such as $\qrv{q}$ to denote density matrices, and for a unitary $U$ we write $U(\qrv{q})$ to denote the state $U \qrv{q} U^\dagger$ (see \Cref{sec:prelims} for more details about notation).

\subsection{Our Approach: Quantum Computation via Teleportation} The basis of our approach to quantum RE is \emph{computation by teleportation}, an idea that is common to many prior results on protocols for delegated quantum computation and computing on encrypted data~\cite{broadbent2009universal,BJ15,DSS16}. We briefly review this concept.

Recall that quantum teleportation allows one party to transmit a qubit $\qrv{q}$ to another party using only classical communication and a preshared EPR pair $\qrv{e} = (\qrv{e}_\win,\qrv{e}_\wout)$. Specifically, the sender performs a measurement on the qubit $\qrv{q}$ and $\qrv{e}_\win$ to obtain two (uniformly distributed) classical bits $a,b$ (often called ``teleportation keys''). The receiver's qubit $\qrv{e}_\wout$ collapses to $X^a Z^b(\qrv{q})$ where $X$, $Z$ are the bit-flip and phase-flip Pauli matrices respectively. Using the teleportation keys $(a,b)$ the original qubit $\qrv{q}$ can be recovered. 

Teleportation can be used to apply gates: let $G$ be a single-qubit unitary (the generalization to multi-qubit unitaries is straightforward), and suppose that the sender and receiver share the state $(\qrv{e}_\win,G(\qrv{e}_\wout))$ instead, in which $G$ is applied to the second half of an EPR pair. When the sender teleports the qubit $\qrv{q}$ and obtains teleportation keys $(a,b)$, the resulting state on the receiver's side is $G ( X^a Z^b(\qrv{q}))$. If $G$ is a Clifford gate (i.e. a unitary that normalizes the Pauli group), then this is equal to $X^{a'} Z^{b'} (G(\qrv{q}))$ for some updated keys $(a',b')$ that are a deterministic function of $(a,b)$ and $G$. We call $X^{a'} Z^{b'}$ the \emph{Pauli error} on the state.

This already suggests a method of quantum randomized encoding for the class of Clifford circuits. Let $C$ be a circuit consisting of gates $G_1,\ldots,G_m$. The encoding of the circuit $C$ and an $n$-qubit quantum input $\qrv{x}$ can be computed in the following way:
\begin{enumerate}
	\item Generate EPR pairs $\qrv{e}^w = (\qrv{e}^w_\win,\qrv{e}^w_\wout)$ for each wire $w$ of the circuit $C$.\footnote{One can think of a \emph{wire} as a line segment in a circuit diagram in between the gates, as well as the segments for the inputs/output qubits.}
	\item For each gate $G_i$, if the input wires as specified by circuit $C$ are $v_1,v_2$ (if $G_i$ is a two-qubit gate, for example), then apply $G_i$ to the ``second halves'' $(\qrv{e}^{v_1}_\wout,\qrv{e}^{v_2}_\wout)$ of the corresponding EPR pairs. Note that after this operation the qubits $(\qrv{e}^{v_1}_\wout,\qrv{e}^{v_2}_\wout)$ now store the output of $G_i$.
	\item If wire $v$ is connected to wire $w$ via some gate, perform the teleportation measurement on the qubits $(\qrv{e}^{v}_\wout,\qrv{e}^{w}_\win)$ to obtain classical teleportation keys $(a_{vw},b_{vw})$.
	\item If wire $w$ is the $i$-th input wire to circuit $C$, then perform the teleportation measurement on qubits $(\qrv{x}_i,\qrv{e}^w_\win)$ to obtain classical teleportation keys $(a_i,b_i)$.
	\item Compute from all the intermediate teleportation keys $(a_i,b_i)_i $ and $(a_{vw},b_{vw})_{v,w}$ the final teleportation keys $(a_j',b_j')$ corresponding to the $j$-th output qubit, for each $j$. The final teleportation keys are a deterministic function $f_\corr$ of all the intermediate teleportation keys, as well as the gates $G_1,\ldots,G_m$. 
\end{enumerate}
Each of the teleportation operations will yield uniformly random teleportation keys for each pair of connected wires, inducing Pauli errors that accumulate as the teleported state ``moves'' through the circuit. Since all gates are Clifford, the Pauli errors get adjusted in a deterministic way, and the resulting state in the qubits $(\qrv{e}^w_\wout)$ for output wires $w$ will be $X^{a_1'} Z^{b_1'} \otimes \cdots \otimes X^{a_n'} Z^{b_n'} (C(\qrv{x}))$. This output state, along with the final teleportation keys $(a_j',b_j')_j$, yields a QRE of circuit $C$ and input $\qrv{x}$, because the final state $C(x)$ can be recovered from this, and it yields no information about the gates or the original input $\qrv{x}$ (as long as the intermediate teleportation keys are not revealed). The \emph{quantum} complexity of this encoding is quite low: preparing the EPR pairs, applying the gates, and applying the teleportation measurements can be parallelized and thus performed in constant-depth. However the \emph{classical} complexity of this encoding is dominated by the complexity of computing the final teleportation keys $(a_j',b_j')_j$, which takes time that is linear in the size of the circuit $C$. 

This complexity issue can be solved by leveraging \emph{classical randomized encodings} (CREs): the encoder, instead of computing $(a_j',b_j')_j$ itself, computes a randomized encoding $\hat{f}_\corr(\vec{k})$ of the function $f_\corr$ and intermediate teleportation keys $\vec{k}$. Using a decomposable RE scheme such as (classical) garbled circuits, it is possible to compute $\hat{f}_\corr$ using a constant-depth circuit; this encoding corresponds of an offline part $\hf_{\corr,\off}$ and an online part that consists of \emph{labels} for each bit of the teleportation keys $\vec{k}$. Thus the overall quantum encoding of $C(\qrv{x})$ will be the quantum state $X^{a_1'} Z^{b_1'} \otimes \cdots \otimes X^{a_n'} Z^{b_n'} (C(\qrv{x}))$, along with the CRE $\hat{f}_\corr(\vec{k})$. The decoder can compute from the CRE the final teleportation keys, and then recover $C(\qrv{x})$.

We see that this yields a simple QRE for Clifford circuits. If the CRE used is decomposable, the QRE is decomposable as well: observe that the input qubits $\qrv{x}$ are encoded separately from the encoding of the circuit (the input teleportation measurements and the computation of the input labels for $\hat{f}_\corr$ can be done independently). Furthermore, the QRE has information-theoretic (resp. computational) privacy if the CRE has information-theoretic privacy (resp. computational).

\subsection{The Challenge: Going Beyond Clifford Gates}

The real challenge comes from dealing with the case of non-Clifford gates in the circuit (such as the $T = \begin{pmatrix} 1 & 0 \\ 0 & e^{i \pi/4} \end{pmatrix}$ gate).\footnote{Recall that when augmenting the Clifford group with any non-Clifford such as the $T$ gate, the resulting set of gates is universal for quantum computation.} The QRE described above does not work when one of the gates is a $T$ gate; this is because the gate teleportation protocol induces a non-Pauli error: $T( X^a Z^b (\qrv{q})) = X^{a'} Z^{b'} P^{a'} (T(\qrv{q}))$, where $P = \begin{pmatrix} 1 & 0 \\ 0 & i \end{pmatrix}$ is the phase gate (a Clifford gate). Thus we no longer have the invariant that the intermediate states of the teleportations are the masked by Pauli errors. This is problematic because the errors that are induced via the sequence of teleportations will no longer be simple to compute classically. 

Instead, a phase error induced by a gate teleportation should be removed before the next teleportation. However, there appears to be a catch-22: in order to know whether there is a phase error, a teleportation measurement needs to be performed to get the keys $(a,b)$. On the other hand, the teleportation can only be performed if one was certain that there was no phase error from previous teleportations! If the encoder wants to avoid performing a sequential computation, it appears that both the teleportation measurements and the corresponding Pauli/phase error corrections have to be performed by the evaluator -- and this must be done in a way that does not violate the privacy of the encoding. 

We now describe the key ideas used in our Quantum Garbled Circuits scheme to handle these issues.

\paragraph{Encrypted Teleportation Gadgets.} First, to allow the teleportation measurements to be performed by the evaluator in a manner that maintains the privacy of the encoding, the encoder will apply \emph{encrypted teleportation gadgets} between the connected EPR pairs. For simplicitly assume that $G$ is a single-qubit gate in the circuit that connects wire $v$ to wire $w$. The encoder applies gate $G$ to the EPR qubit $\qrv{e}^v_\wout$ as in the Clifford encoding, but instead of performing the teleportation measurements on the pair $(\qrv{e}^v_\wout,\qrv{e}^w_\win)$, the encoder applies the following circuit to the two qubits as well as additional ancilla:\footnote{If $G$ is a multiqubit gate, then the encoder will apply $G$ to qubits $(\qrv{e}^{v_1}_\wout,\ldots,\qrv{e}^{v_p}_\wout)$ before applying the teleportation gadget to each $(\qrv{e}^{v_i}_\wout,\qrv{e}^{w_i}_\win)$.}
\begin{figure}[H]
  \centering
  $
        \Qcircuit @C=0.6em @R=0.8em {
          \qrv{e}^v_\wout  & & &  & \qw & \ctrl{3} & \gate{H}  & \qw & \ctrl{1} & \gate{X^{s_x}} & \gate{Z^{s_z}} & \qw \\
          \ket{0 \cdots 0} & & & & {/} \qw  & \qw & \qw & \gate{X(\ell_{z,0})} & \gate{X(\ell_{z,0} \oplus \ell_{z,1})} & \qw & \qw & \qw  \\
          \ket{0 \cdots 0}  & & & & {/} \qw & \qw & \qw     & \gate{X(\ell_{x,0})} & \gate{X(\ell_{x,0} \oplus \ell_{x,1})} & \qw & \qw & \qw  \\
          \qrv{e}^w_\win & & & & \qw & \targ & \qw  & \qw & \ctrl{-1} & \gate{X^{t_x}} & \gate{Z^{t_z}} & \qw \\          
        }
        $
    \caption{Encrypted teleportation gadget.}
    \label{fig:teleportation-gadget-intro}
\end{figure}
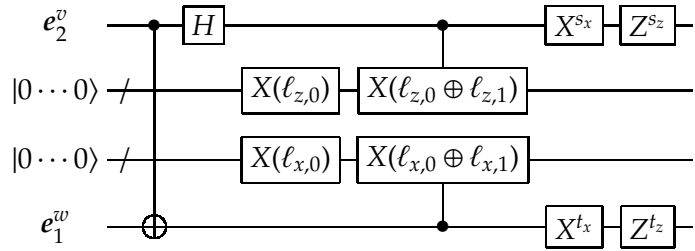
Here, $(\ell_{z,0},\ell_{z,1},\ell_{x,0},\ell_{x,1})$ are random strings in $\{0,1\}^\kappa$ (we will say in a moment where they come from), and $s_x,s_z,t_x,t_z$ are uniformly random bits. For a string $r \in \{0,1\}^\kappa$, the gate $X(r)$ denotes the $\kappa$-qubit gate that applies a bit-flip gate $X$ to qubit $i$ if $r_i = 1$. 

The functionality of this circuit is to perform the teleportation measurement on $(\qrv{e}^v_\wout,\qrv{e}^w_\win)$, but instead of obtaining teleportation keys $(a,b) \in \{0,1\}^2$ as usual, the middle two wire-bundles yield \emph{teleportation labels} $(\ell_{z,a}, \ell_{x,b})$ which indicate the Pauli error $X^b Z^a$ on the teleported qubit.\footnote{The ``label'' terminology is inspired by how random labels are used to encrypt the contents of each wire in classical garbled circuits.} As long as the randomness used to choose $(\ell_{z,0},\ell_{z,1},\ell_{x,0},\ell_{x,1})$ and the bits $s_x,s_z,t_x,t_z$ are kept secret from the evaluator, then this circuit completely hides all information about the teleportation keys $(a,b)$, and thus the teleported qubit is completely randomized (from the evaluator's point of view).

\paragraph{Switching The Order of Correction and Teleportation.} There still is the issue that before applying the encrypted teleportation gadget to a specific pair $(\qrv{e}^v_\wout,\qrv{e}^w_\win)$, the input qubit may have a Pauli/phase error that needs correcting. For example, imagine that the qubit $\qrv{e}^v_\wout$, originally the second half of an EPR pair, now stores a qubit $X^a Z^b(\qrv{q})$ due to a previous teleportation. When we apply the gate $G$ (which we again assume is a single-qubit gate for simplicity) to $\qrv{e}^v_\wout$, we obtain the following correction: $G(X^a Z^b(\qrv{q})) = R (G (\qrv{q}))$ where $R$ is a Pauli or phase correction. We need to avoid having the correction $R$ before applying the teleportation gadget.

We show that we can essentially \emph{switch} the order of the two operations: instead of performing the Pauli/phase correction and then applying the encrypted teleportation gadget, we show that a circuit similar to the teleportation gadget can be applied \emph{first} by the encoder (obliviously to whether the correction $R$ was needed or not), and \emph{then} the consequences of this bold move can be dealt with by the encoder. More precisely, if we let $\sf{TP}_{\ell,s,t}$ denote the encrypted teleportation gadget, then there exist unitaries $\Lambda_1,\Lambda_2,\Lambda_3$ such that the following circuit identity holds:
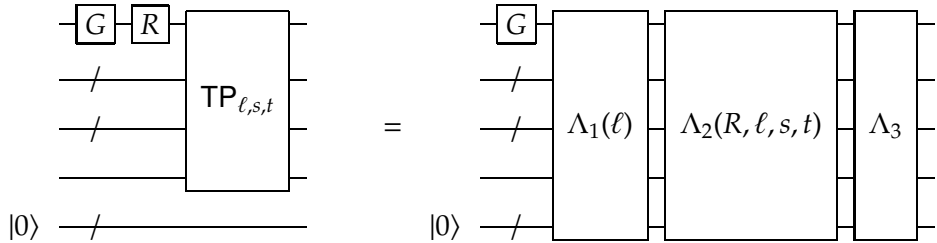
\begin{figure}[H]
  \centering
  $
        \Qcircuit @C=0.6em @R=0.8em {
            & & & \gate{G} & \gate{R} & \multigate{3}{\sf{TP}_{\ell,s,t}} & \qw &  	& && &	 & & & &  & & \gate{G} & \multigate{4}{\Lambda_1(\ell)}  & \multigate{4}{\Lambda_2(R,\ell,s,t)} &\multigate{4}{\Lambda_3} & \qw \\
           &  & & {/} \qw & \qw & \ghost{\sf{TP}_{\ell,s,t}} & \qw  			   & && 	& &		& &  & &   & & {/}\qw  & \ghost{\Lambda_1(\ell)} 		& \ghost{\Lambda_2(R,\ell,s,t)} & \ghost{\Lambda_3} & \qw \\
            &  & & {/} \qw & \qw & \ghost{\sf{TP}_{\ell,s,t}}  & \qw &  &&  & & = & & & & & & {/} \qw & \ghost{\Lambda_1(\ell)}	& \ghost{\Lambda_2(R,\ell,s,t)} &   \ghost{\Lambda_3} & \qw \\
           & &  & \qw &  \qw & \ghost{\sf{TP}_{\ell,s,t}} & \qw  & &&& &	&&  & & & & \qw & \ghost{\Lambda_1(\ell)} & \ghost{\Lambda_2(R,\ell,s,t)} & \ghost{\Lambda_3} & \qw & \\
		  \ket{0} & &  & {/} \qw & \qw & \qw & \qw  & 		& &&&		&&  & \ket{0} &  & & {/} \qw & \ghost{\Lambda_1(\ell)} & \ghost{\Lambda_2(R,\ell,s,t)} & \ghost{\Lambda_3} & \qw & 
        }
$
    \caption{Switching the order of correction and teleportation.}
    \label{fig:switch-order-intro}
\end{figure}
Here, $\Lambda_1$ is a circuit that only depends on the labels $\ell$. The circuit $\Lambda_2$ depends on all parameters $R, \ell, s,t$. The circuit $\Lambda_3$ does not depend on any parameters. The circuits on the right-hand side use additional ancilla. 

Thus, the encoder can first apply the gate $G$ and then the circuit $\Lambda_1(\ell)$ to the connecting EPR pairs and ancilla qubits, because it knows the random labels $\ell$. The idea is then to offload the task of applying the remaining circuits $\Lambda_2(R,\ell,s,t)$ and $\Lambda_3$ -- which we call \emph{correction gadgets} -- to the evaluator in a way that hides the values of $R, \ell, s,t$ (and thus maintains the privacy of the encoding). Since $\Lambda_3$ doesn't depend on any parameters, the decoder can always automatically apply this. For $\Lambda_2$, the encoder computes a quantum randomized encoding of $\Lambda_2(R,\ell,s,t)$ --- but this time we have an easier task because $\Lambda_2(R,\ell,s,t)$ is a Clifford circuit (in fact, it comes from a special subset of depth-$1$ Clifford circuits). In principle one could then recursively compute the QRE of $\Lambda_2$ using the Clifford scheme described above, but for our Quantum Garbled Circuits scheme we employ a different method called \emph{group-randomizing QRE}.

\paragraph{Group-Randomizing QRE.} In this section we explain group-randomizing QRE in greater generality than is needed for our Quantum Garbled Circuits scheme, because we believe that it is a useful conceptual framework for thinking about randomizing computations and may be useful for other applications of QRE. 

Let $\cal{G}$ denote a subgroup of the $n$-qubit unitary group. Let $C$ denote a circuit computing a unitary in $\cal{G}$, and let $\qrv{x}$ denote an input. The encoding of $C(\qrv{x})$ is a pair $\hat{C}(\qrv{x}) = (R(\qrv{x}), CR^\dagger)$, where $R$ is a uniformly random\footnote{For the purpose of this discussion we can always assume that $\cal{G}$ is finite.} element of $\cal{G}$, and we assume that $CR^\dagger$ is described using some canonical circuit representation for elements of $\cal{G}$. Clearly, the output $C(\qrv{x})$ can be decoded from this encoding: the decoder simply applies $CR^\dagger$ to $R(\qrv{x})$. This encoding is also perfectly private: for a uniformly random $R \in \cal{G}$, the unitary $CR^\dagger$ is also uniformly distributed in $\cal{G}$. Thus, a simulator on input $\qrv{y}$ can output $(E^\dagger(\qrv{y}), E)$ for a uniformly random element $E \in \cal{G}$. When $\qrv{y} = C(\qrv{x})$ this yields the same distribution as the encoding $\hat{C}(\qrv{x})$. We note that this method is conceptually similar to the well-known classical RE techniques for branching programs using matrix randomization \cite{Kilian88}.

For general unitaries, group-randomizing QRE is exorbitantly inefficient, because a general unitary (even after discretization) requires exponentially many random bits to describe. However when the group $\cal{G}$ is the Clifford group, for example, then this encoding is efficient: the $n$-qubit Clifford group is a finite group of order $2^{\Theta(n^2)}$, and a uniformly random element can be sampled in $\poly(n)$ time. Furthermore, there are efficient algorithms to compute canonical representations of Clifford circuits~\cite{G98,AG04}. In particular, an $n$-qubit Clifford unitary can always be written as a Clifford circuit with $O(n^2)$ gates, and this is tight.

We now apply this group-randomizing framework to encoding the correction circuits $\Lambda_2$ described above. Let $\cal{G}$ denote a unitary subgroup that contains $\Lambda_2$; this will be a group of a special class of depth-$1$ Clifford circuits. The circuits $\Lambda_2$ depends on parameters $R, \ell, s, t$. The parameters $(\ell,s,t)$ are randomness that is known to the encoder, so this can be be fixed. However, as discussed the correction unitary $R$ is not known ahead of time to the encoder -- it depends on the teleportation keys $(a,b)$ of a \emph{previous} gate teleportation, as well as the gate $G$ that was involved. In fact $R$ is a deterministic function of $(a,b)$, and therefore the description of $\Lambda_2(R,\ell,s,t)$ is a deterministic function of $(a,b)$. 

To encode the circuit $\Lambda_2(R,\ell,s,t)$ with quantum input $\qrv{x}$, the encoder first samples a uniformly random element $A \in \cal{R}$. Then it applies the unitary $A$ to the input $\qrv{x}$. This constitutes the quantum part of the encoding. 

For the classical part: let $f_{\corr,A,\ell,s,t}(a,b)$ denote the function that computes a canonical description of the unitary $\Lambda_2(R,\ell,s,t) \cdot A^\dagger \in \cal{G}$. Since $\Lambda_2$ is a Clifford unitary on $O(\kappa^2)$ qubits (the ancillas at the bottom of \Cref{fig:switch-order-intro} consist of $O(\kappa^2)$ qubits), the function $f_{\corr,A,\ell,s,t}$ can be computed by a $\poly(\kappa)$-sized circuit. The encoder can then efficiently compute a classical decomposable RE $\hat{f}_{\corr,A,\ell,s,t}$ of $f_{\corr,A,\ell,s,t}$. Since the input $(a,b)$ is not known yet, the encoding consists of an offline part $\hf_\off$ and online labels $(\ell_{z,0}',\ell_{z,1}',\ell_{x,0}',\ell_{x,1}')$ for the $4$ possible values of $(a,b)$. The classical part of the encoding of $\Lambda_2$ is simply the offline part $\hf_\off$.

Thus, the decoder has the encoding $(A(\qrv{x}), \hf_\off)$. Suppose for the moment that the decoder had two out of the four labels $(\ell_{z,a}', \ell_{x,b}')$ for some $(a,b) \in \{0,1\}^2$. Then the tuple $(\hf_\off, \ell_{z,a}', \ell_{x,b}')$ denotes the RE $\hat{f}_{\corr,A,\ell,s,t}(a,b)$, from which the value $f_{\corr,A,\ell,s,t}(a,b) = \Lambda_2(R,\ell,s,t) \cdot A^\dagger$ can be efficiently decoded. Given this, the decoder can then decode $\Lambda_2(R,\ell,s,t)(\qrv{x})$, as desired. The privacy of this encoding follows from the group-randomizing property described above.

\subsection{Putting it Together} 

We now put everything together. For every pair $v,w$ of wires in the circuit connected by a gate $G$, the encoder applies the circuit $\Lambda_1(\ell_{vw})$ to the EPR pairs corresponding to $v,w$. Then, it computes the encoding of $\Lambda_2(\cdot,\ell_{vw},s_{vw},t_{vw})$ as described. Here, $s_{vw},t_{vw}$ are uniform randomness sampled for every connected pair $(v,w)$. The random labels $\ell_{vw}$ are generated from the classical RE $\hf_\corr$ of a $\Lambda_2$ circuit corresponding to a \emph{subsequent} pair $(v',w')$ of connected wires. 

Thus, the decoder can sequentially decode each wire\footnote{This should be viewed as analogous to the decoding in classical garbled circuits.} of the circuit in the following manner: inductively assume that the decoder has two-of-four labels $(\ell_{z,a}', \ell_{x,b}')$ resulting from a \emph{previous} encrypted teleportation measurement. Then the decoder can decode the group-randomizing QRE of the $\Lambda_2$ circuits, and then apply the $\Lambda_3$ circuits. For each pair $(v,w)$ of wires, the resulting state on the qubits $(\qrv{e}^v_\wout,\qrv{e}^w_\win)$ will be $\Lambda_3 \cdot \Lambda_2(R,\ell,s,t) \cdot \Lambda_1(\ell) \cdot G$, which by \Cref{fig:switch-order-intro} is equivalent to $\sf{TP}_{\ell,s,t} \cdot R \cdot G$, the desired gate followed by ``correct-then-teleport'' operation\footnote{Again we assume that the gate $G$ is a single-qubit gate; in the multiqubit case the gate $G$ is ``spread'' across multiple $\qrv{e}^v_\wout$ qubits.} Measuring the middle wire-bundles of the teleportation gadget to obtain labels $(\ell_{z,d}, \ell_{x,e})$ sets the decoder up for decoding a subsequent wire pair $(v',w')$. 

The encoder also provides a ``dictionary'' of labels for the output qubits, so that the decoder can undo any final Pauli/phase errors. All together, the decoding is equivalent to a sequence of correct-then-teleport operations, followed by a final correction, which allows the decoder to recover the value $C(\qrv{x})$. 

The encoding is constant-depth. The classical RE used in the encoding of each $\Lambda_2$ circuit can be done all in parallel and in constant-depth (assuming the CRE scheme has constant-depth encoding). Then the quantum part of the encoding for each pair of wires $(v,w)$ can be computed in constant-depth, because the $\Lambda_1$ circuits are constant-depth and the group $\cal{G}$ consists of constant-depth circuits. 

In the setting of information-theoretic privacy, the size of the encoding grows doubly-exponentially with the depth of the circuit. This is because if the labels $\ell$ for a pair of wires is $\kappa$ bits, then the corresponding teleportation gadget $\sf{TP}_{\ell,s,t}$ has size $\Theta(\kappa)$, and thus the corresponding $\Lambda_2$ circuit has size $\Theta(\kappa^2)$, which means that the complexity of computing the correction function $f_\corr$ is at least $\Theta(\kappa^2)$. This means that the labels for statistically-secure classical RE of $f_\corr$ are at least $\Omega(\kappa^2)$ bits long. Thus the label sizes square with each layer. In the worst-case, this means that we can only achieve efficient statistically-private encodings of $O(\log \log n)$-depth circuits.\footnote{In a previous version of this paper, we erroneously claimed that the label sizes grow linearly with each layer in the information-theoretic setting, and concluded that the encoding size grows exponentially with the depth (which would match the encoding complexity of known information-theoretic classical RE schemes). We find it an interesting problem to determine whether we can avoid the doubly-exponential complexity blow-up for information-theoretic QRE.} With some simple modifications to this scheme, it is possible to handle certain circuits of larger depth, e.g., $O(\log n)$-depth circuits where all layers except for the first $O(\log \log n)$ layers consist of Clifford gates.

In the setting of computational privacy (e.g., where we assume the existence of pseudorandom generators), the size of the encoding is polynomial in the size of the circuit. This is because the label sizes of computationally-secure classical RE are \emph{independent} of the complexity of the function being encoded; they only depend on the security parameter.

\subsection{Privacy of Our Scheme}
The privacy properties of our randomized encoding scheme is established via the existence of a \emph{simulator}, which is an efficient procedure $\Sim$ that takes as input a quantum state $F(\qrv{x})$ for some quantum circuit $F$ and state $\qrv{x}$, and produces another quantum state that is indistinguishable from the randomized encoding $\hat{F}(\qrv{x})$. This formalizes the idea that the only thing that can be learned from the randomized encoding $\hat{F}(\qrv{x})$ is the value $F(\qrv{x})$. 

An important feature of the randomized encoding $\hat{F}(\qrv{x})$ is that it hides the specific names of the gates being applied in a circuit $F$, and only reveals the \emph{topology} of the circuit $F$. This feature automatically implies the privacy of the randomized encoding: this means that it is not possible to distinguish between the encoding of $F$ on input $\qrv{x}$, or the encoding of a circuit $E$ with the same topology as $F$, but with all identity gates, with input $F(\qrv{x})$. In both cases, the decoding process (which is a public procedure which does not require any secrets) produces the output $F(\qrv{x})$. Thus, there is a canonical choice of simulator for such a randomized encoding: given input $F(\qrv{x})$, it computes the randomized encoding $\hat{E}(F(\qrv{x}))$, which is indistinguishable from the randomized encoding $\hat{F}(\qrv{x})$ via the circuit hiding property.

We now discuss a subtle point. We have described the decoding/evaluation procedure as involving measurements: the decoder is supposed to measure the middle wire-bundles of the teleportation gadgets to obtain the labels needed to decode the group-randomizing QRE in a subsequent teleportation. However, there is no guarantee that a malicious evaluator (who is trying to learn extra information from the encoding) will perform these measurements. The encoding of a gate technically gives a superposition over all possible labels of a teleportation gadget, and a malicious evaluator could try to perform some coherent operation to extract information about labels for different teleportation keys simultaneously (which would in turn compromise the privacy of the classical RE). 

In our scheme, we in fact define the honest decoding procedure $\Dec$ to be unitary. The randomization bits $s_x,s_z,t_x,t_z$ in the teleportation gadgets, which are never revealed to the decoder, effectively \emph{force} a measurement of the teleportation gadget labels. In particular, randomizing over the $s_z,t_z$ bits destroy any coherence between the different labels in the superposition, which means that the decoder (even a malicious one) can only get labels for a single teleportation key per teleportation gadget at a time. 

The unitarity of $\Dec$ is used in the privacy analysis in the following way: we show that given the randomized encoding of $F(\qrv{x})$ and $E(F(\qrv{x}))$, applying a unitary decoder $\Dec$ to both states yields outputs that are indistinguishable (statistical or computational, depending on the privacy properties of the classical RE scheme). Thus, applying the inverse unitary $\Dec^{-1}$ to the outputs preserves the indistinguishability (up to a loss related to the complexity of $\Dec$), and shows that the encodings of $F(\qrv{x})$ and $E(F(\qrv{x}))$ are indistinguishable.

\subsection{Another Simple QRE Using Group-Randomizing QRE}
\label{sec:overview:simpleqre}

We note that it is possible to construct a simpler QRE scheme for circuits using the group-randomizing QRE in different and arguably more straightforward manner. This construction does not have the low complexity property, or the gate-by-gate encoding property as our main construction presented above, but it is simpler and carries conceptual resemblance to classical branching program RE via matrix randomization as the well known Kilian RE \cite{Kilian88}.

The idea is to use the ``magic state'' representation of quantum circuits \cite{BK05magic}. At a high level and using the terminology of this work, \cite{BK05magic} shows that any quantum circuit can be represented (without much loss in size and depth) by a layered circuit as follows. Each layer consists of a unitary Clifford circuit with two types of outputs. Some of the output qubits are passed to the next layer as inputs (hence each layer has fewer inputs than its predecessor), and some of them are measured and the resulting classical string determines the gates that will be applied in the next layer (the last layer of course contains no measured qubits).\footnote{In fact, the \cite{BK05magic} characterization is much more specific about what the layers look like and only uses very specific classical characterization, in particular each measured bit controls one gate, but for our purposes even the above suffices.} The inputs to the first layer consist of the input to the original circuit, and in addition some auxiliary qubits, each of which is independently sampled from an efficiently samplable distribution over single-qubit quantum states (called ``magic state'').

Given our methodology above, one can straightforwardly come up with a QRE for such circuits. Given a circuit in the aforementioned layered form, and an input, the encoding is computed as follows.
\begin{enumerate}
	\item Generate the required number of magic states and concatenate them with the given input.
	
	\item If the circuit contains only one layer, i.e.\ is simply a Clifford circuit, use group-randomizing encoding (there is no need for classical RE in this case). 
	
	\item If the circuit contains more than one layer, consider the last layer (that produces the output), we refer to the layer before last as the ``predecessor layer''. 
	
	\begin{enumerate}
		\item 	Generate (classical) randomness that will allow to apply group-randomizing QRE on the last layer (including decomposable classical RE of the classical part of the encoding). This includes a randomizing Clifford $R$ and randomness for classical RE.
		
		\item Modify the description of the predecessor layer so that instead of outputting its designated output, it essentially outputs the QRE of the last layer. Specifically, modify the predecessor layer as follows. For the outputs that are passed to the next layer, add an application of $R$ before the values are actually output (since $R$ is Clifford, the layer remains a Clifford layer).

		 For the outputs that are to be measured, add a $Z$-twirl (i.e.\ $Z^s$ for a random $s$)\footnote{Applying $Z^s$ for a randomly chosen bit $s$ removes all information except the information that is recoverable via measurement (in the computational basis). Therefore one can think of $Z$-twirl as equivalent to measurement (or as a randomized encoding of the measurement operation).} followed by a Clifford circuit that selects between the two labels of the classical RE of the following (i.e.\ last) layer description. Also always output the (fixed classical) offline part of the classical RE.
		 
		 This transformation maintains the invariant that the new last layer (which is the augmented predecessor layer) is a Clifford circuit where the identity of the gates is determined by a classical value that comes from predecessor layers.
		
		\item Remove the last layer from the circuit and continue recursively.
	\end{enumerate}
\end{enumerate}

Correctness and security follow from those of the group-randomizing QRE and the $Z$-twirl. 
While this QRE is not natively decomposable, it can be made decomposable by adding a single layer of teleportation-based encoding (as in our full-fledged scheme) at the input. Interestingly, the only quantum operation required in this QRE is an application of a random Clifford on the input (more accurately, extended input containing the actual input and a number of auxiliary qubits in a given fixed state).

Carefully keeping track of the lengths of the labels of the classical RE would imply again that for perfect security we may incur up to a double-exponential blowup of the label length as a function of the depth.\footnote{While the description above was completely sequential, one can notice that a blowup in labels of a certain gates does not effect other gates that are in the same level in the circuit, even though the encoding is described sequentially.} In the computational setting, the blowup is only polynomial.

In terms of efficiency of encoding, we tried to present the scheme in a way that would make it easiest to verify its correctness and security, but an efficiency-oriented description would allow to encode all layers in parallel. This is because the modification to each layer only depends on the randomness of the QRE of the next layer, which can be sampled ahead of time.

\section{Preliminaries}
\label{sec:prelims}

\subsection{Notation}
\label{sec:prelim:notation}

\paragraph{Registers.} A \emph{register} is a named Hilbert space $\C^d$ for some dimension $d$. We denote registers using sans-serif font such as $\reg{a},\reg{b},\reg{c}$, etc. Let $\reg{w}_1,\reg{w}_2,\ldots,\reg{w}_n$ be a collection of registers. For a subset $S = \{i_1,\ldots,i_k\} \subseteq [n]$, we write $\reg{w}_S$ to denote the union of registers $\bigcup_{i \in S} \reg{w}_i$. We write $\dim(\reg{a})$ to denote the dimension of register $\reg{a}$. 

We also use underbrackets to denote the registers associated with a state, e.g.,
\[
	\ubket{\psi}{\reg{a}}\;,  \qquad \text{and} \qquad \frac{1}{\sqrt{2}} \sum_{e} \ubket{e,e}{\reg{v} \reg{u}} \;.
\]
The first denotes a pure state $\ket{\psi}$ in the register $\reg{a}$, and the second denotes an EPR pair on registers $\reg{v} \reg{u}$, respectively.

\paragraph{Quantum Random Variables.} A \emph{quantum random variable (QRV)} $\qrv{a}$ on a register $\reg{a}$ is a density matrix on register $\reg{a}$. Note that we denote QRVs using bolded font. When referring to a collection $(\qrv{a},\qrv{b},\qrv{c})$ of QRVs simultaneously, we are referring to the reduced density matrix of a global state on the registers $\reg{a} \reg{b} \reg{c}$ -- we say that $(\qrv{a},\qrv{b},\qrv{c})$ is a \emph{joint QRV}, which is also a QRV itself. We say that a QRV $\qrv{a}$ is \emph{independent} of a collection of QRVs $(\qrv{b},\qrv{c})$ if the density matrix corresponding to $\qrv{a}$ is in tensor product with the density matrix corresponding to $(\qrv{b},\qrv{c})$. We denote this by $\qrv{a} \otimes (\qrv{b},\qrv{c})$.

\paragraph{Quantum Operations.} Given a quantum operation $F$ mapping register $\reg{a}$ to register $\reg{a}'$ and a collection of QRVs $(\qrv{a},\qrv{b})$, we write 
$(F(\qrv{a}), \qrv{b})$ to denote the a density matrix on registers $\reg{a}' \reg{b}$ that is the result of applying $F$ to the density matrix $(\qrv{a},\qrv{b})$. Given quantum operations $F$ and $G$ that act on disjoint qubits, we write $(F,G)$ to denote the product operation $F \otimes G$. For a unitary $U$, we also write $U(\qrv{x})$ as shorthand for $U \, \qrv{x} U^\dagger$.

\paragraph{Quantum Circuits and Their Descriptions.} Throughout this paper we will talk about quantum circuits both as quantum operations (e.g. a unitary), and as classical descriptions of a sequence of gates. Formally, one is an algebraic object and the other is a classical string (using some reasonable encoding format for quantum circuits). To distinguish between the two presentations we write $\gadget{Ckt}$ to denote a classical description of a circuit (i.e. a sequence of gates on some number of qubits), and use sans-serif font such as $\sf{Ckt}$ to denote the corresponding unitary.

\paragraph{EPR Pair.} We let $\ket{\epr}$ denote the maximally entangled state on two qubits, i.e. $\frac{1}{\sqrt{2}} \left ( \ket{00} + \ket{11} \right)$.

\paragraph{Distinguishability of Quantum States.} We say that two quantum states $\qrv{a}, \qrv{b}$ on the same number of qubits are \emph{$(t,\epsilon)$-indistinguishable} if for all auxiliary quantum states $\qrv{q}$ that are unentangled with either $\qrv{a}$ or $\qrv{b}$, for all quantum circuits $D$ of size $t$, 
\[
	\Big |\, \Pr[D(\qrv{a} \otimes \qrv{q}) = 1] - \Pr[D(\qrv{b} \otimes \qrv{q}) = 1] \, \Big | \leq \epsilon.
\]
We often write $\qrv{a} \approx_{(t,\epsilon)} \qrv{b}$ to denote this.
This notion of indistinguishability satisfies the triangle inequality: if $\qrv{a},\qrv{b},\qrv{c}$ are quantum states such that $\qrv{a} \approx_{(t_1,\epsilon_1)} \qrv{b}$, and $\qrv{b} \approx_{(t_2,\epsilon_2)} \qrv{c}$, then we have $\qrv{a} \approx_{(\min(t_1,t_2),\epsilon_1 + \epsilon_2)} \qrv{c}$.

Furthermore, if $\qrv{a} \approx_{(t,\epsilon)} \qrv{b}$ and $U$ is a size-$s$ quantum circuit, then $U(\qrv{a}) \approx_{(t - s,\epsilon)} U(\qrv{b}))$. This is because if there was a size-$(t-s)$ circuit $D$ that could distinguish between the two with advantage more than $\epsilon$, then there exists a circuit $D' = D \circ U$ of size $t$ that could distinguish between $\qrv{a}$ and $\qrv{b}$ with advantage more than $\epsilon$.

\subsection{Quantum Gates and Circuits}

\subsubsection{Pauli, Clifford, and PX Groups}
\label{sec:prelim:pauliclifford}

\paragraph{Pauli Group.} The single-qubit \emph{Pauli group} $\cal{P}$ consists of the group generated by the following Pauli matrices: 
\[
	I = \begin{pmatrix} 1 & 0 \\ 0 & 1 \end{pmatrix} \qquad X = \begin{pmatrix} 0 & 1 \\ 1 & 0 \end{pmatrix} \qquad Y = \begin{pmatrix} 0 & i \\ -i & 0 \end{pmatrix} \qquad Z = \begin{pmatrix} 1 & 0 \\ 0 & -1 \end{pmatrix}
\]
The $n$-qubit Pauli group $\cal{P}_n$ is the $n$-fold tensor product of $\cal{P}$. 

\paragraph{Clifford Group.} The $n$-qubit \emph{Clifford group} $\cal{C}_n$ is defined to be the set of unitaries $C$ such that
\[
	C \cal{P}_n C^\dagger = \cal{P}_n.
\]
Elements of the Clifford group are generated by CNOT, Hadamard ($H = \frac{1}{\sqrt{2}} \begin{pmatrix} 1 & 1 \\ 1 & -1 \end{pmatrix}$), and Phase ($P = \begin{pmatrix} 1 & 0 \\ 0 & i \end{pmatrix}$) gates.

The \emph{third level of the $n$-qubit Clifford hierarchy} (Pauli and Clifford groups being the first and second, respectively) are unitaries $U$ where
\[
	U \cal{P}_n U^\dagger = \cal{C}_n.
\]
In other words, conjugating Paulis by $U$ yield a Clifford element.

\def\pxg{\cal{PX}}

\paragraph{PX Group.} The following subgroup of the Clifford group, which we call the PX-group, will be of particular interest to us. The group is defined as a group of single-qubit unitaries, and can be extended to multiple qubits by tensoring as usual. We define the single-qubit PX-group to be the group generated by the Pauli $X$ gate and Phase $P$ gate.

The aforementioned $P$ operation is the ``square root'' of the Pauli $Z$, so $P^2=Z$ (and therefore the Pauli group is a subgroup of the PX-group). The Pauli group is a strict subset of the PX-group (since the Pauli group does not contain $P$), and the PX-group itself is a strict subgroup of the single-qubit Clifford group (since the PX group does not contain the Hadamard gate). 

Calculation shows that $P X = i X P^3$, which implies that any element in the PX-group can be written as $i^{c} X^a P^b$, where $a \in \binset$ and $b,c \in \{0,1,2,3\}$. This immediately implies that given a circuit that contains only $I, X, P$ gates (and of course also $Z=
P^2$ gates which are equivalent to two consecutive applications of $P$), it is possible to efficiently find its canonical representation as $i^{c} X^a P^b$. We refer to such a circuit as a ``PX circuit''. Given the canonical representation of a PX element, it is possible to find a canonical PX circuit that implements the operation of the PX circuit. Since the PX group is a group, applying a random PX operation on an arbitrary PX operation results in a random PX operation.

\subsubsection{Universal Gate Set} 
\label{sec:prelim:universal}

A universal set of gates is $\cal{C}_2 \cup \{T\}$, i.e. the set of two-qubit Clifford gates along with
\[
	T = \begin{pmatrix} 1 & 0 \\ 0 & e^{i\pi/4} \end{pmatrix} \;. 
\]
The $T$ gate is an example of a unitary from the third level of the Clifford hierarchy, as formalized in the following fact. 
\begin{fact}\label{fact:tpauli}
For all $a,b \in \binset$ it holds that $T X^a Z^b = P^{a} X^{a} Z^{b}  T$.
\end{fact}

\subsubsection{Classical Circuits}
\label{sec:classical-circuits}

Here we briefly review classical circuits. A \emph{classical circuit topology} $T$ consists of a directed acyclic graph (DAG) where the nodes are divided into \emph{input terminals}, \emph{placeholder gates}, and \emph{output terminals}. Input terminals have in-degree $0$ and arbitrary out-degree (i.e. they are source nodes). Output terminals have in-degree $1$ and out-degree $0$ (i.e. they are sink nodes). Placeholder gates have constant in-degree and out-degree (without loss of generality this constant can be $2$ while incurring only a constant blowup in size and depth compared to any other constant).
A classical circuit $C$ with topology $T$ is simply an assignment of 
boolean functionalities to the placeholder gates.

The \emph{depth} of a circuit is simply the length of the longest path from an input terminal to an output terminal. The \emph{size} of a circuit is the number of wires (i.e.\ edges) in the circuit topology. 

An important class of circuits are \emph{constant-depth circuits}. These are captured by the complexity class $\ncz$, which technically consists of function families $\{f_n\}$ that can be computed by a family of polynomial-size circuits whose depth is bounded by a constant (i.e.\ does not grow with $n$). As we shall see in \Cref{sec:cre}, the class $\ncz$ captures the complexity of encoding in classical randomized encoding schemes.

\subsubsection{Quantum Circuits and Their Topology}
\label{sec:circuits}

\paragraph{Circuit Topology.} A \emph{quantum circuit topology} $\scr{T}$ is a tuple $(\cal{B},\cal{I},\cal{O},\cal{W},\winwire,\woutwire,\cal{Z},\cal{T})$ where 
\begin{enumerate}
	\item $G_{\scr{T}} = (\cal{B} \cup \cal{I} \cup \cal{O},\cal{W})$ forms a directed acyclic graph (DAG) where the vertex set consists of the union of disjoint sets $\cal{B}$, $\cal{I}$, $\cal{O}$, and the edge set is $\cal{W}$. 
	\item The set of edges $\cal{W}$ are called \emph{wires} of the circuit topology $\scr{T}$. 
	\item The set $\cal{I}$ is ordered, and consist of \emph{input terminals}, and have in-degree $0$, and out-degree $1$. Throughout this paper we will overload notation and let $\cal{I}$ denote the subset of wires $\cal{W}$ that are incident to the input terminals. 
	\item The set $\cal{O}$ is ordered, and consist of \emph{output terminals}, and have in-degree $1$, and out-degree $0$. Throughout this paper we will overload notation and let $\cal{O}$ denote the subset of wires $\cal{W}$ that are incident to the output terminals. 
	\item The vertices $\cal{B}$ are called \emph{placeholder gates}, and for every $g \in \cal{B}$, the in-degree and out-degree of $g$ are equal. We let $p_g$ denote the degree, which we call the \emph{arity} of the placeholder gate $g$.
	\item For every $g \in \cal{B}$, we let $\winwire(g)$ denote an ordering of the wires $w \in \cal{W}$ that enter $g$, and let $\woutwire(g)$ denote an ordering of the wires that exit $g$.
	\item The set $\cal{Z}$ is a subset of $\cal{I}$ that denotes \emph{zero qubits} (i.e. qubits to initialize in the state $\ket{0}$).
	\item $\cal{T} \subseteq \cal{O}$ denotes the set of \emph{discarded qubits} (i.e. output qubits to trace out).
\end{enumerate}	

A topology thus specifies a quantum circuit in a natural way, except only ``placeholder gates'' are specified. Note that the number of input terminals must be equal to the number of output terminals; these correspond to the input and output wires of the circuit topology. We often let $n$ denote the number of input (and therefore output) terminals. Furthermore, a circuit topology allows some input qubits to be initialized in the $\ket{0}$ state, and some output qubits to be traced out.

Given a topology $\scr{T}$, we define an \emph{evaluation order} $\pi: \cal{B} \to |\cal{B}|$ to be a topological ordering of the ``blank gates''.

\paragraph{Quantum Circuits.} A \emph{general quantum circuit $F$} is a pair $(\scr{T},\cal{G})$ where 
\[
\scr{T} = (\cal{B},\cal{I},\cal{O},\cal{W},\winwire,\woutwire,\cal{Z},\cal{T})
\] 
is a circuit topology and $\cal{G}$ is a set of unitaries such that for every $g \in \cal{B}$, there is a corresponding $p_g$-qubit unitary $U_g \in \cal{G}$. We often write $g \in \cal{G}$ to denote the unitary itself. The \emph{size} of a unitary circuit $C$ is the number of wires $|\cal{W}|$. For an $(n - |\cal{Z}|)$-qubit state $\qrv{w}$ supported on the qubits not indexed by $\cal{Z}$, we write $F(\qrv{w})$ to denote the density matrix resulting from applying the gates $g \in \cal{G}$ to $(\qrv{w},0)$ where the qubits indexed by $\cal{Z}$ are set to $\ket{0}$, and at the end the qubits indexed by $\cal{T}$ are traced out. 

A \emph{unitary quantum circuit $C$} is a circuit where the set $\cal{Z} = \cal{T} = \emptyset$. In other words, it maps $n$ qubits to $n$ qubits, and no qubits are discarded at the end.

By the Stinespring dilation theorem, every quantum operation can be realized as a general quantum circuit. We associate the complexity of a quantum operation with the size of the quantum circuit that implements it, with respect to a universal set of gates. For this work, we choose to work with the universal set $\cal{C}_2 \cup \{T\}$ as described in Section~\ref{sec:prelim:universal}.

\paragraph{Constant-Depth Quantum Circuits.} The main model of constant-depth quantum circuits that we consider in this paper are $\qnczf$ circuits, which are constant-depth circuits consisting of one- or two-qubit gates, as well as \emph{fan-out} gates of arbitrary arity, which copy a control qubit to a number of target qubits (i.e., a fan-out gate with fan-out $k$ performs the following transformation: $\ket{x,y_1,\ldots,y_k} \mapsto \ket{x,y_1\oplus x,y_2 \oplus x,\ldots,y_k \oplus x}$). This is a natural analogue of classical $\ncz$ circuits, yet is surprisingly powerful: functions such as PARITY can be computed in this model~\cite{M99,HS05}. We will show that $\qnczf$ captures the complexity of the encoding procedures of our quantum randomized encoding scheme.

\subsection{Classical Randomized Encoding}
\label{sec:cre}

We define classical randomized encoding schemes and their properties. See survey by Applebaum~\cite{Applebaum17} for details and references.

\begin{definition}[Classical Randomized Encoding]
\label{def:cre}
	Let $f \in \binset^n \to \binset^m$ be some function. The function $\hf: \binset^n \times \binset^\ell \to \binset^{m'}$ is a $(t,\epsilon)$-private \emph{classical randomized encoding} (CRE) of $f$
	if there exist a deterministic function $\CDec$ (called a \emph{decoder}) and a randomized function $\CSim$ (called a \emph{simulator}) with the following properties.
	
	\begin{itemize}
		\item \textbf{Correctness.} For all $x,r$ it holds that $f(x) = \CDec(\hf(x ; r))$.\footnote{This is known as \emph{perfect correctness} and is the only notion of correctness considered in this work.}
		
		\item \textbf{$(t,\epsilon)$-Privacy.} For all $x$ and for all circuits $D$ of size $t$ it holds that
		\[
		\Abs{\Pr_r[D(\hf(x ; r))=1] - \Pr[D(\CSim(f(x))=1]} \le \epsilon~
		\]
		where the second probability is over the randomness of the simulator $\CSim$. The case of $\epsilon=0$ is called \emph{perfect privacy}.
	\end{itemize}
	The encoding $\hf$ is a \emph{decomposable} CRE (DCRE) of $f$ if there exist functions $\hf_{\off}(r)$ (called the \emph{offline part of the encoding}) and $\lab_{i,b}(r)$ (called the \emph{label functions}) for all $i \in [n], b \in \binset$, such that for all $(x,r)$, 
	\[
		\hf(x; r) = \Big ( \, \hf_{\off}(r) \, ,\,  (\lab_{i,x_i}(r))_{i \in [n]} \, \Big) \;.
	\]
\end{definition}

\begin{remark}
	We refer to the second input $r$ of $\hat{f}$ as the \emph{randomness} of the encoding, and we use a semicolon to distinguish it from the deterministic input $x$. 	We will sometimes write $\hat{f}(x)$ to denote the random variable $f(x; r)$ induced by sampling $r$ from the uniform distribution. Furthermore, we say that the \emph{value} $\hf(x)$ is the randomized encoding of a function $f$ and a deterministic input $x$. 
\end{remark}

Note that as presented in \Cref{def:cre}, there is no requirement that the randomized encoding $\hat{f}$ can be efficiently computed from the original function $f$. Furthermore, the decoder $\CDec$ and simulator $\CSim$ are technically allowed to depend arbitrarily on the function $f$ be encoded. However, it is a highly desirable feature that randomized encodings be efficiently computable given a description of $f$, and also have a \emph{universality} property (see, e.g., Section 7.6.2 of \cite{Applebaum14szk}), where the encoding $\hat{f}(x)$ hides information not just about the input $x$, but also about the function $f$. This is formalized by requiring that the decoding and simulation procedures depend only partially on $f$. In many cases, including in this work, they should only depend on the \emph{topology} of the circuit computing $f$ (see \Cref{sec:classical-circuits} for an overview of classical circuit topology). This motivates the following general definition.

\begin{definition}[Universal RE Schemes for Circuits]
\label{def:universal-re}
	Let $\cC$ denote a class of circuits and let $\cR$ denote an equivalence relation over $\cC$. An \emph{$(t,\epsilon)$-private and efficient $\cR$-universal RE scheme for the class $\cC$} is a tuple of polynomial-time algorithms $(\CEnc,\CDec,\CSim)$ such that for all circuits $f \in \cC$ (here we identify the circuit with the function it computes),
	\begin{itemize}
		\item \textbf{Efficient Encoding}. For all $x, r$, $\CEnc(f, x ; r)$ computes a randomized encoding $\hf(x, r)$.
		\item \textbf{Correctness}. For all $x$ and $r$ it holds that $f(x) = \CDec(c,\hf(x ; r))$ where $c$ denotes the equivalence class of $f$ in $\cR$.
		\item \textbf{$(t,\epsilon)$-Privacy}. For all $x$ and circuits $D$ of size $t$ it holds that
		\[
		\Abs{\Pr_r[D(\hf(x ; r))=1] - \Pr[D(\CSim(c,f(x))=1]} \le \epsilon~.
		\]		
		where $c$ denotes the equivalence class of $f$ in $\cR$.
	\end{itemize}
	Furthermore, we say that the $\cR$-universal RE scheme $(\CEnc,\CDec,\CSim)$ is \emph{decomposable} if the randomized encoding $\hf(x; r)$ is decomposable, and furthermore we say that it is \emph{label-universal} if the label functions $\lab_{i,b}(r)$ only depend on the equivalence class of $f$ in $\cR$.
\end{definition}

\begin{remark}[Universal RE and Universal Circuits]\label{rmk:universalcircuitcre}
	For many classes of functions it is possible to achieve universality for RE using the notion of a universal circuit (or machine). If the class of functions admits a universal circuit (which depends on a property such as the topology but takes the remainder of the description as input), and this universal circuit itself belongs to the class that can be encoded using the RE scheme, then one can apply the RE to the universal circuit and consider the description of $f$ as additional input to the circuit. This will result in a universal RE scheme, and if the original RE was decomposable then the resulting scheme will be decomposable with respect to both the input and the description of the function.
\end{remark}

For the remainder of this paper, when we speak of encoding a function $f$, we are referring to encoding a specific \emph{circuit implementation} of $f$ (see \Cref{sec:classical-circuits}).
Furthermore, in this paper we will be focused on \emph{topologically-universal RE schemes} -- in other words, the equivalence relation $\cR$ is such that two circuits $f,f'$ are equivalent if they have the same topology (i.e.\ the same interconnection between gates, but possibly different gate functionality). Randomized encoding schemes in the literature are typically topologically-universal (for example Yao's garbled circuits scheme).

We note that label universality can be derived from $\cR$-universality in a generic way (essentially by using a straightforward label-universal encoding of the multiplexer functions), but the constructions we cite from the literature will have this property even without this transformation.

\paragraph{Existing Decomposable Classical RE Schemes.} We use decomposable CRE (DCRE) as a building block for our construction of quantum RE. In particular we rely on the following information theoretical and computational schemes \cite{BMR90,Rogawaythesis,Applebaum17}.

\begin{theorem}[Information Theoretic DRE]\label{thm:classicalperfectre}
	There exists an efficient topologically-universal and label-universal DRE scheme $\CRE = (\CEnc,\CDec,\CSim)$ with the following properties: 
	\begin{itemize}
		\item \textbf{Efficiency.} For every function $f$ computable by a size-$s$ and depth-$d$ classical circuit and for every input $x$, the encoding $\CEnc(f, x; r)$ is computable in time $\poly(2^d) \cdot s$. Furthermore, the label functions $\lab_{i,b}(r)$ are computable in time $\poly(2^d)$. The decoding and simulation algorithms $\CDec$ and $\CSim$ are also computable in time $\poly(2^d) \cdot s$.

		\item \textbf{Perfect Information-Theoretic Privacy.} The scheme has perfect privacy against the class of \emph{all} distinguishers. 
		
		\item \textbf{Locality.} Every output bit of $\hf(x; r) = \CEnc(f, x; r)$ depends on at most $4$ bits of $(x,r)$.
	\end{itemize}
\end{theorem}

We note that the locality and efficiency properties of the DRE scheme specified by \Cref{thm:classicalperfectre} implies that the randomized encodings $\hf(x ; r)$ are computable by $\ncz$ circuits that take as input $x$ and the randomness $r$.

\begin{theorem}[Computational DRE]\label{thm:classicalcomre}
	Assume there exists a length doubling pseudorandom generator (PRG) $\prg$ that is secure against polynomial time classical (resp. quantum) adversaries.	
There exists an efficient topologically-universal and label-universal DRE scheme $\CRE = (\CEnc,\CDec,\CSim)$, which implicitly depends on a security parameter $\secp$, and has the following properties:
	\begin{itemize}
		\item \textbf{Efficiency.} For every function $f$ computable by a size-$s$ classical circuit and for every input $x$, the encoding $\CEnc(f, x; r)$ is computable in time $\poly(\secp) \cdot s$. Furthermore, the label functions $\lab_{i,b}(r)$ are computable in time $\poly(\secp)$. The decoding and simulation algorithms $\CDec$ and $\CSim$ are also computable in time $\poly(\secp) \cdot s$.

		\item \textbf{Computational Privacy.} For every polynomial $t(\secp)$, there exists a negligible function $\epsilon(\secp)$ such that the scheme is $(t(\secp),\epsilon(\secp))$-private against the class of size-$t(\secp)$ classical (resp. quantum) distinguishers.

		\item \textbf{Locality.} Every output bit of $\hf(x; r) = \CEnc(f,x; r)$ depends on at most $4$ bits of $(x,r,\prg(r))$. If the PRG $\prg$ can be computed by a $O(\log(\secp))$-depth circuit, then every output bit of $\hf(x;r)$ can be made to depend on at most $4$ bits of $(x,r)$, via non-black-box use of the PRG.
		
	\end{itemize}
\end{theorem}

We note that the locality and efficiency properties of the DRE scheme specified by \Cref{thm:classicalcomre} implies that the randomized encodings $\hf(x; r)$ are computable by $\ncz$ circuits that take as input $x$, the randomness $r$, and the output of the PRG $\prg$.

\begin{remark}
In the case of computational security, the RE scheme $(\CEnc,\CDec,\CSim)$ may also depend on a security parameter $\secp$. Since the security parameter will always be set and fixed in application, we do not explicitly point it out in our notation. 
\end{remark}

 \section{Quantum Randomized Encoding -- Definition and Existence}
\label{sec:qredefresult}

\subsection{Definition}
\label{sec:qredef}

We propose the following quantum analogue of randomized encoding.

\begin{definition}[Quantum Randomized Encoding]
	Let $F(\qrv{x})$ be a quantum operation that maps $n$ qubits to $m$ qubits. 
	The quantum operation $\hF(\qrv{x}; r)$ where $r$ is classical randomness is a \emph{$(t,\epsilon)$-private quantum randomized encoding} (QRE) of $F$ if there exist quantum operations $\Dec$ (called the \emph{decoder}) and $\Sim$ (called the \emph{simulator}) with the following properties.
	
	\begin{itemize}
		\item \textbf{Correctness.} For all quantum states $(\qrv{x},\qrv{q})$ and all randomness $r$, it holds that $(\Dec(\hF(\qrv{x}; r)),\qrv{q})=(F(\qrv{x}),\qrv{q})$.
		
		\item \textbf{$(t,\epsilon)$-Privacy.} For all quantum states $(\qrv{x},\qrv{q})$, we have
		\[
		\Big (\, \hF(\qrv{x}; r) \, ,\, \qrv{q} \, \Big) \approx_{(t,\epsilon)} \Big (\, \Sim(F(\qrv{x})) \, ,\, \qrv{q} \, \Big)
		\]
		where the state on the left-hand side is averaged over $r$.
		The case of $\epsilon=0$ is called \emph{perfect privacy}.
	\end{itemize}
The encoding $\hF$ is a \emph{decomposable} QRE (DQRE) if there exists
a quantum state $\qrv{e}$ (called the \emph{resource state of the encoding}), 
an operation $\hF_\off$ (called the \emph{offline part of the encoding}) and 
a collection of \emph{input encoding operations} $\hF_1,\ldots,\hF_n$ such that for all inputs $\qrv{x} = (\qrv{x}_1,\ldots,\qrv{x}_n)$,
\[
	\hF(\qrv{x}; r) = \big (\hF_\off, \hF_1, \hF_2, \ldots, \hF_n \big) \, (\qrv{x}, r, \qrv{e})
\]
where the functions $\hF_\off, \hF_1,\ldots,\hF_n$ act on disjoint subsets of qubits from $\qrv{e}, \qrv{x}$ (but can depend on all bits of $r$), each $\hF_i$ acts on a single qubit $\qrv{x}_i$, and $\hF_\off$ does not act on any of the qubits of $\qrv{x}$.

\end{definition}

Similarly to the classical case, we refer to the second input $r$ of $\hat{F}$ as the randomness of the encoding. We will often write $\hat{F}(\qrv{x})$ to denote the quantum state $\hat{F}(\qrv{x}; r)$ when $r$ is sampled from the uniform distribution. Furthermore, we say that the \emph{quantum state} $\hat{F}(\qrv{x})$ is the randomized encoding of the operation $F$ and an input $\qrv{x}$.

One can see that this definition of quantum randomized encoding is syntactically similar to \Cref{def:cre}, with a couple differences. First, the correctness and privacy properties involves the pair $(\qrv{x},\qrv{q})$. We refer the reader to \Cref{sec:prelim:notation} for the full explanation of the quantum random variable notation; but in short $(\qrv{x},\qrv{q})$ represents a bipartite density matrix with an $\qrv{x}$ part, and a $\qrv{q}$ part, and these parts may be entangled. The $\qrv{q}$ part is never acted upon by the decoder or simulator, but distinguishability is measured with respect to the encoding of $\qrv{x}$ as well as $\qrv{q}$, which we think of as quantum side information. In other words, correlations between the input and an external system are preserved through the encoding, decoding, and simulation. 

A second difference involves the definition of decomposable QRE. In addition to receiving a random string $r$, the randomized encoding also receives a auxiliary quantum state $\qrv{e}$ (that is independent of the input $\qrv{x}$). The definition allows for any resource state $\qrv{e}$, but in this paper we focus on decomposable QREs where the resource state $\qrv{e}$ is a collection of EPR pairs, which is perhaps the most natural quantum analogue of a randomness string.

Furthermore, similar to the classical setting, it is highly desirable to have efficient QRE schemes that are universal with respect to some property of the quantum operations being encoded, say the topology of some circuit implementation of them. This motivates the following definition of universal QRE scheme, in analogy to \Cref{def:universal-re}.

\begin{definition}[Universal QRE Schemes for Circuits]
\label{def:universal-qre}
	Let $\cC$ denote a class of general quantum circuits\footnote{See \Cref{sec:circuits} for the definition of general quantum circuits.} and let $\cR$ denote an equivalence relation over $\cC$. An \emph{$(t,\epsilon)$-private and efficient $\cR$-universal QRE scheme for the class $\cC$} is a tuple of polynomial-time quantum algorithms $(\Enc,\Dec,\Sim)$ such that given a circuit $F \in \cC$ (here we identify the circuit with the function it computes), 
	\begin{itemize}
		\item \textbf{Efficient Encoding}. For all quantum inputs $\qrv{x}$ and randomness $r$, $\Enc(F, \qrv{x}; r)$ computes a quantum randomized encoding $\hF(\qrv{x}; r)$.

		\item \textbf{Correctness}. For all quantum states $(\qrv{x},\qrv{q})$ and randomness $r$ it holds that $(F(\qrv{x}),\qrv{q}) = (\Dec(c,\hF(\qrv{x}); r),\qrv{q})$ where $c$ denotes the equivalence class of $F$ in $\cR$.
		
		\item \textbf{$(t,\epsilon)$-Privacy}. For all quantum states $(\qrv{x},\qrv{q})$, we have
		\[
		\Big (\, \hF(\qrv{x}; r) \, ,\, \qrv{q} \, \Big) \approx_{(t,\epsilon)} \Big (\, \Sim(c,F(\qrv{x})) \, ,\, \qrv{q} \, \Big)
		\]
		where the state on the left-hand side is averaged over the randomness $r$, and $c$ denotes the equivalence class of $F$ in $\cR$.
	\end{itemize}
	Furthermore, we say that the $\cR$-universal QRE scheme $(\Enc,\Dec,\Sim)$ is \emph{decomposable} if the randomized encoding $\hF(\qrv{x})$ is decomposable and if the input encoding operations $\hF_i$ only depend on the equivalence class of $F$ in $\cR$.
\end{definition}

For the remainder of this paper, when we speak of encoding a quantum operation $F$, we are referring to encoding a specific \emph{circuit implementation} of $F$. Furthermore, in this paper we will be focused on \emph{topologically-universal QRE schemes} -- in other words, the equivalence relation $\cR$ is such that two circuits $F,F'$ are equivalent if they have the same topology (see \Cref{sec:circuits} for the definition of quantum circuit topology).

\subsection{Our Main Result: Existence of Decomposable Quantum Randomized Encodings}
\label{sec:qreresult}

Our main result is an efficient topologically-universal decomposable QRE scheme, which we call \emph{Quantum Garbled Circuits}. We use classical decomposable RE as a building block.

\begin{lemma}[Quantum Garbled Circuits Scheme]\label{lem:main}

	Let $\CRE$ be an efficient topologically-universal and label-universal classical DRE scheme such that for classical circuits $f$ of size $s$ and depth $d$, the time complexity of encoding $f$ is $c(d,s)$ and the length of the labels is $\kappa(d,s)$. Furthermore, suppose that the encoding of $\CRE$ can be computed by $\ncz$ circuits, and that the scheme is $(t,\epsilon)$-private with respect to quantum adversaries.

Recursively define $\kappa_0 = O(1)$, $\kappa_i=\kappa(O(1), O(\kappa^2_{i-1}))$, and define $c_i=c(O(1), O(\kappa^2_{i-1}))$. (All the $O(\cdot)$'s refer to universal constants.)
	
	Then there exists an efficient topologically-universal decomposable QRE scheme $\QRE = (\Enc,\Dec,\Sim)$ that satisfies the following properties: 
	\begin{itemize}
		\item \textbf{Efficiency.} For every operation $F$ computable by a size-$s$ and depth-$d$ quantum circuit and for every quantum input $\qrv{x}$, the encoding $\Enc(F, \qrv{x}; r)$ is computable by a $\qnczf$ circuit of size $O(c_d \cdot s)$. The $\qnczf$ encoding circuit takes as input a string of random bits $r$, the quantum input $\qrv{x}$, and a collection of EPR pairs. Furthermore, the input encoding operations $\hF_i$ can be computed by $\qnczf$ circuits of size $O(\kappa_d)$. The running time of $\Dec$ and $\Sim$ is $O(c_d \cdot s)$. 
		
		\item \textbf{Classical Inputs.} If an input qubit $\qrv{x}_i$ is classical, then the input encoding operation $\hF_i$ is computable by a classical circuit.

		\item \textbf{Privacy.} The scheme $\QRE$ is $(t',\epsilon')$-private where $t' = t - \poly(c_d) \cdot s$ and $s' = \epsilon \cdot s$. Here, $s$ and $d$ refer to the size and depth of the circuit being encoded. 
	\end{itemize}
\end{lemma}

\begin{remark}
	As mentioned in Remark~\ref{rmk:universalcircuitcre}, we can apply our encoding scheme to a universal circuit 
	rather than to $F$ itself, and consider the classical description of $F$ as an additional input. This would incur some overhead due to the use of the universal circuit but will have properties that may be useful in some settings. In particular, the dependence of the encoding on the description of $F$ becomes very simple and since the description is classical, the encoding of the input $F$ also becomes classical. Furthermore, if the input $\qrv{x}$ is classical as well, then the quantum part of the encoding $\hF(\qrv{x})$ is independent of both $F,\qrv{x}$ and can be generated beforehand as a ``resource state'' that is given to the encoder. 
\end{remark}

The proof of Lemma~\ref{lem:main} is presented in Sections~\ref{sec:construction} and~\ref{sec:analysis}. Specifically, Section~\ref{sec:construction} describes the encoding, decoding, and simulation procedures of our Quantum Garbled Circuits scheme. The correctness and privacy properties of the scheme are then analyzed in \Cref{sec:analysis}.

By instantiating $\CRE$ in Lemma~\ref{lem:main} with the classical RE schemes from Theorem~\ref{thm:classicalperfectre} and Theorem~\ref{thm:classicalcomre}, respectively, the following theorems immediately follow.

\begin{theorem}[Information Theoretic DQRE]\label{thm:perfectqre}
		There exists an efficient topologically-universal decomposable QRE scheme $\QRE = (\Enc,\Dec,\Sim)$ with the following properties: 
\begin{itemize}
	\item \textbf{Efficiency.} For every operation $F$ computable by a size-$s$ and depth-$d$ quantum circuit and for every quantum input $\qrv{x}$, the encoding $\Enc(F, \qrv{x}; r)$ is computable by a $\qnczf$ circuit of size $\poly(2^{2^d}) \cdot s$. The $\qnczf$ encoding circuit takes as input a string of random bits $r$, the quantum input $\qrv{x}$, and a collection of EPR pairs. Furthermore, the input encoding operations $\hF_i$ can be computed by $\qnczf$ circuits of size $\poly(2^{2^d})$. The running time of $\Dec$ and $\Sim$ is $\poly(2^{2^d}) \cdot s$. 
	
	\item \textbf{Classical Inputs.} If an input qubit $\qrv{x}_i$ is classical, then the input encoding operation $\hF_i$ is computable by a classical circuit.
	
	\item \textbf{Perfect Information-Theoretic Privacy.} The scheme has perfect privacy against the class of \emph{all} distinguishers.
	\end{itemize}
\end{theorem}
\begin{proof}
Plugging the properties of the DCRE scheme from Theorem~\ref{thm:classicalperfectre} into the conditions of Lemma~\ref{lem:main}, we get $\kappa_i = \poly(2^{2^i})$, $c_i = \poly(2^{2^i})$ and $(t,\epsilon=0)$-privacy. The result thus follows.
\end{proof}

\begin{theorem}[Computational DQRE]\label{thm:compdqre}
	Assume there exists a length doubling pseudorandom generator (PRG) $\prg$ that is secure against polynomial-time quantum adversaries. There exists an efficient topologically-universal decomposable QRE scheme $\QRE = (\Enc,\Dec,\Sim)$, which implicitly depends on a security parameter $\secp$, and has the following properties:

	\begin{itemize}
		\item \textbf{Efficiency.} For every operation $F$ computable by a size-$s$ and depth-$d$ quantum circuit and for every quantum input $\qrv{x}$, the encoding $\Enc(F, \qrv{x}; r)$ is computable by a $\qnczf$ circuit of size $\poly(\secp) \cdot s$. The $\qnczf$ encoding circuit takes as input a string of random bits $r$, the output $\prg(r)$ of the PRG, the quantum input $\qrv{x}$, and a collection of EPR pairs.\footnote{As in the classical case, the PRG is used in a black-box manner. If we allow non-black-box use of the PRG and if it can be computed by a $O(\log \secp )$-depth circuit, then the encoding circuit can be made fully in $\qnczf$ that takes as input only $\qrv{x}$, the randomness $r$, and EPR pairs.} Furthermore, the input encoding operations $\hF_i$ can be computed by $\qnczf$ circuits of size $\poly(\secp)$. The running time of $\Dec$ and $\Sim$ is $\poly(\secp) \cdot s$. 
	
	\item \textbf{Classical Inputs.} If an input qubit $\qrv{x}_i$ is classical, then the input encoding operation $\hF_i$ is computable by a classical circuit.

		\item \textbf{Computational Privacy.} For every polynomial $t(\secp)$, there exists a negligible function $\epsilon(\secp)$ such that the scheme is $(t'(\secp,s),\epsilon'(\secp,s))$-private with respect to quantum adversaries, where $t'(\secp,s) = t(\secp) - \poly(\secp,s)$ and $\epsilon'(\secp,s) = \epsilon(\secp) \cdot s$ with $s$ being the size of the circuit being encoded.

	\end{itemize}
\end{theorem}
\begin{proof}
	Plugging the properties of the DCRE scheme from Theorem~\ref{thm:classicalcomre} into the conditions of Lemma~\ref{lem:main}, we get $\kappa_i = \poly(\secp)$, $c_i = \poly(\secp) \cdot s$. Since $\poly(c_d) \cdot s = \poly(\secp,s)$, we get that the scheme is $(t',\epsilon')$-private for the functions $t'(\secp,s)$ and $\epsilon'(\secp,s)$ specified in the Theorem statement.
	
\end{proof}

 \SetAlgorithmName{Protocol}{Protocol} 

\newcommand{\com}{\mathsf{Com}}
\newcommand{\ver}{\mathsf{Ver}}
\newcommand{\ZKSim}{\mathsf{ZKSim}}
\newcommand{\View}{\mathrm{View}}
\newcommand{\malv}{V^*}
\def\malV{\malv}
\newcommand{\malP}{P^*}

\newcommand{\malp}{P^*}
\def\malP{\malp}

\section{A New Zero-Knowledge $\Sigma$-Protocol for $\qma$}
\label{sec:zk}

We present a $3$-message zero-knowledge $\Sigma$-protocol for any $\qma$ problem. Our construction generically uses QRE with certain properties, as well as quantum-secure classical commitment schemes. Both are instantiable under (flavors of) quantum-secure one-way functions (QRE only requires general one-way functions, and the commitment schemes we use can be achieved with injective one-way functions or from arbitrary one-way functions in the presence of a common random string).

\subsection{Building Blocks}

We start by going over the building blocks that are used in our protocol.

\paragraph{A $\qma$ Problem with Almost-Perfect Soundness.} Let $L = (L_{yes},L_{no})$ be a promise problem in $\qma$. Let $V_L = \{V_{L,n}\}$ denote the corresponding $\qma$ verifier, specified as a (uniform) family of polynomial-size circuits that takes as input $(x,\qrv{w})$ where $x$ is an instance of the language and $\qrv{w}$ is a $\poly(|x|)$-sized quantum witness. If $x \in L_{yes}$, then $V_L(x,\qrv{w})$ accepts with probability at least $1 - \mu(|x|)$ and if $x \in L_{no}$, $V_L(x,\qrv{w})$ accepts with probability at most $ \mu(|x|)$ where $\mu(\cdot)$ is a negligible function. We let $\cal{R}_L(x)$ denote the set of witnesses $\qrv{w}$ on which $V_L$ accepts with probability at least $1-\mu(|x|)$.

\paragraph{A Quantum Randomized Encoding Scheme.} Let $\QRE$ denote an efficient, computationally-secure quantum randomized encoding scheme satisfying the properties of \Cref{thm:compdqre}. Since the $\QRE$ is decomposable and has classical encoding for classical inputs, we can assume that the scheme has the following structure: given a circuit $F$ and input $(\qrv{q},c)$ where $\qrv{q}$ is an $n_1$-qubit quantum state and $c$ is an $n_2$-bit classical string, the encoding $\hat{F}(\qrv{q},c)$ can be written as
\[
	\Big( \hF_\off, \hF_1, \ldots, \hF_{n_1 + n_2} \Big) (\qrv{q},c , r, \qrv{e})
\]
where $r$ is a uniformly random string, $\qrv{e}$ is a collection of EPR pairs, $\hF_\off$ acts on $(r,\qrv{e})$ only, and $\hF_{n_1+1}(c_1 ; r),\ldots,\hF_{n_1+n_2}(c_{n_2}; r)$ are classical circuits that encode each bit of $c$ separately. Thus for each $i \in [n_2]$ and $r$, we can define classical labels for each $c_i$:
\[
	\lab_{i,b}(r) = \hF_{n_1 + i}(b ; r)\;.
\]
Furthermore, we assume that for a fixed value of $r$, the operations $\hF_\off$ and $\hF_i$ can be implemented via polynomial-size unitary circuits, possibly using some additional zero ancillas.\footnote{This is true for our $\QGC$ scheme and can be assumed without loss of generality. The reason, briefly, is that we can always consider purified versions of the $\hF$ functions, and then instead of tracing out a part of the output, just encrypt it with a quantum one-time pad, using classical bits from a (possibly extended) $r$.}

We now consider the quantum functionality that, for a fixed value of $r$, takes as input the quantum input $\qrv{q}$ and a sequence of zero ancilla bits, and outputs the quantum state
\[
\Big (\hat{\qrv{F}}_0, \{ \lab_{i,b} \}_{i \in [n_2], b \in \binset}, 0 \Big) ~,
\] 
where $\hat{\qrv{F}}_0$ denotes the part output by $(\hF_\off,\hF_1,\ldots,\hF_{n_1})$, and where $0$ represents a sequence of zero qubits. 

Since the encoding is unitary as we explained above, this functionality can be implemented by a unitary circuit $U_r$ that first creates a number of EPR pairs $\qrv{e}$ from the zero qubits, and then applies the encoding functions to compute $\hat{\qrv{F}}_0$ and $\{\lab_{i,b} \}_{i,b}$. 
In other words, $U_r(\qrv{q},0)$ is almost the same as the QRE encoding of $\qrv{F}(\qrv{q},c)$, except for the classical input it outputs \emph{all} labels, and uses zero ancillas for scratch space. We note that for a fixed $c \in \binset^{n_2}$, the state $(\hat{\qrv{F}}_0, \{ \lab_{i,c_i} \}_{i \in [n_2]})$ constitutes the encoding of $F(\qrv{q}, c)$.

\paragraph{A Quantum-Secure Classical Commitment Scheme.} We require a perfectly-binding non-interactive commitment scheme secure against quantum adversaries. Such a commitment scheme is defined as a polynomial-time function $\com$ that takes as input a security parameter $\secp$, an input message $m$ and randomness $s$ and outputs a string $c = \com(1^\secp, x; s)$ with the following properties: 
\begin{itemize}
	\item (\emph{Perfect Binding}). For all $\secp$ there does not exist $x_1 \neq x_2$ and $s_1, s_2$ such that $\com(1^\secp, x_1; s_1)=\com(1^\secp, x_2; s_2)$.
	
	\item (\emph{Computational Hiding Against Quantum Adversaries}). 
	For any sequence of $\{x_{1,\secp}, x_{2,\secp}\}_{\secp}$, where $|x_{1,\secp}| = |x_{2,\secp}| = \poly(\secp)$, it holds that the distributions $\com(1^\secp, x_{1,\secp}; s_1)$ and $\com(1^\secp, x_{2,\secp}; s_2)$ where $s_1, s_2$ are sampled uniformly cannot be distinguished by any polynomial-time quantum circuit with non-negligible advantage. 
\end{itemize}
Such a commitment scheme follows from the existence of quantum-secure injective one-way functions, or in the common-random-string model from any one-way function. There are also explicit constructions from post-quantum assumptions (see, e.g., ~\cite{jain2012commitments}). 

Commitment schemes like this can be used in the following manner: to commit to a message $x$, the sender samples a randomness $s$ and computes $c = \com(1^\secp,x; s)$, and sends $c$ to the receiver. To reveal the message $x$, the sender sends $(x,s)$ to the receiver. The receiver can verify that this is a valid commitment by checking that $c = \com(1^\secp,x; s)$. 

For more background on commitment schemes see, e.g.,~\cite{GoldreichBook1}.

\paragraph{Quantum Teleportation.} We recall the functionality of quantum teleportation. Let $\qrv{e} = (\qrv{e}_1,\qrv{e}_2)$ denote $m$ EPR pairs, where $\qrv{e}_1$ denotes all ``first-halves'' of the EPR pairs, and $\qrv{e}_2$ denotes all the ``second halves''. Suppose that Alice has $\qrv{e}_1$ along with an $m$-qubit quantum state $\qrv{w}$, and Bob has $\qrv{e}_2$. Alice can teleport $\qrv{w}$ to Bob by performing a measurement on $(\qrv{w},\qrv{e}_1)$ to obtain measurement outcomes $(u,v) \in (\binset^m)^2$. The post-measurement outcome on Bob's side is $X^{u_1} Z^{v_1} \otimes \cdots \otimes X^{u_m} Z^{v_m}(\qrv{w})$. Alice can then send $(u,v)$ to Bob so he can undo the Pauli corrections.

\subsection{Delayed-Input Zero-Knowledge $\Sigma$-Protocols} 

We now recall the definition of zero-knowledge $\Sigma$-protocols with the delayed-input property. We will use the following notation. Let $P,V$ be a pair of interactive machines (interpreted as prover and verifier respectively). We let $\langle P(y_p), V(y_v) \rangle (x)$ denote the output of the verifier $V$ after completing an interaction with $P$, in which $P$ has private input $y_p$, $V$ has private input $y_v$ and they also have a common input $x$ (the inputs $y_p, y_v$ can be classical or quantum).

A pair $(P, V)$ of a quantum polynomial-time ``honest'' prover $P$ and an ``honest'' verifier $V$ is a \emph{quantum interactive proof system} for $L$ if there exist numbers $\alpha$ (called the \emph{completeness}) and $\beta$ (called the \emph{soundness}) such that $\alpha > \beta$ such that
\begin{itemize}
	\item (\emph{Completeness}.) If $x \in L_{yes}$ and $\qrv{w} \in \cal{R}_L(x)$, then 
	$\Pr[\langle P(\qrv{w}),V \rangle(x) \text{ accepts}] \geq \alpha$. 	
	
	\item (\emph{(Statistical) Soundness}.) If $x \in L_{no}$, then for any prover $\malP$ (possibly computationally-unbounded) it holds that $\Pr[\langle \malP,V \rangle(x) \text{ accepts}] \leq \beta$.
\end{itemize}

\paragraph{Adaptive Soundness.}
For public-coin protocols (and in particular in our protocol) we can consider a stronger notion of \emph{adaptive} soundness, in which the instance $x$ is not specified ahead of time, but rather is produced by $\malP$ in the end of the interaction. In this case it will be convenient for us to specify that the output of the verifier includes the accept/reject bit also the instance $x$ produced by $\malP$. Under this convention, adaptive soundness is the requirement that for any adaptive adversary $\malP$ it holds that 
\[
\Pr[\langle \malP,V \rangle = (\text{accept},x) \land x \in L_{no}] \leq \beta~.
\]
Note that the winning event in this case is not necessarily efficiently recognizable.

\paragraph{Zero-Knowledge}
The proof system is furthermore \emph{(computational) zero-knowledge} if there exists a a quantum polynomial-time simulator $\ZKSim$ such that for any malicious verifier $\malV$ and for any asymptotic sequence of instances $x \in L_{yes}$, $\qrv{w} \in \cal{R}_L(x)$ and quantum $\poly(|x|)$-qubit auxiliary input $\qrv{y}$ it holds that the output state of $\ZKSim(\malV, x, \qrv{y})$ and $\langle P(\qrv{w}), \malV(\qrv{y}) \rangle(x)$ are computationally indistinguishable. We recall that the latter expression refers to the quantum output produced by $\malV$ at the end of the interaction.

An interactive protocol is a \emph{$\Sigma$-protocol} if it consists of a prover sending the first message $\qrv{m}_1$, the verifier sending a uniformly random (classical) challenge $m_2$, the prover sending a response $\qrv{m}_3$, and the verifier decides whether to accept or reject based on the transcript. We note that for an honest verifier the transcript is well defined even if the messages are quantum, since $m_2$ is a classical random string independent of $\qrv{m}_1$.

Finally, we say that a zero-knowledge $\Sigma$-protocol has the \emph{delayed-input property} if the prover only receives the instance $x$ and a quantum witness $\qrv{w}$ after it has sent the first message. In other words, the prover's first message $\qrv{m}_1$ is computed independently of the instance $x$ and witness $\qrv{w}$.

\subsection{Our Proof System}

We present a zero-knowledge $\Sigma$-protocol $\langle P, V \rangle$ for $L$ in Protocol~\ref{fig:zk-qma}. Note that the prover only gets the instance $x$ and witness $\qrv{w}$ right before generating the third message. 

\paragraph{Parameters and Definitions.}
We assume that at the beginning all parties know the instance length $n$. To avoid clutter we set the security parameter $\secp$ to be equal to $n$.
We let $m=m(n)$ denote the length of the witness required for instances of length $n$ of $L$. We let $F$ denote the following quantum circuit, that takes $n+2m$ classical bits (partitioned into strings $x,u,v$ of length $n,m,m$ respectively) and $m$ quantum bits as inputs. On input $(x,u,v,\qrv{\tilde{w}})$ the circuit does the following. It first computes $\qrv{w} = (X^{u_1} Z^{v_1} \otimes \cdots \otimes X^{u_m} Z^{v_m}) (\qrv{\tilde{w}})$, i.e.\ applies a quantum one-time pad, indicated by the vectors $u,v$ to the quantum input. It then executes $V_L(x,\qrv{w})$ and outputs the single-bit outcome of this computation.

\vspace{10pt}
\IncMargin{1em}
\begin{algorithm}
\DontPrintSemicolon

\textbf{Global Parameters:} Instance length $n$, witness length $m=m(n)$.
\BlankLine
\textbf{Prover}: \;
\Indp Generate $m$ EPR pairs $\qrv{e} = (\qrv{e}_1, \qrv{e}_2)$ for the purpose of quantum teleportation. \;

Sample a random string $r$ and execute the circuit $U_r$ on input $(\qrv{e}_2,0)$, where the $0$ represents a sufficient number of ancilla zeroes. Let $\Big (\hat{\qrv{F}}_0,  \{\lab_{i,b}\}_{i \in [n+2m], b \in \binset} \Big)$ denote the outcome of this computation. \;

 Sample random strings $s_{i,b}$ for all $i \in [n+2m]$, $b \in \binset$ and compute the commitment $c_{i,b} = \com(1^\secp, \lab_{i,b}; s_{i,b})$. Then, sample a random string $s_0$ and compute the commitment $c_0 = \com(1^\secp, r; s_0)$. \; \label{i:aftergc}

Send $\Big( \hat{\qrv{F}}_0,c_0,(c_{i,b})_{i \in [n+2m], b \in \binset} \Big)$ to the verifier.\;
\BlankLine
\Indm \textbf{Verifier}: \;
\Indp Send random bit $b$.
\BlankLine
\Indm \textbf{Prover} (given $x, \qrv{w})$: \;

\Indp If $b = 0$: 

\Indp Open all commitments, i.e.\ send $(r, s_0)$ and $(\lab_{i,b}, s_{i,b})_{i \in [n+2m], b \in \binset}$ to the verifier.
\;
\Indm 
If $b = 1$: 

\Indp Teleport the witness $\qrv{w}$ into the ``first halves'' $\qrv{e}_1$ of the EPR pairs $\qrv{e}$ to obtain classical strings $(u,v) \in \binset^m$. \; 

Consider the concatenated string $z = (x,u,v)$. Open the commitments corresponding to $z$, i.e.\ send $(z_i, \lab_{i,z_i}, s_{i,z_i})_{i \in [n+2m]}$ to the verifier.\label{i:sendopenings}
\;
\BlankLine
\Indm \Indm \textbf{Verifier} (given $x$): \;
\Indp If $b=0$:

\Indp Check that all commitments are valid. If any of them are invalid, then reject. \;

Apply the inverse circuit $U_r^{-1}$ to $(\hat{\qrv{F}}_0, \{\lab_{i,b}\}_{i \in [n+2m], b \in \binset})$, and let $(\qrv{e}_2',\qrv{q})$ denote the output. Check that $\qrv{q}$ is the all zeroes state. If so, then accept. Otherwise, reject. \;

\Indm If $b=1$:

\Indp  Check that $z=(x,u,v)$ for the instance $x$ and some $u,v$. If not then reject. \label{zk:check-x} \;

Check that all commitment openings are valid. If any of them are invalid, then reject. \label{zk:check-com-2} \; 

Decode the QRE $(\hat{\qrv{F}}_0, \{\lab_{i,z_i}\}_{i \in [n+2m]})$ to obtain a single-qubit output; return the output of this evaluation. \label{zk:eval-qre}

\caption{A zero-knowledge $\Sigma$-protocol for $\qma$ problem $L$}\label{fig:zk-qma}
\end{algorithm}\DecMargin{1em}
\vspace{10pt}

\begin{lemma}[Completeness] 
\label{lem:zk-completeness}
There exists a negligible function $\mu$ such that for all $x \in L_{yes}$ and $\qrv{w} \in \cal{R}_L(x)$, the honest prover described in Protocol~\ref{fig:zk-qma} runs in polynomial time given the input $(x,\qrv{w})$, and is accepted by the verifier with probability at least $1 - \mu(|x|)$.
\end{lemma}
\begin{proof}
If the prover behaves honestly, then with challenge $b = 0$, the verifier will be able to verify that the state $(\hat{\qrv{F}}_0, \{\lab_{i,b}\}_{i \in [n+2m], b \in \binset})$ is indeed equal to $U(\qrv{e}_2,0)$; and with challenge $b = 1$, the verifier obtains the output qubit of the circuit $V_L(x,\qrv{w})$ by applying the QRE decoding procedure (which we assume has perfect correctness). Thus if the maximum acceptance probability of $V_L(x,\qrv{w})$ over all choices of $\qrv{w}$ is $1 - \mu(|x|)$, then the acceptance probability of the verifier is $1 - \mu(|x|)$. This establishes the completeness property.
\end{proof}

We now prove statistical adaptive soundness for our protocol as stated in the next Lemma.

\begin{lemma}[Statistical Adaptive Soundness] 
	\label{lem:zk-soundness-half}
	There exists a negligible function $\mu'$ such that for any adaptive prover $\malP$, it holds that 
	\[
	\Pr[W] \leq 1/2 + \mu'(|x|)~,
	\]
	where $W$ is the event where $\langle \malP,V \rangle = (\text{accept},x) \land x \in L_{no}$.
\end{lemma}

In order to prove the lemma, we require the following two claims.

\begin{claim}\label{claim:triangle}
	Let $\ket{\phi}, \ket{\phi_0}, \ket{\phi_1}$ be unit vectors over some Hilbert space representing states of a quantum system. Assume there exist non-negative real values $\alpha_1, \alpha_2$ so that $\ket{\phi} = \alpha_1 \ket{\phi_0} + \alpha_2 \ket{\phi_1}$ (note that $\ket{\phi_0}, \ket{\phi_1}$ are not necessarily orthogonal, so $\alpha_1^2 +  \alpha_2^2$ is not necessarily $1$).

	Let $M$ be some measurement operator defined over this Hilbert space. Let $p, p_1, p_2$ be the probability that the measurement $M$ succeeds when the system is in state $\ket{\phi}, \ket{\phi_0}, \ket{\phi_1}$ respectively. Then $p \le (\alpha_1\sqrt{p_1} + \alpha_2\sqrt{p_2})^2$.
\end{claim}

\begin{proof}
	The proof follows by triangle inequality. We define $\ket{\tilde{\phi}_i} = M \ket{\phi_i}$ for $i\in\{0,1\}$ and note that by definition $p_i = \langle \tilde{\phi}_i | \tilde{\phi}_i \rangle$.
	
	\begin{align*}
		p & = \bra{\phi} M^\dagger M \ket{\phi}\\
		& = \alpha_1^2  \langle \tilde{\phi}_0 | \tilde{\phi}_0 \rangle + \alpha_2^2  \langle \tilde{\phi}_1 | \tilde{\phi}_1 \rangle + 
		\alpha_1 \alpha_2 (\langle \tilde{\phi}_0 | \tilde{\phi}_1 \rangle + \langle \tilde{\phi}_1 | \tilde{\phi}_0 \rangle)\\
		& \le \alpha_1^2 p_1 + \alpha_2^2 p_2 + 2\alpha_1\alpha_2 \sqrt{p_1 p_2}\\
		& = (\alpha_1\sqrt{p_1} + \alpha_2\sqrt{p_2})^2~.
	\end{align*}
The claim thus follows.
\end{proof}

The following claim establishes a very weak form of soundness, namely it asserts that soundness $1/2+\negl(n)$ holds when the $b=0$ test is guaranteed to pass.

\begin{claim}\label{claim:perfect0}
	For any adversarial strategy $\malP$ for which $\Pr[W | b=0]=1$ it holds that $\Pr[W | b=1]=\negl(n)$.
\end{claim}
\begin{proof}
	Assume that $\malP$ is such that $\Pr[W | b=0]=1$. Then the message $\qrv{m}_1$ is guaranteed to be a properly generated QRE of the correct functionality $F$. It follows by the correctness of the QRE that such a circuit cannot accept an instance $x \in L_{no}$ except with negligible probability. The claim thus follows. 
\end{proof}

We can now prove the soundness lemma.

\begin{proof}[Proof of Lemma~\ref{lem:zk-soundness-half}]
	Let $\malP$ be an adaptive possibly-cheating prover. Consider the point in time after $\malP$ sent its first message $\qrv{m}_1$ and let us consider the joint quantum state of $\malP$ and $\qrv{m}_1$, denote this state by $\ket{\phi}$. Assume without loss of generality that this state is pure (we can always add the purification of the state into the internal state of $\malP$. Assume without loss of generality that if $\malP$ sends valid commitments (i.e.\ ones that can be opened), then it indeed opens them correctly upon challenge $b=0$ (this can only increase the advantage of $\malP$ and since it is computationally unbounded it can always find the openings).
	
	We consider the (inefficient) measurement operator defined on this joint state in the following way.
	\begin{enumerate}
		\item Check (via brute force search) that all commitments in $\qrv{m}_1$ are valid. If they are valid, let $r, \lab_{i,b}$ be their openings. If not then reject. Note that this is a projection since it simply accepts a subset of the possible commitments (in the computational basis).
		
		\item Apply $U_r^{-1}$ on $\Big ( \hat{\qrv{F}}_0, \{\lab_{i,b}\}_{i,b} \Big )$, where $\hat{\qrv{F}}_0$ comes from $\qrv{m}_1$ and $r, \lab_{i,b}$ are the openings of the commitments (which at this point are well defined). The outcome is a pair $(\qrv{q}, \qrv{e_2})$. Accept if $\qrv{q}=0$ and reject otherwise. Note that this part is a projection as well since it only involves applying a unitary followed by accepting a value in the computational basis.
	\end{enumerate}
We conclude that this measurement operator is a projection $\Pi$, and notice that this measurement corresponds to the event $W | b=0$. Therefore
\[
\Pr[W | b=0] = \bra{\phi} \Pi \ket{\phi}~,
\]
and denote this value by $\epsilon$. If $\epsilon=0$ then the proof is complete because the overall acceptance probability is at most $1/2$. From now on assume $\epsilon >0$. Note that since $\Pi$ is a projection it holds that
\[
\bra{\phi} (I - \Pi) \ket{\phi} = 1- \epsilon~.
\]

Let $M$ denote the (not necessarily projective and inefficient) measurement that acts on the joint state of $\malP$ and $\qrv{m}_1$ as follows. It acts on the state of $\malP$ to generate the third message $m_3$ and instance $x$, and then checks whether $x \in L_{no}$ and if so whether the verifier $V$ accepts. Again by definition it holds that
\[
\Pr[W | b=1] = \bra{\phi} M^\dagger M \ket{\phi}~.
\]

Let us now define $\ket{\phi_0} = \Pi \ket{\phi} / \sqrt{\epsilon}$ and $\ket{\phi_1} = (I- \Pi) \ket{\phi} / \sqrt{1-\epsilon}$, so we can write $\ket{\phi} = \sqrt{\epsilon} \ket{\phi_0} + \sqrt{1-\epsilon} \ket{\phi_1}$. By Claim~\ref{claim:triangle} (triangle inequality) we have
\begin{align}
\Pr[W | b=1] & \le \left( \sqrt{\epsilon \bra{\phi_0} M^\dagger M \ket{\phi_0}} +  \sqrt{(1-\epsilon) \bra{\phi_1} M^\dagger M \ket{\phi_1}} \right)^2\\
& \le \left( \sqrt{\epsilon \bra{\phi_0} M^\dagger M \ket{\phi_0}} +  \sqrt{1-\epsilon}\right)^2~.\label{eq:aftertriangle}
\end{align}

However, suppose that instead of starting with the joint prover-message state $\ket{\phi}$, we consider a different prover $\tilde{\malP}$ which starts with the state $\ket{\phi_0}$ instead. By construction this prover will pass the challenge $b = 0$ with probability $1$, because $\Pi \ket{\phi_0} = \ket{\phi_0}$. 
Therefore, by Claim~\ref{claim:perfect0} it follows that
\[
\Pr[W_{\tilde{\malP}} | b=1]= \bra{\phi_0} M^\dagger M \ket{\phi_0} \leq \mu(n)~,
\]
for some negligible function $\mu$, because $M$ corresponds to the test performed on challenge $b = 1$. 

Finally we can plug $\bra{\phi_0} M^\dagger M \ket{\phi_0} \leq \mu(n)$ into Eq.~\eqref{eq:aftertriangle} to obtain
\begin{align*}
	\Pr[W | b=1] & \le \left( \sqrt{\epsilon \bra{\phi_0} M^\dagger M \ket{\phi_0}} +  \sqrt{1-\epsilon}\right)^2 \\
	& \le \left( \sqrt{\epsilon \cdot \mu(n)}  +  \sqrt{1-\epsilon}\right)^2\\
		& = 1-\epsilon + 2\mu'(n)~,
\end{align*}
for some $\mu'(n) = O(\sqrt{\mu(n)})$.

Finally we can conclude that 
\[
	\Pr[W] \le \frac{1}{2} \epsilon  + \frac{1}{2} \left ( 1-\epsilon + 2\mu'(n) \right ) \leq \frac{1}{2} +\mu'(n)
\]
which completes the proof of the Lemma.
\end{proof}

\begin{lemma}[Computational zero-knowledge] 
\label{lem:zk-zk}
There exists a quantum polynomial-time simulator $\ZKSim$ satisfying the following: for all cheating verifiers $\malV$, for all asymptotic sequences of $(x,\qrv{y},\qrv{w})$ where $x \in L_{yes}$, the state $\qrv{y}$ is an arbitrary quantum state, and $\qrv{w} \in \cal{R}_L(x)$ is a witness for $x$, we have that the output of $\ZKSim(\malV,x,\qrv{y})$ is computationally indistinguishable from $\langle P(\qrv{w}), \malV(\qrv{y}) \rangle(x)$.

\end{lemma}

Let $\malv$ denote the malicious verifier. In order to prove \Cref{lem:zk-zk} we first analyze a ``conditional'' simulator $\ZKSim^0$ that takes as input $(\malV,x,\qrv{y})$, and outputs a quantum state as well as a flag indicating whether the simulator aborted. We show that the output of $\ZKSim^0$ is, conditioned on not aborting, computationally indistinguishable from the view of the interaction $\langle P(\qrv{w}), \malV(\qrv{y}) \rangle(x)$. The probability of aborting in the conditional simulator $\ZKSim^0$ is (negligibly close to) $1/2$, but we can use Watrous's Rewinding Lemma~\cite{watrous2009zero} to argue the existence of a quantum polynomial-time algorithm $\ZKSim$ that satisfies the conclusions of \Cref{lem:zk-zk}. In particular, we use the formulation of the Rewinding Lemma as presented in~\cite[Lemma 2.1]{bitansky2020post}; the resulting algorithm $\ZKSim$ queries $\ZKSim^0$ as a blackbox polynomially many times to amplify the success probability. The reason we need to use the Rewinding Lemma instead of simply just repeating $\ZKSim^0$ until it doesn't abort is because of the quantum auxiliary input $\qrv{y}$; running $\ZKSim^0$ once and aborting may alter the state $\qrv{y}$ significantly. Watrous's Rewinding Lemma gets around this issue.

\SetAlgorithmName{Algorithm}{Algorithm} 

\vspace{10pt}
\IncMargin{1em}
\begin{algorithm}
\DontPrintSemicolon
\textbf{Input}: cheating verifier $\malV$, instance $x \in \binset^n$, auxiliary quantum state $\qrv{y}$ \;
\BlankLine
Sample $t \in \{0,1\}$ uniformly at random. \;
\If{$t = 0$}{

Execute the honest prover $P$ to generate the first message $\qrv{m}_1$ (note $P$ needs no input for this). \; 

Run the cheating verifier $\malV(x,\qrv{y})$ on the first message $\qrv{m}_1$ to generate the challenge bit $b$. If $b \neq 0$, abort (i.e., output $(a,0)$ where $a = 1$). \;

Otherwise, continue simulating the honest prover $P$ on challenge $b = 0$ to generate the third message $\qrv{m}_3$. \;

}
\Else{

	Generate $m$ EPR pairs $\qrv{e} = (\qrv{e}_1,\qrv{e}_2)$. 
	
	Sample uniformly random $u,v\in\binset^m$.\label{simstep:uv}
	
	Run the simulator $\Sim$ of the $\QRE$ on input $\ketbra{1}{1}$, with respect to the same circuit topology as $F$, to obtain an output $\Big ( \hat{\qrv{F}}_0, \{\lab_i\}_{i \in [n+2m]} \Big )$. \label{simstep:qresim}
	
	For all $i \in [n+2m]$,$b\in\binset$, define $\lab_{i,b} = \lab_i$. Set $r=0$.\label{simstep:fakecom}

	Execute the honest prover $P$, starting at Step \ref{i:aftergc} of Protocol~\ref{fig:zk-qma}, where the prover computes the commitments to $r$ and $(\lab_{i,b})_{i,b}$. Let $\qrv{m}_1$ denote the prover's first message. \;
	
	Run the cheating verifier $\malV(x,\qrv{y})$ on the first message $\qrv{m}_1$ to generate the challenge bit $b$. If $b \neq 1$, abort (i.e., output $(a,0)$ where $a = 1$). \;

	Otherwise, execute Step~\ref{i:sendopenings} of the honest prover $P$ in Protocol~\ref{fig:zk-qma} to open the relevant commitments, which forms message $\qrv{m}_3$. Note that $\qrv{w}$ is note used by $P$ in this step, therefore the simulator can perform it. \label{simstep:continue}

}

Continue simulating the cheating verifier $\malV$ on the third message $\qrv{m}_3$ to obtain a state $\qrv{o}$, and output $(a,\qrv{o})$ where $a = 0$.

\caption{The simulator $\ZKSim^0$ for the zero-knowledge protocol}\label{fig:zk-sim}
\end{algorithm}\DecMargin{1em}
\vspace{10pt}

\begin{lemma}
\label{lem:zk-zk-0}
There exists a quantum polynomial-time conditional simulator $\ZKSim^0$ satisfying the following: for all cheating verifiers $\malV$, for all asymptotic sequences of $(x,\qrv{y},\qrv{w})$ where $x \in L_{yes}$, the state $\qrv{y}$ is an arbitrary quantum state, and $\qrv{w} \in \cal{R}_L(x)$ is a witness for $x$, we have that the output of $\ZKSim^0(\malV,x,\qrv{y})$, conditioned on not aborting, is computationally indistinguishable from $\langle P(\qrv{w}), \malV(\qrv{y}) \rangle(x)$.
\end{lemma}

\begin{proof}
The proof proceeds by a sequence of experiments (or hybrids). We keep track of the output distribution of the simulator and the rejection probability across experiments. Fix $\malV$, $\qrv{y}$, $x$, $\qrv{w}$ as in the Lemma statement.

\begin{itemize}
	\item\textbf{Experiment $0$.} This is the experiment of running the conditional simulator $\ZKSim^0$ on input $(\malV,x,\qrv{y})$ as presented in Protocol~\ref{fig:zk-sim}.
	
\item\textbf{Experiment $1$.} Modify the previous experiment by replacing Step~\ref{simstep:uv} with the following:
teleport $\qrv{w}$ into $\qrv{e_1}$ and let $u,v$ be the (classical) outcome of the teleportation. By the properties of the teleportation the strings $u,v$ are uniformly random and therefore this experiment produces an identical distribution to the previous one.

\item\textbf{Experiment $2$.} Now, consider Experiment $1$ except we modify Step~\ref{simstep:qresim} of $\ZKSim^0$. Instead of using the QRE simulator $\Sim$, do the following: sample randomness $r^*$ and evaluate the unitary $U_{r^*}$ on input $(\qrv{e_2},0)$ to generate the state $\Big( \qrv{\hat{F}}_0, (\lab^*_{i,b})_{i \in [n+2m], b \in \binset} \Big)$. 
Set $\lab_i$ to be $\lab^*_{i,z_i}$ for $z=(x,u,v)$. By definition $F(x,u,v,\qrv{e_2})$ computes $V_L(x, \qrv{w})$, which by the completeness of the $\qma$ verifier outputs $\ketbra{1}{1}$ with probability $1-\negl(n)$. 

By the privacy of $\QRE$, the QRE encoding of circuit $F(x,u,v,\qrv{e_2})$ is computationally indistinguishable from $\Sim(\ketbra{1}{1})$ when we marginalize over the randomness $r^*$ and the labels $(\lab_{i,b} : b \neq z_i)$. Since the randomness $r^*$ and the unused labels are not used anywhere else in the experiment, the output distribution of this experiment is computationally indistinguishable from that of the previous one. 

\item\textbf{Experiment $3$.} Change Step~\ref{simstep:fakecom} of $\ZKSim^0$ so that $r=r^*$ and $\lab_{i,b} = \lab^*_{i,b}$. Since this only changes locations of the commitment that are never opened by the prover $P$, the hiding property of the commitment scheme guarantees that the views of $\malV$ between Experiment 2 and Experiment 3 remain computationally indistinguishable. 

\item\textbf{Experiment $4$.} Move the modified Step~\ref{simstep:uv} (i.e. the teleportation of $\qrv{w}$ into the EPR pairs) to be right before Step~\ref{simstep:continue}. This does not change anything in the simulation because none of the steps until Step~\ref{simstep:continue} in Experiment $3$ depend on teleportation outcomes $(u,v)$.

\item\textbf{Experiment $5$.} Note that in Experiment~$4$, the steps up to and including receiving the bit $b$ from $\malV$ are identical between $t=0$ and $t=1$. Thus we can move these steps outside of the ``if'' statement, and before the sampling of $t$. The output of the Experiment is unchanged from Experiment $4$.
\end{itemize}

We see that the output in Experiment~$5$ is computationally indistinguishable from that of $\ZKSim^0(x)$. Furthermore, in Experiment $5$ the bit $t$ is sampled after receiving the bit $b$ and abort occurs if and only if $t \neq b$. It follows that the abort probability is $1/2$ in Experiment~$5$. Furthermore, conditioned on not aborting, the experiment is identical to that of the execution of $\malV$ with $P(x, \qrv{w})$. 

It follows that Experiment~$0$ (which is to run $\ZKSim^0$ on input $(V^*,x,\qrv{y})$), the probability of abort is negligibly close to $1/2$ and the output conditioned on not aborting is indistinguishable from $\langle P(\qrv{w}), \malV(\qrv{y}) \rangle(x)$.
\end{proof}

 \section{Quantum Garbled Circuits -- Construction}
\label{sec:construction}

\newcommand{\topol}{\mathfrak{T}}

In this section we present topologically-universal QRE scheme of \Cref{lem:main}, called the Quantum Garbled Circuits scheme and denoted by $\QGC$.
In particular, given a circuit that computes a quantum operation $F$, we show how to compute the encoding $\hF$ in \Cref{sec:circuit-encoding}, how to decode the encoded value in \Cref{sec:ckt-eval}, and how to simulate the randomized encoding in \Cref{sec:simulator}. We then prove the correctness and security of the scheme in \Cref{sec:analysis}.

In this paper we assume that quantum circuits $F$ being encoded use the universal gate set $\cal{C}_2 \cup \{T\}$. The only property of this gate set we use (other that the arity of the gates being bounded by a global constant) is the following: each $p$-qubit gate $U_g$ of circuit $F$ has the property that for any single-qubit Pauli unitaries $P_1,\ldots,P_p$, there exist single-qubit gates $R_1,\ldots,R_p$ from the PX group such that
\begin{equation}
\label{eq:pauli-to-clifford}
	U_g (P_1 \otimes P_2 \otimes \cdots \otimes P_p) = (R_1 \otimes R_2 \otimes \cdots \otimes R_p) U_g.
\end{equation}
This property indeed holds for the $\cal{C}_2 \cup \{T\}$ universal set as described in Sections~\ref{sec:prelim:pauliclifford} and~\ref{sec:prelim:universal}.

\paragraph{A Building Block: Topologically-Universal Decomposable RE for Classical Circuits.}
As stated in \Cref{lem:main}, we assume the existence of an efficient topologically-universal and label-universal DCRE for classical circuits. We refer to this DCRE scheme as $\CRE = (\CEnc,\CDec,\CSim)$. In our construction we use $\CRE$ as a generic building block, and different instantiations of $\CRE$ will result in quantum encoding scheme with different properties. As in the Lemma statement, we let $\kappa(\cdot,\cdot)$ and $c(\cdot, \cdot)$ be such that for functions $f$ computable by size-$s$ and depth-$d$ classical circuits, the complexity of encoding $f$ is $c(d,s)$ and the length of the labels is $\kappa(d,s)$. We let $\CSim_{\topol}$ and $\CDec_{\topol}$ denote the polynomial-time simulator and decoding procedures of $\CRE$, respectively, for classical circuits with topology ${\topol}$.

\subsection{Gadgets}
\label{sec:gadgets}

In this section we introduce various gadgets that are used in our $\QGC$ scheme.

\subsubsection{Teleportation Gadget}
\label{sec:tp-gadget}
Let $\ell = (\ell_{b,a})_{a \in \binset, b \in \xz}$ be a vector of strings of length $\kappa$, and let 
$s = (s_z,s_x), t = (t_z,t_x) \in \binset^2$.  Define $\sf{TP}_{\ell,s,t}$ to be the unitary computed by the following circuit:
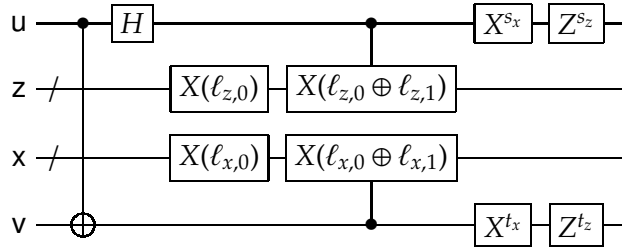
\begin{figure}[H]
  \centering
  $
        \Qcircuit @C=0.6em @R=0.8em {
          \reg{u}  & & \qw & \ctrl{3} & \gate{H}  & \qw & \ctrl{1} & \gate{X^{s_x}} & \gate{Z^{s_z}} & \qw \\
          \reg{z} & & {/} \qw  & \qw & \qw & \gate{X(\ell_{z,0})} & \gate{X(\ell_{z,0} \oplus \ell_{z,1})} & \qw & \qw & \qw  \\
          \reg{x}  & & {/} \qw & \qw & \qw     & \gate{X(\ell_{x,0})} & \gate{X(\ell_{x,0} \oplus \ell_{x,1})} & \qw & \qw & \qw  \\
          \reg{v} & & \qw & \targ & \qw  & \qw & \ctrl{-1} & \gate{X^{t_x}} & \gate{Z^{t_z}} & \qw \\          
        }
        $
    \caption{Teleportation gadget}
    \label{fig:teleportation-gadget}
\end{figure}
Here, for a string $r \in \{0,1\}^\kappa$ the notation $X(r)$ denotes applying the tensor product of $X$ gates acting on the $i$-th qubit whenever $r_i = 1$, and identity otherwise. The controlled $X(\ell_{z,0} \oplus \ell_{z,1})$ and $X(\ell_{x,0} \oplus \ell_{x,1})$ gates can also be seen as \emph{fan-out} gates that are applied to the qubits indexed by $\ell_{z,0} \oplus \ell_{z,1}$ and $\ell_{x,0} \oplus \ell_{x,1}$, respectively, because it copies the control qubit into the target qubits simultaneously. Thus the teleportation gadget is a $\qnczf$ circuit. We refer to $\ell$ as the \emph{teleportation labels} and $s,t$ as the \emph{randomization bits} of the \emph{teleportation gadget} $\sf{TP}_{\ell,s,t}$.

\begin{lemma}
\label{lem:tp-gadget}
Let $\reg{u},\reg{v},\reg{u}'$ denote qubit registers, and let $\reg{z},\reg{x}$ denote ancilla registers. For all $\ell, s,t$ and for all qubit states $\ket{\psi}$ we have 
\[
	\sf{TP}_{\ell,s,t} \, \ubket{\psi}{\reg{u}} \otimes \ubket{0,0}{\reg{z} \reg{x}} \otimes \ubket{\epr}{\reg{v} \reg{u}'}  = \frac{1}{2} \sum_{d,e \in \binset} \ub{Z^{s_z} X^{s_x} \ket{d}}{\reg{u}} \otimes \ub{Z^{t_z} X^{t_x} \ket{e}}{\reg{v}} \otimes \ubket{\ell_{z,d}, \ell_{x,e}}{\reg{z} \reg{x}} \otimes \ub{X^e Z^d \ket{\psi}}{\reg{u}'}
\]
\end{lemma}

\begin{proof}
	Since $\sf{TP}_{\ell,s,t}$ is unitary, it suffices to prove the Lemma when $\ket{\psi} = \ket{c}$ is a standard basis state.
	After the first CNOT, Hadamard, and $X(\ell_{z,0})$ and $X(\ell_{x,0})$ we have
	\[
		\frac{1}{2} \sum_{d,a \in \binset} (-1)^{dc} \ubket{d}{\reg{u}} \otimes \ubket{\ell_{z,0}}{\reg{z}} \otimes \ubket{\ell_{x,0}}{\reg{x}} \otimes \ubket{ a\oplus c, a}{\reg{v} \reg{u}'} \;.
	\]
	After the controlled $X$'s we have
	\[
		\frac{1}{2} \sum_{d,a \in \binset} (-1)^{dc} \ubket{d}{\reg{u}} \otimes \ubket{\ell_{z,d}}{\reg{z}} \otimes \ubket{\ell_{x,a \oplus c}}{\reg{x}} \otimes \ubket{ a\oplus c, a}{\reg{v} \reg{u}'}  \;.
	\]	
	Relabeling $e = a \oplus c$ and re-arranging, we get
	\[
		\frac{1}{2} \sum_{d,e \in \binset} \ub{\ket{d}}{\reg{u}} \otimes \ub{\ket{e}}{\reg{v}} \otimes \ubket{\ell_{z,d}, \ell_{x,e}}{\reg{z} \reg{x}} \otimes \ub{X^e Z^d \ket{c}}{\reg{u}'} \;.
	\]
	Applying the $X^{s_x}, X^{t_x}, Z^{s_z}, Z^{t_z}$ gates, we obtain the desired Lemma statement.
	
\end{proof}

\subsubsection{Correction Gadget}
\label{sec:corr-gadget}

\newcommand{\regu}{\reg{u}}
\newcommand{\regv}{\reg{v}}
\newcommand{\rega}{\reg{a}}
\newcommand{\regz}{\reg{z}}
\newcommand{\regx}{\reg{x}}
\newcommand{\regb}{\reg{b}}

In our $\QGC$ scheme, the encoding consists of a collection of EPR pairs that are connected via gates of the circuit $F$ and teleportation gadgets as described above. However, the teleportation gadget induces a \emph{correction} on the output (as demonstrated by \Cref{lem:tp-gadget}) that, ostensibly, needs to be fixed before applying the next gate and teleportation operation -- but since the encoding is performed in parallel, it is up to the evaluator to perform the corrections \emph{after} the gates and teleportation gadgets are appled. The next Lemma demonstrates that the desired operation (correction, then teleportation) is equivalent a sequence of circuits $\Lambda_1, \Lambda_2, \Lambda_3$ where the encoder can apply $\Lambda_1$, and the decoder can apply $\Lambda_2$ and $\Lambda_3$. Before stating the Lemma, we first have to describe the \emph{randomization group}.

\paragraph{Randomization Group.} In the following Lemma, the circuits $\Lambda_1,\Lambda_2,\Lambda_3$ will act on the same registers that the teleportation gadget acts on (single-qubit registers $\reg{u}, \reg{v}$ and $\kappa$-qubit registers $\reg{z}, \reg{x}$), as well as ancilla registers $\regb$ that is $(\kappa+1)^2$ qubits wide. The individual qubits of registers $\regb$ are indexed by $(i,j) \in \{0,1,\ldots,\kappa\}^2$. An important conclusion of the Lemma is that the circuit $\Lambda_2$ computes a unitary belonging to the \emph{randomization group} $\scr{R}_\kappa$, which consists of depth-one circuits that are tensor products of
\begin{itemize}
	\item Two-qubit Clifford gates acting on the pair of qubits $(\regb_{ij},\regb_{ji})$ for all $i < j$, and
	\item Single-qubit Clifford gates acting on all other qubits.
\end{itemize}
The set $\scr{R}_\kappa$ indeed forms a group under the natural gate multiplication operation. It has finite order with $\exp(O(\kappa^2))$ elements, and a uniformly random element of $\scr{R}_\kappa$ can be sampled via an $\ncz$ circuit that is given a uniformly random bitstring as input (essentially, the single- and two-qubit Clifford gates are chosen independently in parallel).

\begin{lemma}
\label{lem:commute-correction}
Let $\ell$ be $\kappa$-bit teleportation labels, and let $s,t \in \binset^2$ be randomization bits. Let $R$ be an element from the single-qubit PX group. Then there exist
\begin{itemize}
	\item $\qnczf$ circuits $\Lambda_1(\ell)$, $\Lambda_3$ and 
	\item A depth-one Clifford circuit $\Lambda_2(R,\ell,s,t)$ that computes a unitary in $\scr{R}_\kappa$
\end{itemize}
all acting on registers $\regu, \regz, \regx, \regv, \regb$ such that the following circuit identity holds:
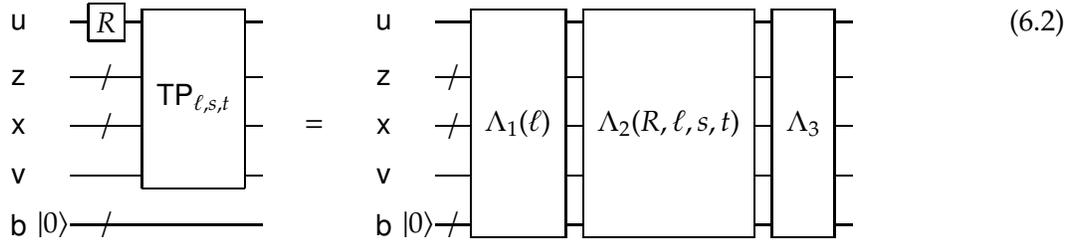
\begin{equation}
\label{eq:commute-correction}
        \Qcircuit @C=0.6em @R=0.8em {
          \reg{u}  & & & & \gate{R} & \multigate{3}{\sf{TP}_{\ell,s,t}} & \qw &  	& && &	 & & \reg{u} & &  & & \qw & \multigate{4}{\Lambda_1(\ell)}  & \multigate{4}{\Lambda_2(R,\ell,s,t)} &\multigate{4}{\Lambda_3} & \qw \\
          \regz &  & & & {/} \qw & \ghost{\sf{TP}_{\ell,s,t}} & \qw  			   & && 	& &		& & \regz & &   & & {/}\qw  & \ghost{\Lambda_1(\ell)} 		& \ghost{\Lambda_2(R,\ell,s,t)} & \ghost{\Lambda_3} & \qw \\
          \regx  &  & & & {/} \qw & \ghost{\sf{TP}_{\ell,s,t}}  & \qw &  && = & && & \regx & & & & {/} \qw & \ghost{\Lambda_1(\ell)}	& \ghost{\Lambda_2(R,\ell,s,t)} &   \ghost{\Lambda_3} & \qw \\
          \regv & &  & &  \qw & \ghost{\sf{TP}_{\ell,s,t}} & \qw  & &&& &	&& \regv & & & & \qw & \ghost{\Lambda_1(\ell)} & \ghost{\Lambda_2(R,\ell,s,t)} & \ghost{\Lambda_3} & \qw & \\
		 \regb & & \ket{0} & & {/} \qw & \qw & \qw  & 		& &&&		&& \regb & & \ket{0} & & {/} \qw & \ghost{\Lambda_1(\ell)} & \ghost{\Lambda_2(R,\ell,s,t)} & \ghost{\Lambda_3} & \qw & 
        }
\end{equation}
Furthermore, the description of $\Lambda_2(R,\ell,s,t)$ can be computed by a $\ncz$ circuit of size $O(\kappa^2)$.
\end{lemma}

\begin{proof}

We can write the left-hand side of~\eqref{eq:commute-correction} (omitting the ancilla registers $\regb$) as
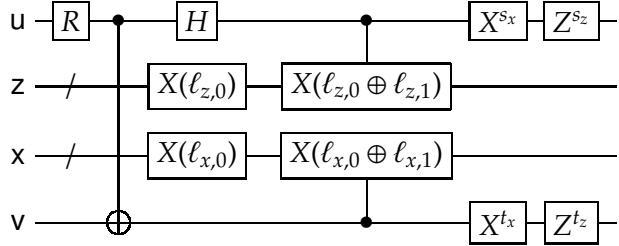
\[
        \Qcircuit @C=0.6em @R=0.8em {
          \reg{u}  & & \gate{R} & \ctrl{3} & \gate{H}  & \qw & \ctrl{1} & \gate{X^{s_x}}  & \gate{Z^{s_z}} & \qw \\
          \reg{z} &  & {/} \qw  & \qw & \gate{X(\ell_{z,0})} & \qw  & \gate{X(\ell_{z,0} \oplus \ell_{z,1})} &  \qw  & \qw & \qw \\
          \reg{x}  & & {/} \qw & \qw   & \gate{X(\ell_{x,0})} & \qw   & \gate{X(\ell_{x,0} \oplus \ell_{x,1})} & \qw  & \qw & \qw \\
          \reg{v} & & \qw & \targ & \qw  & \qw & \ctrl{-1} & \gate{X^{t_x}}  & \gate{Z^{t_z}}  & \qw \\          
        }
\]
Since $R$ is a PX group element, we can propagate it past the first CNOT gate to get the equivalent circuit
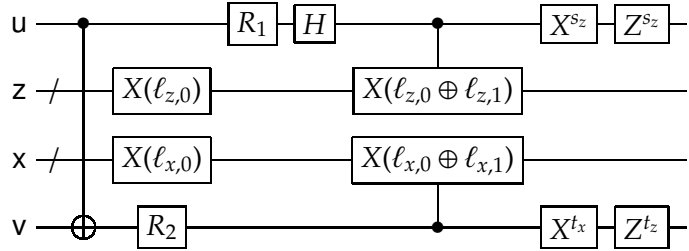
\[
        \Qcircuit @C=0.6em @R=0.8em {
          \reg{u}  & & \qw & \ctrl{3} & \qw & \gate{R_1}  & \gate{H} & \ctrl{1} & \gate{X^{s_z}}  & \gate{Z^{s_z}} & \qw \\
          \reg{z} &  & {/} \qw  & \qw & \gate{X(\ell_{z,0})} & \qw & \qw  & \gate{X(\ell_{z,0} \oplus \ell_{z,1})} &  \qw  & \qw & \qw \\
          \reg{x}  & & {/} \qw & \qw   & \gate{X(\ell_{x,0})} & \qw & \qw   & \gate{X(\ell_{x,0} \oplus \ell_{x,1})} & \qw  & \qw & \qw \\
          \reg{v} & & \qw & \targ & \gate{R_2} & \qw  & \qw & \ctrl{-1} & \gate{X^{t_x}}  & \gate{Z^{t_z}}  & \qw \\          
        }
\]
for some single-qubit gates $R_1$ and $R_2$ that are PX group elements. Propagating $R_2$ past the bottom fan-out operation yields another equivalent circuit
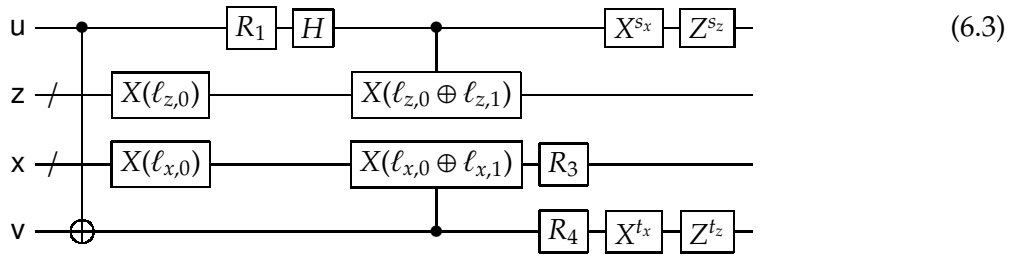
\begin{equation}
\label{eq:commute-correction-2}
        \Qcircuit @C=0.6em @R=0.8em {
          \reg{u}  & & \qw & \ctrl{3} & \qw & \gate{R_1}  & \gate{H} & \ctrl{1} & \qw & \gate{X^{s_x}}  & \gate{Z^{s_z}} & \qw \\
          \reg{z} &  & {/} \qw  & \qw & \gate{X(\ell_{z,0})} & \qw & \qw  & \gate{X(\ell_{z,0} \oplus \ell_{z,1})} &  \qw  & \qw & \qw & \qw \\
          \reg{x}  & & {/} \qw & \qw   & \gate{X(\ell_{x,0})} & \qw & \qw   & \gate{X(\ell_{x,0} \oplus \ell_{x,1})} & \gate{R_3}  & \qw & \qw & \qw  \\
          \reg{v} & & \qw & \targ & \qw & \qw  & \qw & \ctrl{-1} & \gate{R_4} & \gate{X^{t_x}}  & \gate{Z^{t_z}}  & \qw \\          
        }
\end{equation}
where $R_3$ is a single-qubit PX group element, and $R_4$ is a tensor product of single-qubit PX group elements. 

Next, we use the following Proposition in order to ``push'' the $R_1$ gate through the circuit as far to the right as possible:

\begin{restatable}{proposition}{ph}
\label{prop:p-h}
For all single-qubit PX gates $R$, $s = (s_z,s_x) \in \{0,1\}^2$, and strings $r \in \binset^\kappa$ there exist $\qnczf$ circuits $C_1(r)$ and $C_3$ and a depth-one Clifford circuit $C_2(R,r,s)$ in the randomization group $\scr{R}_\kappa$ such that the following circuit identity holds:
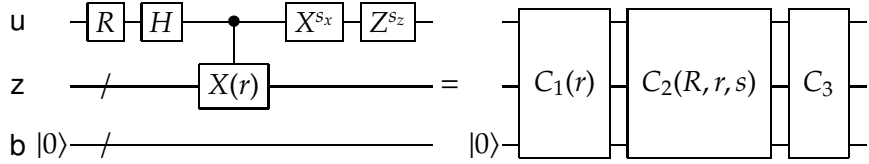
\[
        \Qcircuit @C=0.6em @R=0.8em {
          \reg{u} & &  & & \gate{R}  & \gate{H} & \ctrl{1} & \gate{X^{s_x}}  & \gate{Z^{s_z}} & \qw &		& & & & \multigate{2}{C_1(r)} & \multigate{2}{C_2(R,r,s)} & \multigate{2}{C_3} & \qw \\
          \reg{z} & & &  & {/} \qw & \qw  & \gate{X(r)} &  \qw & \qw & \qw & =   & & & &  \ghost{C_1(r)} & \ghost{C_2(R,r,s)} & \ghost{C_3} & \qw \\
          \regb && \ket{0}  & & {/} \qw & \qw   & \qw & \qw  & \qw & \qw  &			& & \ket{0}& & 		\ghost{C_1(r)} & \ghost{C_2(R,r,s)} & \ghost{C_3} & \qw         
        }
\]
Here, the register $\reg{b}$ consists of $(\kappa+1)^2$ qubits. Furthermore, descriptions of the circuits $C_1(r)$ and $C_2(R,r,s)$ can be computed by $\ncz$ circuits of size $O(\kappa^2)$.
\end{restatable}
The proof of this Proposition is deferred to \Cref{sec:p-h-proof}. Letting $r = \ell_{z,0} \oplus \ell_{z,1}$, we thus get that Circuit~\eqref{eq:commute-correction-2} is equivalent to 
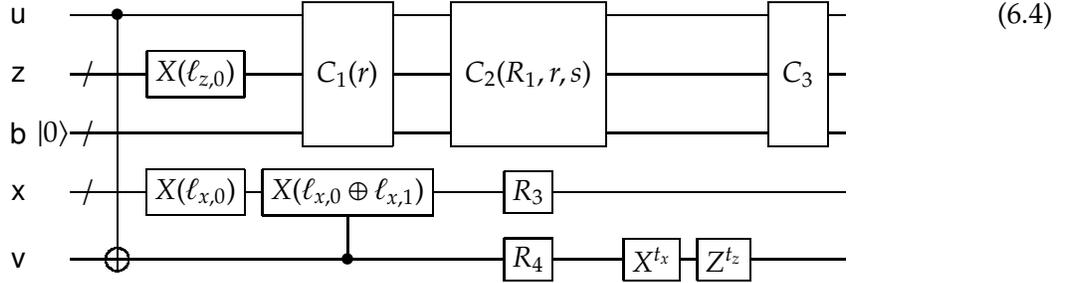
\begin{equation}
\label{eq:commute-correction-3}
        \Qcircuit @C=0.6em @R=0.8em {
          \reg{u}  & & & & \qw & \ctrl{4} & \qw & \multigate{2}{C_1(r)}  & \multigate{2}{C_2(R_1,r,s)} & \qw & \qw  & \multigate{2}{C_3} & \qw \\
          \reg{z} &  & & & {/} \qw  & \qw & \gate{X(\ell_{z,0})}  & \ghost{C_1(r)} & \ghost{C_2(R_1,r,s)}  & \qw &  \qw  & \ghost{C_3} & \qw \\
          \regb & & \ket{0}& & {/} \qw  & \qw & \qw & \ghost{C_1(r)} & \ghost{C_2(R_1,r,s)} &	\qw & \qw & \ghost{C_3} & \qw   \\
          \reg{x}  & & & & {/} \qw & \qw   & \gate{X(\ell_{x,0})} &  \gate{X(\ell_{x,0} \oplus \ell_{x,1})} & \gate{R_3}  & \qw & \qw & \qw & \qw  \\
          \reg{v} & & & & \qw & \targ & \qw & \ctrl{-1} & \gate{R_4} & \gate{X^{t_x}}  & \gate{Z^{t_z}}  & \qw         & \qw
        }
\end{equation}
Define $\Lambda_1(\ell)$ to be the circuit
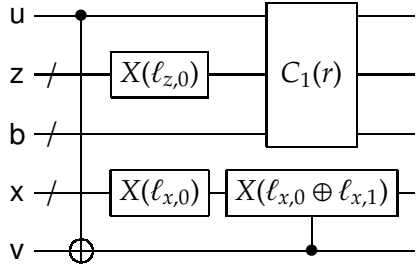
\[
        \Qcircuit @C=0.6em @R=0.8em {
          \reg{u}  & & \qw & \ctrl{4} & \qw & \multigate{2}{C_1(r)}  & \qw \\
          \reg{z}  & & {/} \qw  & \qw & \gate{X(\ell_{z,0})}  & \ghost{C_1(r)} & \qw \\
          \regb  & & {/} \qw  & \qw & \qw & \ghost{C_1(r)} & \qw   \\
          \reg{x}   & & {/} \qw & \qw   & \gate{X(\ell_{x,0})} &  \gate{X(\ell_{x,0} \oplus \ell_{x,1})} & \qw  \\
          \reg{v}  & & \qw & \targ & \qw & \ctrl{-1} & \qw
        }
\]
Since $C_1(r)$ is a $\qnczf$ circuit, so is $\Lambda_1(\ell)$. 

Define $\Lambda_2(R,\ell,s,t)$ to be the circuit
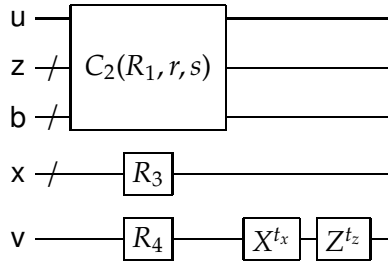
\[
        \Qcircuit @C=0.6em @R=0.8em {
          \reg{u}   & & \qw & \multigate{2}{C_2(R_1,r,s)} & \qw & \qw   & \qw \\
          \reg{z} & & {/} \qw  & \ghost{C_2(R_1,r,s)}  & \qw &  \qw  & \qw \\
          \regb & & {/} \qw  & \ghost{C_2(R_1,r,s)} &	\qw & \qw & \qw   \\
          \reg{x}   & & {/} \qw & \gate{R_3}  & \qw & \qw & \qw  \\
          \reg{v}  & & \qw & \gate{R_4} & \gate{X^{t_x}}  & \gate{Z^{t_z}}  & \qw  
        }
\]
Since $C_2(R_1,r,s)$ is a depth-one Clifford circuit in the randomization group $\scr{R}_\kappa$, $R_3$ is a tensor product of single-qubit Clifford gates, and $R_4$ is a single-qubit Clifford gate, it follows that $\Lambda_2(R,\ell,s,t)$ is also a member of $\scr{R}_\kappa$. The fact that a description of $\Lambda_2(R,\ell,s,t)$ can be computed by an $\ncz$ circuit follows from the fact that $C_2(R_1,r,s)$ and $R_1,R_2,R_3,R_4$ can all be computed via $\ncz$ circuits (given $R, \ell, s, t$ as input). 

Finally, define $\Lambda_3$ to be the circuit $C_3$, which is a $\qnczf$ circuit. This establishes the Lemma.
\end{proof}

\begin{remark}
Henceforth we will abbreviate the tuple of registers $(\reg{z},\reg{x},\reg{b})$ as register $\reg{a}$. We call the qubits in this register ``ancilla qubits''. 
\end{remark}

\paragraph{The Correction Gadget.} Let $\kappa$ be a positive integer and let $A$ be a unitary from the randomization group $\scr{R}_\kappa$, let $R$ be a single-qubit PX unitary, let $\ell$ be $\kappa$-bit teleportation labels, and let $s = (s_x,s_z), t = (t_x,t_z) \in \binset^2$ be randomization bits. Define the \emph{$\kappa$-correction gadget unitary} $\sf{Corr}_{A,R,\ell,s,t}$ to be the unitary computed by the following circuit:
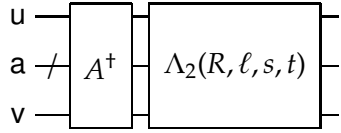
\[
        \Qcircuit @C=0.6em @R=0.8em {
			\reg{u}  & & \qw & \multigate{2}{A^\dagger}  & \multigate{2}{\Lambda_2(R,\ell,s,t)} & \qw \\
			\reg{a}   & & {/}\qw  & \ghost{A^\dagger} 		& \ghost{\Lambda_2(R,\ell,s,t)} &  \qw \\
          \reg{v} & & \qw & \ghost{A^\dagger}	& \ghost{\Lambda_2(R,\ell,s,t)} &   \qw
        }
\]
Note that since $A \in \scr{R}_\kappa$, it follows that the unitary $\sf{Corr}_{A,R,\ell,s,t} \in \scr{R}_\kappa$. Thus a canonical representation of a correction gadget $\sf{Corr}_{A,R,\ell,s,t}$, denoted by $\gadget{Corr}_{A,R,\ell,s,t}$, is an ordered list of single- and two-qubit Clifford gates on the respective qubits.

Furthermore, then for every $\kappa$-bit teleportation labels $\ell$, randomization bits $s,t \in \binset^2$, and single-qubit $R$ from the PX group, for a uniformly random $A$ drawn from $\scr{R}_{\kappa}$, the correction gadget is $\gadget{Corr}_{A,R,\ell,s,t}$ is uniformly distributed over the set of all $\kappa$-correction gadgets.

\subsection{Encoding a Single Gate}
\label{sec:gate-encoding}

We now present the \emph{gate encoding} unitary $\ggate_{g,r,A,\ell,s,t}$ in \Cref{fig:gate-encoding}. This is used by $\QGC$ to encode each gate of the circuit. In what follows, we let $g \in \cal{B}$ denote a ``placeholder gate'' in the circuit topology $\scr{T}$, and let $U_g$ denote a specific unitary for the gate.

The unitary $\ggate_{g,r,A,\ell,s,t}$ is a function of a $p$-qubit gate $U_g$, a string $r$ (which represents the randomness used by $\CRE$), randomization unitaries $A = (A_j)_{j \in [p]}$ with $A_j \in \scr{R}_{\kappa_j}$ for some integer $\kappa_j$, strings $\ell = (\ell_{j,b,a})_{j \in [p],b\in \xz,a \in \binset}$ with $\ell_{j,b,a} \in \binset^{\kappa_j}$ (which represents the teleportation labels), and randomization bits $s = (s_{j,b})_{j \in [p],b \in \binset}$, $t = (t_{j,b})_{j \in [p],b \in \binset}$.
As described at the beginning of this section, we assume that the gate $U_g$ satisfes the property described in \Cref{eq:pauli-to-clifford}.

At a high level, the gate encoding produces a quantum and a classical part. The quantum part is comprised of input qubits in registers $\reg{u}_1,\ldots,\reg{u}_p$, ancillas in registers $\reg{a}_1,\ldots,\reg{a}_p$, and target qubits in registers $\reg{v}_1,\ldots,\reg{v}_p$. These target qubits are each initialized as part of an EPR pair. The gate $U_g$ is first applied to the input registers $\reg{u}$. Then the the circuit $\Lambda_1(\ell_j)$ from \Cref{lem:commute-correction} are applied to registers $(\reg{u}_j,\reg{a}_j,\reg{v}_j)$ for $j \in [p]$, and then the randomization unitaries $A_j$ are applied to registers $(\reg{u}_j,\reg{a}_j,\reg{v}_j)$ for $j \in [p]$. 

Since we think of the input qubits as having been previously teleported, the input qubits will have $X$ and $Z$ corrections on them, and applying the gate $U_g$ will incur further corrections. The evaluator will have to fix these corrections; they will use the classical part of the gate encoding to do this, which is a classical randomized encoding of a \emph{correction function}, which we describe next.

\subsubsection{Correction Functions and Their Randomized Encoding} 
\label{sec:corfunc}

Consider a vector $\kappa = (\kappa_j)_{j \in [p]}$ of label lengths. Define the classical \emph{correction function} $f_{\kappa}$, which on input \\ $(z_1,x_1,\ldots,z_p,x_p,U_g,A,\ell,s,t)$ outputs a tuple of quantum circuits $(\gadget{Corr}_1,\gadget{Corr}_2,\ldots,\gadget{Corr}_p)$ where $\gadget{Corr}_j = \gadget{Corr}_{A_{j},R_j,\ell_j,s_j,t_j}$ for $j \in [p]$ are the correction gadgets defined in \Cref{sec:corr-gadget}. The unitaries $R_1, R_2,\ldots,R_p$ are the single-qubit PX unitaries satisfying
\[
U_g (Z^{z_1} X^{x_1} \otimes \cdots \otimes Z^{z_p} X^{x_p}) = (R_1^\dagger \otimes\cdots \otimes R_p^\dagger) U_g~.
\]
The unitaries $R_1,\ldots,R_p$ are well-defined because $U_g$ comes from our universal set of gates (two-qubit Clifford and $T$ gates). The correction gadgets $\gadget{Corr}_j$ are specified using their canonical representation (i.e., a tensor product of single- and two-qubit Clifford gates). 

We now describe a classical circuit to compute $f_{\kappa}$, by first describing the circuit to compute $\gadget{Corr}_j$ for a single $j$. As discused in \Cref{sec:corr-gadget}, there is a constant-depth circuit of size $O(\kappa_j^2)$ to compute $\gadget{Corr}_j$, and this circuit is a function of $R_j, A_j, \ell_j, s_j, t_j$. The unitary $R_j$ is itself a function of $U_g$ and the tuple $(z_1,x_1,\ldots,z_p,x_p)$. Since the universal gate set used in this paper has arity bounded by $2$ and has a constant number of elements, there is a constant-sized circuit that computes $R_j$. Composing this circuit with the circuit for $\gadget{Corr}_j$, we get a constant-depth circuit of size $O(\kappa_j^2)$ that takes input $g, z_1,x_1,\ldots,z_p,x_p, A_j, \ell_j, s_j, t_j$. 

Putting everything together, we have a classical $\ncz$ circuit, whose depth will be denoted a universal constant $d_{\corr}$ 
and whose size is $c_\kappa = O(\sum_j \kappa_j^2)$ that computes $f_{\kappa}$. Note that the topology of this circuit only depends on the vector $\kappa$; call this topology $\topol_\kappa$. 

Now, consider the encoding $\hf_\kappa$ of $f_\kappa$ with respect to $\CRE$, the DCRE scheme for classical circuits that we use as a blackbox. Since the scheme is decomposable, we have that $\hf_\kappa(z_1,x_1,\ldots,z_p,x_p,U_g, A, \ell, s,t; r)$ consists of an offline part $\hf_{\kappa,\off}(r)$ that only depends on $f_\kappa$, and an online part which are labels for each input $z_1,x_1,\ldots,z_p,x_p,U_g, A, \ell, s,t$. Let $\lab_\kappa(U_g,A,\ell,s,t; r)$ denote the set of labels encoding the inputs $U_g, A, \ell, s,t$. 

Thus one can consider, for a fixed $g, A, \ell, s,t$ the correction function $f_{g,A,\ell,s,t}(z_1,x_1,\ldots,z_p,x_p) = f_\kappa(z_1,x_1,\ldots,z_p,x_p,U_g,A,\ell,s,t)$. The randomized encoding $\hf_\kappa$ is also a randomized encoding of $\hf_{g,A,\ell,s,t}$, where now the offline part $\hf_{g,A,\ell,s,t,\off}(r)$ consists of both $\hf_{\kappa,\off}(r)$ and $\lab_\kappa(U_g,A,\ell,s,t; r)$. The online part are labels for $z_1,x_1,\ldots,z_p,x_p$; for each $j \in [p]$, $b\in \xz$, and $a \in \binset$, let $\lab_{\kappa}(j,b,a; r)$ denote the label of the input variable $b_j$ when it takes value $a$, when the DCRE randomness is $r$ and the topology of the circuit is $\topol_\kappa$.

\SetAlgorithmName{Protocol}{Protocol} 

\subsubsection{The Gate Encoding Unitary}
\label{sec:gate-encoding-unitary}

We now present the gate encoding unitary $\ggate_{g,r,A,\ell,s,t}$. It acts on registers $\reg{u} = (\reg{u}_1,\ldots,\reg{u}_p)$ (which represents the input qubits to the gate $U_g$), $\reg{a} = (\reg{a}_1,\ldots,\reg{a}_p)$ (which represent the ancillas qubits for the teleportation and correction gadgets), $\reg{v} = (\reg{v}_1,\ldots,\reg{v}_p)$ (which represents the entrance of connecting EPR pairs), and $\reg{c}$ (which represents a register to hold the classical randomized encoding of the correction gadget). In the description of the protocol, the unitary $\Lambda_1(\ell_j)$ is given by \Cref{lem:commute-correction}. 

\vspace{10pt}
\IncMargin{1em}
\begin{algorithm}[H]
\DontPrintSemicolon
\tcp{Compute the quantum part of the QRE} 
\Indp Apply $U_g$ to registers $(\regu_1,\ldots,\regu_p)$ \; 
Apply $\Lambda_1(\ell_j)$ to registers $(\regu_j, \rega_j,\regv_j)$ for all $j \in [p]$ \;
Apply $A_j$ to registers $(\regu_j, \rega_j,\regv_j)$ for all $j \in [p]$ \;
\BlankLine
\Indm \tcp{Compute the classical part of the QRE} 
\Indp Compute classical randomized encoding of the correction function $f_{g,A,\ell,s,t}$ as defined in Section~\ref{sec:corfunc}, and let  $\hf_{g,A,\ell,s,t,\off}(r)$ denote the offline part of the randomized encoding using randomness $r$. Store the string $\hf_{g,A,\ell,s,t,\off}(r)$ in the register $\reg{c}$.
\caption{Gate encoding operation $\ggate_{g,r,A,\ell,s,t}$}\label{fig:gate-encoding}
\end{algorithm}\DecMargin{1em}
\vspace{10pt}

The quantum part of the gate encoding is illustrated in \Cref{fig:mtp}. There, $\Lambda_1$ denotes the tensor product of $\Lambda_1(\ell_1),\ldots,\Lambda_p(\ell_p)$, and $A$ denotes the tensor product of $A_1,\ldots,A_p$.

\begin{figure}[H]
\[
        \Qcircuit @C=0.6em @R=0.8em {
			\reg{u}_1,\ldots,\reg{u}_p   & & & & & {/} \qw & \gate{U_g} & \multigate{2}{\Lambda_1} & \multigate{2}{A}   & \qw \\
			\reg{a}_1,\ldots,\reg{a}_p   & & & & & {/}\qw & \qw  & \ghost{\Lambda_1} & \ghost{A} 		 &  \qw \\
          \reg{v}_1,\ldots,\reg{v}_p & & & & & {/} \qw & \qw & \ghost{\Lambda_1} & \ghost{A}	&    \qw
        }
\]
\caption{The quantum part of the gate encoding $\ggate$} \label{fig:mtp}
\end{figure}
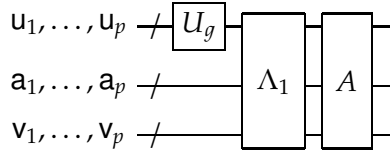

\paragraph{Complexity of the Gate Encoding.} The gate encoding consists of a quantum part and a classical part. The quantum part is applying the gate $U_g$, applying the $\Lambda_1$ circuits, and then applying the randomizers $A_j$. Since the $\Lambda_1$ are $\qnczf$ circuits, and the randomizers is just a single layer of single- and two-qubit Clifford gates, the quantum encoding can be computed in $\qnczf$. 

The classical part is computing the classical randomized encoding of the correction function $f_{g,A,\ell,s,t}$, which has topology $\topol_{\kappa_g}$. The complexity of encoding $f_{g,A,\ell,s,t}$ is inherited from the classical RE used -- if the encoding of $\CRE$ can be computed via a $\ncz$ circuit, then the entire gate encoding procedure can be computed in $\qnczf$.

\subsection{Encoding a Circuit and Input}
\label{sec:circuit-encoding}

We now describe the encoding algorithm $\Enc$ of the $\QGC$ scheme. It takes as input a quantum circuit $F = (\scr{T},\cal{G})$ and quantum input $\qrv{y}$, and outputs the encoding $\hF(\qrv{y}; J)$ where $J$ is a uniform random string.\footnote{We henceforth use $\qrv{y}$ instead of $\qrv{x}$ to denote the input to the function $F$ being encoded; this is to disambiguate it from the qubits in register $\reg{x}$ used by the teleportation and correction gadgets.} We let $n$ denote the number of qubits of $\qrv{y}$.
The encoding is decomposable, so the encoding consists of
\[
	\hat{F}(\qrv{y}; J) = (\hat{F}_\off,\hat{\qrv{y}}_1,\ldots,\hat{\qrv{y}}_n)
\]
where $\hat{F}_\off$ is the offline encoding of $\hat{F}$ and $\hat{\qrv{y}}_i = \hat{F}_i(\qrv{y}_i)$ is the encoding of the $i$-th qubit of $\qrv{y}$.

Let $\cal{Z}$ and $\cal{T}$ be the zero input qubits and discarded output qubits of topology $\scr{T}$, respectively. For simplicity we assume that $\cal{Z}$ is empty (the zero inputs can be incorporated into $\qrv{y}$).

The offline encoding is presented in \Cref{fig:qre-offline} and the input encoding is presented in \Cref{fig:qre-input}. First, we discuss details about the label lengths, the ancillary randomness and quantum random variables used in the encoding.

\paragraph{Label Lengths.} Our quantum randomized encoding contains a classical encoding of a correction function for every gate $g$ in the circuit $F$. To keep track of the lengths of the labels and randomness required, for every wire $w \in \cal{W}$ we let $\kappa_w$ denote the label lengths for encoding each wire $w$, and for every gate $g \in \cal{G}$ we let $c_g$ denote the size of the $\CRE$ encoding of the correction function associated with $g$.

We recursively specify the label lengths $\kappa_w$ and encoding sizes $c_g$, starting from the end of the circuit. 
For all output wires $w \in \cal{O}$, define $\kappa_w = 1$. Then, for every gate $g \in \cal{G}$ whose output wires have label lengths defined so far, let $\kappa_g$ denote the vector $\kappa_g = (\kappa_{w'} : w' \in \woutwire(g))$, which has the label lengths for all the output wires of $g$. Recall from \Cref{sec:corfunc} that classical circuits with topology $\topol_{\kappa_g}$ have depth $d_\corr$ (which is a universal constant) and size $\sigma_g = O(\sum_j \kappa_j^2)$ where $\kappa_j$ is the label length of the $j$-th output wire. Thus the $\CRE$ encoding of a correction function $f_{g,A,\ell,s,t}$ has complexity $c_g = c(d_\corr,\sigma_g)$ and label lengths $\kappa_g = \kappa(d_\corr,\sigma_g)$, where $c(\cdot,\cdot)$ and $\kappa(\cdot,\cdot)$ are given in the statement of Lemma~\ref{lem:main}. For every incoming wire $w \in \winwire(g)$, let $\kappa_w = \kappa_g$. We then recurse on the next layer of gates.

The following observations will be useful for our analysis down the line. First, we note that the encoding complexities $c_g$ and the label lengths $\kappa_w$ only depend on the topology $\scr{T}$ of $F$, and not on the specific unitaries of the gates. Second, as a recursive argument shows, if $d$ is the depth of the quantum circuit $F$ then it holds that $c_g \le c_d$ and $\kappa_g \le \kappa_d$, where $\kappa_d$, $c_d$ are given in the statement of Lemma~\ref{lem:main}.

\paragraph{Quantum Registers.} We now specify the registers in the encoding. 
\begin{itemize}
	\item $\reg{y} = (\reg{y}_1,\ldots,\reg{y}_n)$ is an $n$-qubit register that initially stores the input state $\qrv{y}$.
	\item $\reg{e}^w = (\reg{e}_\win^w,\reg{e}^w_\wout)$ is a two-qubit register for every wire $w \in \cal{W}$, initialized with $\ket{\epr}$. 
	\item $\reg{a}^w = (\reg{z}^w,\reg{x}^w,\reg{b}^w)$ is a $2\kappa_w + (\kappa_w+1)^2$ qubit register for every wire $w \in \cal{W}$, initialized with zeroes. Some of these qubits are used to store teleportation labels, and others are used as part of the correction gadget (see \Cref{sec:corr-gadget} for details). 
	\item $\reg{c}^g$ is a $c_g$-qubit register for every gate $g \in \cal{G}$, initialized with zeroes. These are used to store the description of the $\CRE$ encoding of the correction functions (see \Cref{sec:gate-encoding} for details). 
	\item $\reg{d}^w$ is a $4\kappa_w$-qubit register for every non-traced-out wire $w \in \cal{O} \setminus \cal{T}$, initialized with zeroes. These qubits store the ``label dictionary'' for the output wires (i.e. they store all possible labels for the output wires), so the evaluator can decode the output.
\end{itemize}
Each part of the encoding $\hF_\off, \hF_1,\ldots,\hF_n$ access disjoint subsets of registers.

\paragraph{Classical Randomness.}  We now specify the classical randomness $J$ that is used by the encoding. It consists of the following random strings:
\begin{itemize}
	\item For every gate $g \in \cal{G}$, the random string $r^g$ is a uniformly random string of length $c_g$. The randomness is used for the $\CRE$ encodings of the correction functions. 
	\item For every gate $g \in \cal{G}$, the random string $A^g$ is a sequence $(A^{w})_{w \in \woutwire(g)}$ where for each output wire $w$ of $g$, the string $A^w$ is a uniformly random element of the randomization group $\scr{R}_{\kappa_w}$. These randomizers are used in the encoding of each gate of the circuit.
	\item For every wire $w \in \cal{W}$, $s^w$ is a random pair of bits $(s_z^w,s_x^w) \in \{0,1\}^2$, and $t^w$ is a random pair of bits $(t_z^w,t_x^w) \in \{0,1\}^2$. These values are used to randomize the teleportation measurements (see \Cref{sec:corr-gadget}).
	\item For every output wire $w \in \cal{O}$, $o^w$ is a random pair of bits $(o_z^w,o_x^w)$. These values are used as labels for the output wires.
\end{itemize}
 Because the randomness is classical, it can be copied and thus each part of the encoding $\hF_\off, \hF_1,\ldots,\hF_n$ has access to the entire randomness $J$.

\vspace{10pt}
\IncMargin{1em}
\begin{algorithm}[H]
\DontPrintSemicolon
\tcp{Set up the labels} 
\For{$w \in \cal{W}$}{ \label{item:label-setup}
	\If{$w \notin \cal{O}$}{ 
		Let $g$, $j$ be such that wire $w$ is the $j$-th input wire of gate $g$. \;
		\BlankLine
	Compute the labels $\ell^w = (\ell^w_{b,a})_{b \in \xz,a \in \binset}$ where $\ell^w_{b,a} = \lab_{\kappa_g}(j, b, a; r^g)$ and $\lab_{\kappa_g}$ is the label function corresponding to the $\CRE$ encoding of a circuit with topology $\topol_{\kappa_g}$. \; 
	}
	\ElseIf{$w \in \cal{O}$}{
		Let $\ell^w_{b,0} = o^w_b$ and $\ell^w_{b,1} = o^w_b \oplus 1$ for $b \in \xz$. \;
	}
}
\BlankLine
\tcp{Store the ``label dictionary'' for output wires (ones not traced out)}
\For{$w \in \cal{O} \setminus \cal{T}$}{
	Write the labels $(\ell^w_{z,0}, \ell^w_{z,1}, \ell^w_{x,0}, \ell^w_{x,1})$ into the register $\reg{d}^w$.
}
\BlankLine
\tcp{Encode each gate}
\For{$g \in \cal{G}$}{ \label{item:gate-encoding}
Let $\reg{u}^g = (\reg{e}^v_\wout)_{v \in \winwire(g)}$, $\reg{v}^g = (\reg{e}^w_\win)_{w \in \woutwire(g)}$, and $\reg{a}^g = (\reg{a}^w)_{w \in \woutwire(g)}$. \;
\BlankLine
Let $\ell^g = (\ell^w)_{w \in \woutwire(g)}$, 
$s^g = (s^w)_{w \in \woutwire(g)}$, and $t^g = (t^w)_{w \in \woutwire(g)}$.  \;
\BlankLine
Apply $\ggate_{g,r^g,A^g,\ell^g,s^g,t^g}$ to registers $(\reg{u}^{\, g},\reg{a}^{\, g},\reg{v}^{\, g},\reg{c}^{\, g})$.  \;
}

\caption{The encoding of the offline part $\hat{F}_\off$ of $\hat{F}$.}\label{fig:qre-offline}
\end{algorithm}\DecMargin{1em}
\vspace{10pt}

\vspace{10pt}
\IncMargin{1em}
\begin{algorithm}[H]
\DontPrintSemicolon
	\tcp{Set up the labels}
	Let $w$ denote the $i$-th input wire of $F$, and let $\reg{y}_i$ denote the $i$-th input qubit, and let $g,j$ be such that wire $w$ is the $j$-th input wire of gate $g$. \;
	\BlankLine
	Compute the labels $\ell^w = (\ell^w_{b,a})_{b \in \xz,a \in \binset}$ where $\ell^w_{b,a} = \lab_{\kappa_g}(j, b, a; r^g)$. 
		
	\BlankLine
	\tcp{Teleport the input qubit into the encoding}
	Apply $\sf{TP}_{\ell^w,s^w,t^w}$ to registers $(\reg{y}_i, \reg{z}^w, \reg{x}^w, \reg{e}^w_{\win})$. 
\caption{The encoding $\hat{F}_i(\qrv{y}_i)$ of the $i$-th input qubit $\qrv{y}_i$.}\label{fig:qre-input}
\end{algorithm}\DecMargin{1em}
\vspace{10pt}

\paragraph{Complexity and Locality of the Encoding.} We now argue that $\QGC$ has the claimed complexity properties. First, the QRE is decomposable. The offline part $\hat{F}_\off$ only depends on the circuit $F$, the classical randomness $J = \Big ( r, A, s,t ,o \Big)$, and the EPR pairs $(\qrv{e}^w)$ (except for the qubits $\qrv{e}^w_\win$ for the input wires $w$). The online part $\hat{\qrv{y}}_i = \hat{F}_i(\qrv{y}_i)$ only depends on the classical randomness $J$, the $i$-th input qubit $\qrv{y}_i$ and the qubit $\qrv{e}^w_\win$ for the $i$-th input wire $w$. Thus, the offline encoding and input encoding act on disjoint qubits of the EPR pairs $(\qrv{e}^w)$.

The offline encoding $\hat{F}_\off$ can be computed using a $\qnczf$ circuit. The label setup procedure can be parallelized over all wires, and its complexity is inherited from the complexity of computing labels of $\CRE$ (which we are assuming can be done using an $\ncz$ circuit). Since the complexity of computing the label function $\lab_{\kappa_g}$ is at most $c_d$ where $d$ is the depth of $F$, then the time complexity of the label setup is at most $O(c_d \cdot |\cal{W}|)$. 

The gate encoding can be parallelized over all gates, and as discussed in \Cref{sec:gate-encoding-unitary}, the complexity of encoding a single gate is $O(c_d)$, which includes the complexity of computing the classical encoding of the correction functions. Therefore encoding all gates has time complexity $O(c_d \cdot |\cal{G}|)$. The gate encoding can be implemented by a $\qnczf$ circuit.

Similarly, the encoding of the input qubits can be done in parallel. The complexity of the input encoding comes from setting up the labels and applying the teleportation gadget. This is $O(\kappa_d)$ complexity for encoding each input qubit, and can be done using a $\qnczf$ circuit.

\paragraph{The Case of Classical Inputs.} We observe here that when the input $\qrv{y}$ is classical, the input encoding process can also be taken to be entirely classical. Although applying the teleportation gadget $\sf{TP}$ on the input bits appears to be a fully quantum operation since it involves both a bit of the input $\qrv{y}$ as well as half of an EPR pair, we note that the EPR pair can in this case be ``pre-measured'' in the standard (computational) basis (so that it collapses to a pair of correlated bits), and then one can apply a ``classical'' teleportation gadget:

\begin{figure}[H]
  \centering
  $
        \Qcircuit @C=0.6em @R=0.8em {
          \reg{u}  & & \qw & \ctrl{3} & \qw  & \qw & \qw & \gate{X^{s_x}} & \qw \\
          \reg{z} & & {/} \qw  & \qw & \qw & \gate{X(\ell_{z,0})} & \qw & \qw & \qw  \\
          \reg{x}  & & {/} \qw & \qw & \qw     & \gate{X(\ell_{x,0})} & \gate{X(\ell_{x,0} \oplus \ell_{x,1})} & \qw & \qw \\
          \reg{v} & & \qw & \targ & \qw  & \qw & \ctrl{-1} & \gate{X^{t_x}} & \qw \\          
        }
        $
    \caption{``Classical'' teleportation gadget}
    \label{fig:classical-teleportation-gadget}
\end{figure}
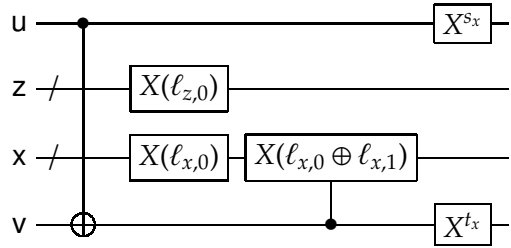
Note that this circuit simply consists of CNOTs and bit flips, which are classically implementable. If the input state at the beginning of the classical teleportation gagdet is $\ubket{y,0,0,r,r}{\reg{u}\reg{z}\reg{x}\reg{v}\reg{u}'}$, then the result of applying the gadget is $\ket{y \oplus s_x, \ell_{z,0}, \ell_{x,e}, e \oplus t_x, e \oplus y}$ where $e = r \oplus y$. Thus the effect of this circuit is to transfer the bit $y$ to register $\reg{u}'$, and then encrypting it using a random bit $e$ that is encoded in the label $\ell_{x,e}$.

\subsection{Circuit Evaluation}
\label{sec:ckt-eval}

We now describe the decoding procedure $\Dec$, which takes a quantum randomized encoding $\hat{F}(\qrv{y})$ and computes $F(\qrv{y})$. The decoding procedure depends on the topology $\scr{T}$ of the circuit $F$, but does not depend on the specific gates themselves. The decoder picks an evaluation order $\pi$ based on the topology $\scr{T}$, and sequentially evaluates each gate $g$ of the circuit $F$. Let $\cal{B}$ denote the set of ``gate placeholders'' in the topology $\scr{T}$. Let $\cal{T} \subseteq \cal{O}$ denote the set of output qubits to be traced out.

\vspace{10pt}
\IncMargin{1em}
\begin{algorithm}[H]
\DontPrintSemicolon
Compute an evaluation order $\pi$ for the topology $\scr{T}$. \;
\tcp{Evaluate the gates according to $\pi$}
\For{$g \in \cal{B}$ ordered according to $\pi$} {
	Apply the unitary $\evg(g)$. \;
}
\tcp{Decode the output}
\For{$w \in \cal{O} \setminus \cal{T}$} {
Given the set of labels $\ell^w = (\ell^w_{z,0},\ell^w_{z,1},\ell^w_{x,0},\ell^w_{x,1})$ in the register $\reg{d}^w$, coherently decode the labels $(z^w,x^w)$ in registers $(\reg{z}^w,\reg{x}^w)$ to bits $(d,e) \in \{0,1\}^2$, and apply $Z^d X^e$ to $\reg{e}^w_\wout$. 

}
\caption{The decoding procedure $\Dec$ of $\QGC$.}\label{fig:ckt-eval}
\end{algorithm}\DecMargin{1em}
\vspace{10pt}

For a $p$-qubit gate $g$, the gate evaluation unitary $\evg(g)$ is defined as follows. Let $(v_1,\ldots,v_p) = \winwire(g)$ and $(w_1,\ldots,w_{p}) = \woutwire(g)$ denote the input and output wires of $g$, respectively. Let
\begin{enumerate}
	\item $(\reg{a}^{v_1},\ldots,\reg{a}^{v_p})$ denote the ancilla registers for the input wires. Each register $\reg{a}^{v_i}$ is composed of subregisters $(\reg{z}^{v_i},\reg{x}^{v_i},\reg{b}^{v_i})$. 
	\item $(\reg{a}^{w_1},\ldots,\reg{a}^{w_{p}})$ denote the ancilla registers for the output wires. 
	\item $(\reg{e}^{v_1}_{\wout},\ldots,\reg{e}^{v_{p}}_{\wout})$ denote the input wire registers for $g$. 
	\item $(\reg{e}^{w_1}_{\win},\ldots,\reg{e}^{w_{p}}_{\win})$ denote the output wire registers for $g$. 
	\item $\reg{c}^g$ denote the register storing the (offline portion of the) classical randomized encoding of the function $f(z_1,x_1,\ldots,z_p,x_p)$ that computes a tuple of correction gadgets $(\gadget{Corr}_1,\ldots,\gadget{Corr}_p)$. 
\end{enumerate}
We note that it may seem strange that the \emph{input} registers are denoted by $e^{v_j}_\wout$, but this is because $e^{v_j}_\wout$ represents the output end of the EPR pair from a \emph{previous} gate teleportation. 

\vspace{10pt}
\IncMargin{1em}
\begin{algorithm}[H]
\DontPrintSemicolon
\tcp{Compute the correction circuits} 
Controlled on the values $(z^{v_1},x^{v_1},\ldots,z^{v_p},x^{v_p})$ of registers $(\regz^{v_1},\regx^{v_1},\ldots,\regz^{v_p},\regx^{v_p})$, and controlled on the value $\hat{f}_\off$ in register $\reg{c}^g$, coherently run $\CDec_{\topol_g}(\hat{f}_\off,z^{v_1},x^{v_1},\ldots,z^{v_p},x^{v_p})$ to obtain descriptions of correction circuits $(\gadget{Corr}_{1},\gadget{Corr}_2,\ldots,\gadget{Corr}_p)$. \;
\BlankLine
\tcp{Apply the correction circuits}
\For{$j \in [p]$}{
	Apply $\sf{Corr}_j$ to registers $(\reg{e}^{v_j}_\wout,\reg{a}^{w_j},\reg{e}^{w_j}_\win)$. \;
	Apply $\Lambda_3$ to registers $(\reg{e}^{v_j}_\wout,\reg{a}^{w_j},\reg{e}^{w_j}_\win)$. \;
}
\caption{The gate evaluation operation, $\evg(g)$, for a $p$-qubit gate $g$.}\label{fig:gate-eval}
\end{algorithm}\DecMargin{1em}
\vspace{10pt}

In the gate evaluation, the map $\Lambda_3$ is the $\qnczf$ circuit given by \Cref{lem:commute-correction}. Furthermore, the CRE decoding procedure $\CDec$ only depends on the topology $\topol_{g}$ of the correction function $f$, which only depends on the label lengths $(\kappa_{w_1},\ldots,\kappa_{w_p})$ of the output wires. The complexity of gate evaluation is dominated by the complexity of running the decoding procedure $\CDec$, which is polynomial in the size $c_g$ of $\hat{f}_\off$.
Thus the complexity of the decoding procedure $\Dec$ is $O(\sum_{g \in \cal{B}} \poly(c_g) + n)$, which is polynomial in the complexity of the encoding procedure.

\subsection{The Simulator}
\label{sec:simulator}

We now present the simulator $\Sim$ for $\QGC$. It depends on the topology $\scr{T}$ of the circuit being simulated, and takes as input a register $\reg{s}$ that is supposed to store the output of a quantum operation $F$ that has topology $\scr{T}$. We assume that $\scr{T}$ has $n$ input wires, and that the input register $\reg{s}$ is $|\cal{O} \setminus \cal{T}|$ qubits wide (i.e. the number of output qubits of $\scr{T}$ that aren't traced out).

Intuitively the simulator computes the encoding of the identity circuit $E$ with input $\reg{s}$ padded with zeroes (so the output should be the quantum state stored in register $\reg{s}$). However the topology $\scr{T}$ could permute the ordering of qubits, so placing identity gates for all the placeholder gates of $\scr{T}$ may still result in a nontrivial operation on the input qubits. Furthermore, the topology $\scr{T}$ may trace out some qubits. Thus the simulator pads the input with some zeroes, shuffles the qubits according to the inverse of the permutation effected by the topology $\scr{T}$, and then computes the randomized encoding of $E$ and the shuffled input.

\vspace{10pt}
\IncMargin{1em}
\begin{algorithm}
\DontPrintSemicolon
Let $E$ denote the general quantum circuit with topology $\scr{T}$ and gate set $\cal{G}$ consisting only of identity gates. \;
Let $n = |\cal{I}| = |\cal{O}|$ denote the number of input and output wires of topology $\scr{T}$. \;
Compute the bijection $\xi: \cal{I} \to \cal{O}$ that is the result of the (unitary part of) $E$. \;
Let $P_\xi$ denote the unitary that swaps $n$ qubits according to the bijection $\xi$. \;
Let $\reg{y}$ denote an $n$-qubit register consisting of $\reg{s}$, padded by zeroes. \;
Permute the qubits of $\reg{y}$ according to the permutation $P^{-1}_\xi$. \;
Compute the encoding $\hat{E}(\qrv{y})$ of circuit $E$ and input in register $\reg{y}$. \;
\caption{The simulator $\Sim$ for $\QGC$.}\label{fig:sim}
\end{algorithm}\DecMargin{1em}
\vspace{10pt}

Clearly the complexity of the simulation procedure is polynomial in the complexity of the encoding procedure.

\section{Correctness and Privacy Analysis}
\label{sec:analysis}

We now analyze the correctness and privacy of the $\QGC$ scheme. We argue that applying the decoding procedure $\Dec$ to a QRE $\hF(\qrv{y})$ of a circuit $F$ and input $\qrv{y}$ yields the state $F(\qrv{y}) \otimes \rho$, where $\rho$ is indistinguishable from a density matrix that only depends on the topology $\scr{T}$ of $F$, and nothing else about $F$. This clearly implies the correctness of the decoding procedure, but also shows the privacy of the $\QGC$ scheme: let $E$ denote the ``empty'' circuit with the same topology $\scr{T}$. As discussed in \Cref{sec:simulator}, the empty circuit $E$ may effect a permutation $P$ on the qubits of $F(\qrv{y})$, along with additional zero qubits. Applying the decoding procedure $\Dec$ to QRE $\hat{E}\left ( \, P \, (F(\qrv{y}) \, ) \, \right)$ yields a state that is indistinguishable from $F(\qrv{y}) \otimes \rho$. Since the decoding procedure $\Dec$ is unitary, this implies that 
\[
	\hat{E}\left ( \, P \, (F(\qrv{y}) \, ) \, \right) \approx \Dec^{-1} \Big( F(\qrv{y}) \otimes \rho \Big) \approx \hat{F}(\qrv{y}).
\]
where ``$\approx$'' denotes indistinguishability either in the computational or information-theoretic sense, depending on the security properties of the classical randomized encoding scheme $\CRE$ used. This shows that the simulator $\Sim$, on input $F(\qrv{y})$, produces a state indistinguishable from $\hF(\qrv{y})$. 

\paragraph{Some Notation.}  Let $\QGC$ denote the quantum randomized encoding scheme described in \Cref{sec:construction}, where we use a classical randomized encoding scheme $\CRE$ that has perfect correctness and has polynomial-time encoding, decoding and simulation procedures. We assume that $\CRE$ has $(t,\epsilon)$-privacy with respect to quantum adversaries. In the setting of computational security, $t$ and $\epsilon$ are implicitly functions of a security parameter $\secp$.  
For a topology $\topol$ of a classical circuit, let $\CDec_{\topol}$ and $\CSim_{\topol}$ denote the decoder and simulator for $\CRE$ for classical circuits with topology $\topol$, respfectively. 

Fix an $n$-qubit input $\qrv{y}$, quantum side information $\qrv{q}$ (which may be entangled with $\qrv{y}$), and a general quantum circuit $F = (\scr{T},\cal{G})$ with $m = \poly(n)$ gates. Let $\cal{B}$ denote the set of ``gate placeholders'' in the topology $\scr{T}$. As described in \Cref{sec:ckt-eval}, the circuit evaluation only depends on the topology $\scr{T}$; it sequentially evaluates each gate in $\cal{B}$ according to some evaluation order $\pi$, and then decodes the output wires using the labels in the set $\cal{O} \setminus \cal{T}$. Let $g_1,\ldots,g_m \in \cal{B}$ denote the placeholder gates of topology $\scr{T}$ ordered according to $\pi$, and let $U_1,\ldots,U_m$ denote the corresponding unitaries in circuit $F$. 
Let $\cal{W}$ denote the set of wires in the topology $\scr{T}$, with $\cal{I},\cal{O}$ denoting the input/output wires respectively.

Let $F_{\leq i}$ denote the part of circuit $F$ up to (and including) gate $g_i$. (We define $F_{\leq 0}$ to denote the identity circuit with no gates). 
Thus, $F_{\leq i}$ represents the unitary operation $U_{i} \cdots U_1$.

Similarly, define $F_{> i}$ to denote the part of circuit $F$ that starts after $g_i$. (We define $F_{> 0}$ to denote $F$). As a quantum operation it first applies the unitaries $U_{i+1}, \ldots, U_m$, and then traces out the qubits specified by set $\cal{T} \subseteq \cal{O}$. As we will be considering randomized encodings of the partial circuit $F_{ > i}$, we define the randomness used by the encodings. As described in \Cref{sec:circuit-encoding}, the randomness used to encode $F = F_{> 0}$ is the sequence $J = (r,A,s,t,o)$. The randomness used to encode $F_{ > i}$ is the sequence $J_{>i} = (r_{> i},A_{> i},s_{> i},t_{> i},o)$, where 
\begin{itemize}
	\item $r_{>i} = (r^g : g \text{ comes after $g_i$})$
	\item $A_{> i} = (A^g : g \text{ comes after $g_i$})$
	\item $s_{> i} = (s^w : \text{ wire $w$ comes after gates $g_1,\ldots,g_i$})$
	\item $t_{> i} = (t^w : \text{ wire $w$ comes after gates $g_1,\ldots,g_i$})$
	\item $o = (o^w)_{w \in \cal{O}}$ is the same as before.
\end{itemize}
We define $J_{> 0} = J$. We define $J_{\leq i} = J \setminus J_{> i}$. This is all the randomness that is ``used up'' in evaluating the first $i$ gates.

Given a classical value $c$, we write $\puretomixed{c}$ to denote the density matrix $\ketbra{c}{c}$.

\subsection{Analysis of the Decoding Procedure}

The decoding procedure $\Dec$ evaluates each of the gates of $F$ in sequence (according to some evaluation order), and then traces out a subset of qubits. The following Lemma gives a characterization of the result of each gate evaluation:

\begin{lemma}
\label{lem:single-eval}
	Let $i \in\{0,1,\ldots,m-1\}$. Let $\evg_{i+1}$ denote the gate evaluation procedure from \Cref{sec:ckt-eval} corresponding to the gate $g_{i+1}$ of topology $\scr{T}$. Let $\qrv{q}$ denote quantum side information that is possibly entangled with $\qrv{y}$, but uncorrelated with the classical randomness used by the encoding $\hF(\qrv{y})$. Then 
	\begin{equation}
	\label{eq:sec-0}
		\E_{J_{> i}} \, \Big( \evg_{i+1} \, ( \,\hF_{> i} \, ( \, F_{\leq i}(\qrv{y} \, ; \, J_{> i}) \, ) , \qrv{q} \Big) = \E_{J_{> i+1}} \, \Big( \hF_{> i+1} ( \, F_{\leq i+1}(\qrv{y} \, ; \, J_{> i+1}) \,), \qrv{q} \Big ) \otimes \qrv{t}
	\end{equation}
	where for $j \in \{i,i+1\}$,
	\begin{itemize}
		\item $\hF_{> j}(\hF_{\leq j}(\qrv{y} \, ; \, J_{> j}))$ denotes the quantum randomized encoding of the circuit $F_{> j}$ and input $\hF_{\leq j}(\qrv{y})$, where the randomness used for the encoding is $J_{ > j}$.
		\item The density matrix $F_{\leq j}(\qrv{y})$ is stored in registers
\[
	( \reg{e}^v_\wout : \text{$v$ is the predecessor wire to some $w \in \winwire(g_{j+1})$} )
\]
where $v$ is a predecessor wire to $w$ if $v,w$ are the $k$-th input and output wires respectively of the same gate $g$, for some $k$. If $w$ is the $k$-th input wire of topology $\scr{T}$, then there is no predecessor wire, but in that case we define $\reg{e}^v_\wout$ to be the register $\reg{y}_k$.
	\end{itemize}
Furthermore, the density matrix $\qrv{t}$ satisfies the following:
\begin{enumerate}
	\item It is on registers $( \reg{e}^v_\wout : \text{$v$ is the predecessor wire to $w \in \winwire(g_{i+1})$} )$ and $(\reg{a}^w, \reg{e}^w_\win)_{w \in \winwire(g_{i+1})}$.
	\item $\qrv{t}$ is $(t,\epsilon)$-indistinguishable from a density matrix $\qrv{\tilde{t}}$ that only depends on topology $\scr{T}$.
	\item $\qrv{t}$ is unentangled with $\hF_{> i+1} ( \, F_{\leq i+1}(\qrv{y}) \,)$ and $\qrv{q}$.
	\item The state  $\qrv{t}$ is independent of the randomness $J_{ > i+1}$. 
\end{enumerate}
\end{lemma}

We first show how \Cref{lem:single-eval} implies the correctness and privacy of $\QGC$. Using \Cref{lem:single-eval} repeatedly, we get
\begin{equation}
\label{eq:sec-1}
\E_J \, \Big( \evg_m \cdot \evg_{m-1} \cdots \evg_1 ( \,\hF(\qrv{y}) \, ; \, J ) , \qrv{q} \Big)
\end{equation}
is equal to
\begin{equation}
\label{eq:sec-2}
	\E_{J_{> m}} \, \Big( \hF_{> m} ( \, F_{\leq m}(\qrv{y} \, ; \, J_{> m}) \,), \qrv{q} \Big ) \otimes \qrv{t}_1 \otimes \cdots \otimes \qrv{t}_m
\end{equation}
where each $\qrv{t}_i$ satisfies the conclusions of the statement of the Lemma. Note that $F_{\leq m}$ is the unitary part of the circuit $F$ (without the tracing out of the output wires $\cal{T}$), and $\hat{F}_{> m}$ is the quantum randomized encoding of the partial trace operation $\Tr_{\cal{T}}(\cdot)$. Let $\qrv{y}^{(m)} = F_{\leq m}(\qrv{y})$ denote the state of the circuit before tracing out. Since the circuit $F_{> m}$ has no gates, the randomized encoding $\hF_{> m}( \qrv{y}^{(m)} )$ is the state $\Big(\hat{\qrv{y}}^{(m)}_1,\ldots,\hat{\qrv{y}}^{(m)}_n, \, (\qrv{d}^w)_{w \in \cal{O} \setminus \cal{T}} \, \Big)$ where $\hat{\qrv{y}}^{(m)}_j$ is the encoding of the $j$-th qubit of $\qrv{y}^{(m)}$ and $\qrv{d}^w$ is the label dictionary $\ell^w = (\ell^w_{z,0},\ell^w_{z,1},\ell^w_{x,0},\ell^w_{x,1})$ for output wire $w \in \cal{O} \setminus \cal{T}$. 

First we note that the randomness $J_{> m}$ used in the encoding $\hat{F}_{> m}(F_{\leq m}(\qrv{y} \, ; \, J_{> m}))$ is the collection of random values $(o^w, s^w, t^w)_{w \in \cal{O}}$, where $o^w = (o^w_x,o^w_z)$ are bits used to generate the labels $(\ell^w)_{w \in \cal{O}}$ for the output wires, and $(s^w,t^w)$ are randomization bits used for the teleportation gadget (see \Cref{sec:circuit-encoding} for details on the randomness used in the encoding). 

Fix an index $j \in [n]$. The state $\hat{\qrv{y}}^{(m)}_j$ is on the following registers. Let $w \in \cal{O}$ denote the $j$-th output wire. Let $v \in \cal{W}$ denote the predecessor wire to $w$. 
\begin{itemize}
	\item $\reg{e}^v_\wout$. This register initially stores the $j$-th qubit of $\qrv{y}^{(m)}_j$. 
	\item $(\reg{e}^w_\win, \reg{e}^w_\wout)$. These initially store the EPR pair corresponding to wire $w$. 
	\item $\reg{a}^w = (\reg{z}^w,\reg{x}^w,\reg{b}^w)$. This is the ancilla register for wire $w$.
\end{itemize}
The joint state of $\hat{\qrv{y}}^{(m)}_j$ and the register $\reg{e}^w_\wout$ in \eqref{eq:sec-2}, when averaged over the randomness $(o^w, s^w, t^w)_{w \in \cal{O}}$, is equal to
\begin{equation}
\label{eq:sec-3}
\E_{o^w, s^w, t^w} \, \Big ( \sf{TP}_{\ell^{w},s^{w},t^{w}}( \ub{  \qrv{y}^{(m)}_j, \puretomixed{0} , \puretomixed{0} , \qrv{e}^{w}_\win }{ \reg{e}^{v}_\wout \, \reg{z}^w \, \reg{x}^w \, \reg{e}^w_\win} ) \, , \, \ub{\qrv{e}^w_\wout}{\reg{e}^w_\wout} \, , \, \ub{\puretomixed{0} }{\reg{b}^w}  \Big) = \E_{o^w} \, \ub{\tau \otimes \E_{d,e} \puretomixed{ \ell^w_{z,d}, \ell^w_{x,e}} \otimes \tau}{ \reg{e}^{v}_\wout \, \reg{z}^w \, \reg{x}^w \, \reg{e}^w_\win} \otimes \ub{X^e Z^d(\qrv{y}^{(m)}_j)}{\reg{e}^w_\wout} \otimes \ub{\puretomixed{0} }{\reg{b}^w} 
\end{equation}
where
the density matrix $\tau$ denotes the maximally mixed qubit, and $\puretomixed{0}$ denotes the pure state $\ketbra{0 \cdots 0}{0\cdots 0}$ for the appropriate number of zero bits. 
\Cref{eq:sec-3} follows from the following lemma.

\begin{lemma}
	\label{lem:tp-twirl}
	Let $\qrv{u}$ denote a qubit density matrix and let $\qrv{q}$ denote quantum side information that is possibly entangled with $\qrv{u}$. Let $s = (s_z,s_x)$ and $t = (t_z,t_x)$ denote randomization bits. Then for all labels $\ell = (\ell_{z,0},\ell_{z,1},\ell_{x,0},\ell_{x,1})$, we have
	\begin{equation}
	\label{eq:tp-twirl-1}
		\E_{s,t} \Big( \sf{TP}_{\ell,s,t} \Big( \qrv{u}, \puretomixed{0 \cdots 0}, \qrv{e}_\win \Big), \, \qrv{e}_\wout, \qrv{q} \Big) = \tau \otimes \E_{d,e} \puretomixed{\ell_{z,d}, \ell_{x,e} } \otimes \tau \otimes \Big( X^e Z^d(\qrv{u}), \, \qrv{q} \Big)
	\end{equation}
	where $(\qrv{e}_\win,\qrv{e}_\wout)$ is initialized in the state $\ket{\epr}$. 
\end{lemma}
\begin{proof}
We prove this by showing that \Cref{eq:tp-twirl-1} holds when $\qrv{u}$ is the outer product $\ketbra{a}{a'}$ and there is no quantum side information $\qrv{q}$; the Lemma follows by linearity of the $\sf{TP}$ operation.

By \Cref{lem:tp-gadget}, we have that
\[
	\sf{TP}_{\ell,s,t} \Big( \ket{a} \otimes \ket{0 \cdots 0} \otimes \ket{\epr} \Big) = \frac{1}{2} \sum_{d,e} Z^{s_z} X^{s_x} \ket{d} \otimes \ket{\ell_{z,d}, \ell_{x,e}} \otimes Z^{t_z} X^{t_x} \ket{e} \otimes X^e Z^d \ket{a}\;.
\]
Thus
\begin{align*}
&\E_{s,t} \, \Big(\sf{TP}_{\ell,s,t}  \Big( \ketbra{a}{a'}, \puretomixed{0 \cdots 0}, \qrv{e}_\win \Big), \qrv{e}_\wout \Big) \\
&= \E_{s,t} \, \frac{1}{4} \sum_{d,e,d',e'} (-1)^{s_z \cdot (d \oplus d')} \cdot (-1)^{t_z \cdot (e \oplus e')} \ketbra{d \oplus s_x}{d' \oplus s_x} \otimes \ketbra{\ell_{z,d}, \ell_{x,e}}{\ell_{z,d'}, \ell_{x,e'}} \otimes \ketbra{e \oplus t_x}{e' \oplus t_x} \otimes \ X^e Z^d (\ketbra{a}{a'}) \\
&= \E_{s_x,t_x} \frac{1}{4} \sum_{d,e} \puretomixed{d \oplus s_x} \otimes \puretomixed{\ell_{z,d}, \ell_{x,e}} \otimes \puretomixed{e \oplus t_x} \otimes  X^e Z^d (\ketbra{a}{a'}) \\
&= \E_{d,e} \tau \otimes \puretomixed{\ell_{z,d}, \ell_{x,e}} \otimes \tau \otimes  X^e Z^d (\ketbra{a}{a'}).
\end{align*}
\end{proof}

The first observation is that for each $w \in \cal{O}$, the maximally mixed qubits $\tau$ in registers $(\reg{e}^v_\wout, \reg{e}^w_\win)$ in \Cref{eq:sec-3} are completely unentangled from the rest of the state described in \eqref{eq:sec-2}, because the Pauli twirl bits $s^w,t^w$ are only used to twirl the registers $(\reg{e}^v_\wout, \reg{e}^w_\win)$, and is uncorrelated with the side information $\qrv{q}$.

The next observation is that for each $w \in \cal{T}$ (i.e. the traced-out wires), the registers $(\reg{z}^w,\reg{x}^w,\reg{e}^w_\wout)$ of \eqref{eq:sec-2} are in the state
\[
	\E_{o_z^w, o_x^w} \, \E_{d,e} \,\puretomixed{ o^w_z \oplus d, o^w_x \oplus e} \otimes X^e Z^d(\qrv{y}^{(m)}_j) = \E_{o^w_z, o^w_x} \, \puretomixed{ o^w_z, o^w_x} \otimes \E_{d,e} \, X^e Z^d(\qrv{y}^{(m)}_j) = \tau \otimes \tau \otimes \tau
\]
where $w$ is the $j$-th output wire. Furthermore, the randomness $(o^w_z,o^w_x)$ is uncorrelated with the rest of the state described in \eqref{eq:sec-2}, so the registers $(\reg{z}^w,\reg{x}^w,\reg{e}^w_\wout)$ are unentangled with the rest of the state.

Finally, for each $w \in \cal{O} \setminus \cal{T}$ (i.e. the non-traced-out wires), the registers $(\reg{d}^w,\reg{z}^w,\reg{x}^w,\reg{e}^w_\wout)$ are in the state
\[
	\E_{o_z^w, o_x^w} \, \puretomixed{o^w_z,o^w_z \oplus 1, o^w_x, o^w_x \oplus 1} \otimes \E_{d,e} \, \puretomixed{ o^w_z \oplus d, o^w_x \oplus e} \otimes X^e Z^d(\qrv{y}^{(m)}_j)
\]

Now consider the decoding procedure $\Dec$, when applied to the encoding $\hat{F}(\qrv{y})$. It evaluates each gate by applying $\evg_1,\evg_2,\ldots$, and then finally at the end decodes each wire $w \in \cal{O} \setminus \cal{T}$ by undoing the Pauli twirl. The resulting state in registers $(\reg{d}^w,\reg{z}^w,\reg{x}^w,\reg{e}^w_\wout)$ is then indistinguishable from
\[
		\E_{o_z^w, o_x^w} \, \puretomixed{o^w_z,o^w_z \oplus 1, o^w_x, o^w_x \oplus 1} \otimes \E_{d,e} \, \puretomixed{ o^w_z \oplus d, o^w_x \oplus e} \otimes \qrv{y}^{(m)}_j = \Big ( \E_a \puretomixed{a, a \oplus 1} \Big)^{\otimes 2} \otimes \tau^{\otimes 2} \otimes \qrv{y}^{(m)}_j\;.
\]
Let $\sigma$ denote the density matrix $\E_a \puretomixed{a, a \oplus 1}$. Putting everything together, the result of applying the decoding procedure $\Dec$ to $(\hat{F}(\qrv{y}), \qrv{q})$ results in the following state:
\[
	\Big( \Tr_{\cal{T}}(\qrv{y}^{(m)}), \qrv{q} \Big) \otimes \underbrace{\Big( \sigma^{\otimes 2} \otimes \tau \Big)^{\otimes |\cal{O} \setminus \cal{T}|} \otimes \Big( \tau^{\otimes 3} \Big)^{\otimes |\cal{T}|} \otimes \qrv{t}_1 \otimes \cdots \otimes \qrv{t}_m}_{\rho} \;.
\]
Since $F(\qrv{y}) = \Tr_{\cal{T}}(\qrv{y}^{(m)})$, this proves the correctness of the decoding procedure. Note that $\rho$ is $(t,\epsilon \cdot m)$-indistinguishable from the density matrix
\[
	\tilde{\rho} = \Big( \sigma^{\otimes 2} \otimes \tau \Big)^{\otimes |\cal{O} \setminus \cal{T}|} \otimes \Big( \tau^{\otimes 3} \Big)^{\otimes |\cal{T}|} \otimes \qrv{\tilde{t}}_1 \otimes \cdots \otimes \qrv{\tilde{t}}_m
\]
where each $\qrv{\tilde{t}}_i$ only depends on $\scr{T}$, as given by \Cref{lem:single-eval}. Thus $\tilde{\rho}$ only depends on the topology $\scr{T}$. 

Similarly, since $E(P(F(\qrv{y}))) = F(\qrv{y})$, it must be that
\[
\Big ( \, \hat{E}( P(F(\qrv{y}))), \qrv{q} \, \Big) = \Big ( \, F(\qrv{y}), \qrv{q} \, \Big) \otimes \rho'
\]
where $\rho'$ is $(t,\epsilon \cdot m)$-indistinguishable from $\tilde{\rho}$. 

This in turn implies that $(\hat{F}(\qrv{y}), \qrv{q})$ is $(t', \epsilon')$-indistinguishable from $(\hat{E}( P(F(\qrv{y}))), \qrv{q}) = (\Sim(F(\qrv{y})),\qrv{q})$, where $t' = t - \poly(s)$ where $\poly(s)$ is the complexity of $\Dec$, and $\epsilon' = \epsilon \cdot m$. This establishes the privacy property of $\QGC$.

\subsection{Proof of \Cref{lem:single-eval}}
\label{sec:single-gate-eval}

We prove \Cref{lem:single-eval} for the case that $i = 0$; the argument is identical for $i > 0$. In particular, we will show that
\begin{equation}
\label{eq:sec-3b}
	\E_J \, \Big ( \evg_1 ( \, \hat{F}( \qrv{y}) \, ; \, J ), \qrv{q} \Big) = \E_{J_{> 1}} \, \Big ( \hat{F}_{> 1} (\, U_1(\qrv{y}) \, ; \, J_{> 1}), \qrv{q} \Big) \otimes \qrv{t}
\end{equation}
where $\qrv{t}, \qrv{q}$ satisfy the conditions in the statement of \Cref{lem:single-eval}.

Let $g_1$ denote a $p$-qubit gate, with $U_1$ being the corresponding unitary operator. For simplicity we drop the subscript and just write $g, U$ from now on.  

In encoding procedure presented in \Cref{sec:circuit-encoding}, each gate of $F$ and each input qubit of $\qrv{y}$ is encoded independently. Thus we can treat $\hat{F}$ as the following process:
\begin{enumerate}
	\item Sample classical randomness $J = (r,A,s,t,o)$;
	\item Compute the labels $(\ell^w)_{w \in \cal{W}}$; 
	\item Apply $M \otimes \ggate_1 \otimes \sf{TP}_1$ where $\ggate_1$ is the gate encoding unitary for gate $g$; the map $\sf{TP}_1$ is the unitary to encode the $p$ input qubits of $\qrv{y}$ that are acted upon by $U$, and $M$ is the concatenation of all the other gate encoding and input encoding unitaries.
\end{enumerate}
In more detail: 
\begin{itemize}
	\item The map $\sf{TP}_1$ is equal to
\[
	\sf{TP}_1 = \bigotimes_{v \in \winwire(g)} \sf{TP}_{\ell^v, s^v, t^v}
\]
where $\sf{TP}_{\ell^v, s^v, t^v}$ acts on the registers $(\reg{y}^v, \reg{a}^v, \reg{e}^v_\win)$. Here, $\reg{y}^v$ denotes the input qubit that is encoded into input wire $v$. The register $\reg{a}^v$ can be decomposed as $(\reg{z}^v, \reg{x}^v, \reg{b}^v)$.
\item The map $\ggate_1$ is equal to
\[
	\ggate_1 = \ggate_{g,r^g,A^g,\ell^g,s^g,t^g}
\]
where $s^g = (s^w)_{w \in \woutwire(g)}$, $t^g = (t^w)_{w \in \woutwire(g)}$, and $\ell^g = (\ell^w)_{w \in \woutwire(g)}$. The map $\ggate_1$ acts on registers $(\reg{e}^v_\wout)_{v \in \winwire(g)}$, $(\reg{a}^w)_{w \in \woutwire(g)}$, $(\reg{e}^w_\win)_{w \in \woutwire(g)}$, and $\reg{c}^g$. 

\item The map $M$ acts on a disjoint set of registers from $\sf{TP}_1$ and $\ggate_1$.
\end{itemize}

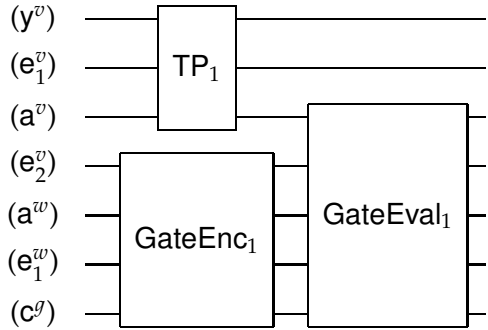
\begin{figure}[H]
\[
        \Qcircuit @C=0.6em @R=0.8em {
          (\reg{y}^v) 				& & & & \qw & \multigate{2}{\sf{TP}_1} & \qw & \qw & \qw \\
          (\reg{e}^v_\win) 			& & & & \qw & \ghost{\sf{TP}_1} & \qw & \qw & \qw \\
          (\reg{a}^v)   	& & & & \qw & \ghost{\sf{TP}_1} & \qw & \multigate{4}{\evg_1} & \qw \\
          (\reg{e}^v_\wout) 		& & & & \qw & \multigate{3}{\ggate_1} & \qw & \ghost{\evg_1} & \qw  \\
          (\reg{a}^w) 				& & & & \qw & \ghost{\ggate_1} & \qw & \ghost{\evg_1} & \qw  \\
          (\reg{e}^w_\win)			& & & & \qw & \ghost{\ggate_1} & \qw & \ghost{\evg_1} & \qw  \\
          (\reg{c}^g)				& & & & \qw & \ghost{\ggate_1} & \qw & \ghost{\evg_1} & \qw  
        }
\]
\caption{A circuit diagram of the encoding of gate $g$ and its inputs, proceeded by its evaluation. The indices $v,w$ range over $\winwire(g)$ and $\woutwire(g)$ respectively.} \label{fig:gate-eval-1}
\end{figure}

We compute the state $\Big ( \evg_1 ( \, \hat{F}( \qrv{y} \, ; \, J) \, ), \qrv{q} \Big)$, which is correlated with the randomness $J$ and labels $(\ell^w)_{w \in \cal{W}}$. We focus specifically on the registers depicted in \Cref{fig:gate-eval-1}, which are the ones acted upon by $\sf{TP}_1$, $\ggate_1$, and $\evg_1$. In what follows, all randomness $J$ and labels $(\ell^w)$ are fixed unless we explicitly average over certain random variables. Furthermore, the indices $v,w$ range over $\winwire(g)$ and $\woutwire(g)$ respectively.

\paragraph{Computing the Input Teleportation.} Fix $v \in \winwire(g)$. We compute the result of applying $\sf{TP}_{\ell^v,s^v,t^v}$ to registers $(\reg{y}^v, \reg{z}^v, \reg{x}^v, \reg{e}^v_\win,\reg{e}^v_\wout)$, after averaging over the randomization bits $(s^v, t^v)$:
\[
	\E_{s^v, t^v} \Big( \sf{TP}_{\ell^v,s^v,t^v} \, \Big( \ub{\qrv{y}^v, \puretomixed{0 \cdots 0}, \qrv{e}^v_\win}{\reg{y}^v \, \reg{z}^v \, \reg{x}^v \,  \reg{b}^v \, \reg{e}^v_\win} \Big), \, \ub{\qrv{e}^v_\wout}{\reg{e}^v_\wout} \Big) = \tau \otimes \E_{d_v,e_v} \, \puretomixed{\ell^v_{z,d_v}, \ell^v_{x,e_v}} \otimes \puretomixed{0} \otimes \tau \otimes X^{e_v} Z^{d_v}(\qrv{y}^v) 
\]
where on the right-hand side, the $\puretomixed{0}$ state is on the $\reg{b}^v$ register (the teleportation gadget acts as the identity on register $\reg{b}^v$). This follows from \Cref{lem:tp-twirl}.

\paragraph{Computing the Gate Encoding.} Now we compute the result of applying $\ggate_1$, where we've conditioned on specific values of $d_v, e_v$: 
\begin{align}
	 &\ggate_1 \Big( \bigotimes_v \ub{X^{e_v} Z^{d_v}(\qrv{y}^v)}{\reg{e}^v_\wout}, \ub{\puretomixed{0 \cdots 0}}{ (\reg{a}^w)_w}, \ub{(\qrv{e}^w_\win)_w}{(\reg{e}^w_\win)_w}, \ub{\puretomixed{0}}{\reg{c}^g} \Big) \notag \\
	 & \qquad \qquad =  \Big ( A^g \cdot \Lambda_1 \cdot \Big(\,  U \Big( \bigotimes_v X^{e_v} Z^{d_v}(\qrv{y}^v) \, \Big), \puretomixed{0 \cdots 0}, (\qrv{e}^w_\win)_w \Big) \Big ) \otimes \puretomixed{ \hf_{g,\off} } \label{eq:sec-4}
\end{align}
The unitary $\Lambda_1$ is the tensor product 
\[
	\Lambda_1 = \bigotimes_{w \in \woutwire(g)} \Lambda_1(\ell^w)
\]
where $\Lambda_1(\ell^w)$ are the $\qnczf$ circuits specified by \Cref{lem:commute-correction}. 
The value $\hf_{g,\off}$ is the (offline part of the) classical randomized encoding of the correction function $f(d_1,e_1,\ldots,d_p,e_p)$ (see \Cref{sec:corfunc} for the definition of correction functions and their encoding). 

By our assumption on the gate set of the circuit $F$, we use \Cref{eq:pauli-to-clifford} to argue that
\[
	 U \cdot \bigotimes_{v \in \winwire(g)} X^{e_v} Z^{d_v} = \Big( \bigotimes_{v \in \winwire(g)} R_{d_v,e_v}^\dagger \Big) \cdot U
\]
for single-qubit PX group elements $(R_{d_v,e_v})$. Thus, the right-hand side of \Cref{eq:sec-4} can be written as
\[
\Big ( A^g \cdot \Lambda_1 \cdot \Big( \bigotimes_v R^\dagger_{d_v,e_v} \Big) \,  \Big( U (\qrv{y}^v)_v , \puretomixed{0 \cdots 0}, (\qrv{e}^w_\win)_w \Big) \Big ) \otimes \puretomixed{ \hf_{g,\off}} 
\]

\paragraph{Computing the Gate Evaluation.} Next, we compute the effect of applying the gate evaluation unitary $\evg_1$, which is described in \Cref{sec:ckt-eval}. 
Controlled on the registers $(\reg{z}^v,\reg{x}^v)_v$, which store the labels $(\ell^v_{z,d_v},\ell^v_{x,e_v})$ encoding the Pauli twirl $X^{e_v} Z^{d_v}$, the map $\evg_1$ uses the encoding $\hf_{g,\off}$ to compute the correction function $f(d_1,e_1,\ldots,d_p,e_p)$. This yields (descriptions of) correction gadgets $(\gadget{Corr}_v)_{v}$ such that
\[
	\sf{Corr}_v = \Lambda_2(R_{d_v,e_v}, \ell^w, s^w, t^w) \cdot (A^w)^\dagger
\] 
where $w$ is the ``successor wire'' to $v$ (i.e. $v,w$ are the $k$-th input and output wires of $g$, respectively), and $\Lambda_2(R,\ell,s,t)$ are the unitaries specified by \Cref{lem:commute-correction}. Here, we assume that $\CRE$ has perfect correctness, which means that the correction gadgets are computed without error. Furthermore, for the rest of the section we will interchangeably index the input wires by either $v \in \winwire(g)$ or integers $1,\ldots,p$. 

The unitary $\evg_1$ then applies the unitaries $\sf{Corr}_v$ to registers $(\reg{e}^{v}_\wout,\reg{a}^{w},\reg{e}^{w}_\win)$, followed by applying the $\qnczf$ circuit $\Lambda_3$ (specified again by \Cref{lem:commute-correction}) to the same registers. Put together, the state of registers $(\qrv{e}^v_\wout,\qrv{a}^w, \qrv{e}^w_\win)_{v,w}$ looks like the following:

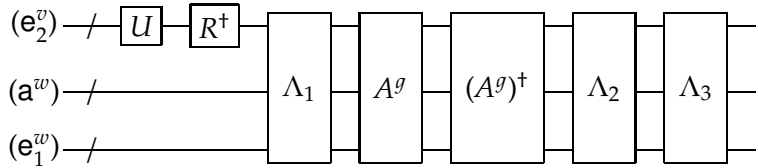
\begin{figure}[H]
\[
        \Qcircuit @C=1em @R=1.1em {
          (\reg{e}^v_\wout) &  & {/} \qw &  \gate{U} & \gate{R^\dagger} & \multigate{2}{\Lambda_1} & \multigate{2}{A^g} & \multigate{2}{(A^g)^\dagger} & \multigate{2}{\Lambda_2} & \multigate{2}{\Lambda_3} & \qw\\
          (\reg{a}^w) &    & {/} \qw & 	\qw	   &\qw &  \ghost{\Lambda_1} 	   & \ghost{A^g} & \ghost{(A^g)^\dagger} & \ghost{\Lambda_2} & \ghost{\Lambda_3} & \qw\\
           (\reg{e}^w_\win) & & {/} \qw &   \qw            &\qw & \ghost{\Lambda_1}       & \ghost{A^g} & \ghost{(A^g)^\dagger} & \ghost{\Lambda_2} & \ghost{\Lambda_3} & \qw
        } \\
 \]
\end{figure}
Here, $R$ denotes $\bigotimes_v R_{d_v,e_v}$, $\Lambda_2$ denotes $\bigotimes_v \Lambda_2(R_{d_v,e_v}, \ell^w, s^w, t^w)$, and $\Lambda_3$ is applied transversally in the circuit. The $A^g$ unitaries cancel out, and we are left with 
\[
\Lambda_3 \cdot \Lambda_2 \cdot \Lambda_1 \cdot \Big( \bigotimes_v R^\dagger_{d_v,e_v} \Big) \cdot  \Big( U (\qrv{y}^v)_v , \puretomixed{0 \cdots 0}, (\qrv{e}^w_\win)_w \Big) 
\]
By \Cref{lem:commute-correction}, this is equal to
\[
\Big( \bigotimes_w \sf{TP}_{\ell^w,s^w,t^w} \Big) \cdot \Big( \bigotimes_v R_{d_v,e_v} \Big) \cdot \Big( \bigotimes_v R^\dagger_{d_v,e_v} \Big) \,  \Big( U (\qrv{y}^v)_v , \puretomixed{0 \cdots 0}, (\qrv{e}^w_\win)_w \Big) 	
\]
The PX unitaries $(R_{d_v,e_v})$ cancel out, and we are left with
\[
\Big( \bigotimes_w \sf{TP}_{\ell^w,s^w,t^w} \Big) \,  \Big( U (\qrv{y}^v)_v , \puretomixed{0 \cdots 0}, (\qrv{e}^w_\win)_w \Big) 		
\]

\paragraph{Finishing the Proof.} Consider the registers $(\reg{y}^v, \reg{e}^v_\win,\reg{a}^v,  \reg{e}^v_\wout)_{v}$, $\reg{c}^g$, and $(\reg{a}^w, \reg{e}^w_\win)_{w}$ of the ``global'' state $\Big ( \evg_1 ( \, \hat{F}( \qrv{y} \, ; \, J) \, ), \qrv{q} \Big)$, where all randomness $J$ is fixed except for the randomization bits $(s^v,t^v)_{v \in \winwire(g)}$, which we've averaged over. The density matrix in this registers can be written as
\begin{equation}
\label{eq:sec-5}
\Big ( \bigotimes_v \ub{\tau^{\otimes 2}}{\reg{y}^v \reg{e}^v_\win} \otimes \E_{d_v,e_v} \, \ub{\puretomixed{\ell^v_{z,d_v}, \ell^v_{x,e_v}}}{\reg{z}^v \reg{x}^v} \otimes \ub{\puretomixed{0}}{\reg{b}^v} \Big)\otimes  \ub{\puretomixed{ \hf_{g,\off}}}{\reg{c}^g} \otimes \Big( \bigotimes_w \sf{TP}_{\ell^w,s^w,t^w} \Big) \,  \Big( \ub{U (\qrv{y}^v)_v}{(\reg{e}^v_\wout)_v} , \ub{\puretomixed{0 \cdots 0}}{(\reg{a}^w)_w}, \ub{(\qrv{e}^w_\win)_w}{(\reg{e}^w_\win)_w} \Big) 
\end{equation}
where as usual the indices $v,w$ range over $\winwire(g),\woutwire(g)$ respectively. 

We now average over $r^g$ and $A^g$. 
Consider the density matrix on registers $(\reg{y}^v ,\reg{e}^v_\win ,\reg{z}^v ,\reg{x}^v, \reg{b}^v)_{v}$ and $\reg{c}^g$, which can be written as
\begin{equation}
\label{eq:sec-5b}
	\qrv{t} = \E_{r^g, A^g} \, \bigotimes_v \ub{\tau^{\otimes 2}}{\reg{y}^v \reg{e}^v_\win} \otimes \E_{d_v,e_v} \, \ub{\puretomixed{\ell^v_{z,d_v}, \ell^v_{x,e_v}}}{\reg{z}^v \reg{x}^v} \otimes \ub{\puretomixed{0}}{\reg{b}^v} \otimes  \ub{\puretomixed{ \hf_{g,\off}}}{\reg{c}^g} 
\end{equation}
Observe that the density matrix $\qrv{t}$ does not depend on any other randomness of $J$.  The remainder of the global state can be seen to be equal to the encoding of the partial circuit $F_{> 1}$ with input $F_{\leq 1}(\qrv{y}) = U(\qrv{y})$, and furthermore does not depend on the randomness $J_{\leq 1} = (r^g, A^g, (s^v, t^v)_{v \in \winwire(g)} )$.

Thus, averaging the global state over $J$, we get
\[
	\E_J \, \Big ( \evg_1 ( \, \hat{F}( \qrv{y} \, ; \, J) \, ), \qrv{q} \Big) = \E_{J_{> 1}} \, (\hat{F}_{> 1}(F_{\leq 1}(\qrv{y} \, ; \, J_{>1}) \,,\, \qrv{q}) \otimes \qrv{t} \;.
\] 
We are almost done -- we just need to argue that $\qrv{t}$ is indistinguishable from a state $\qrv{\tilde{t}}$ that only depends on the topology $\scr{T}$. Note that for fixed $(d_v,e_v)_v$ and $A^g$, the density matrix
\begin{equation}
\label{eq:sec-6}
\E_{r^g} \puretomixed{\ell^v_{z,d_v}, \ell^v_{x,e_v}}\otimes  \puretomixed{ \hf_{g,\off}}
\end{equation}
is precisely the output distribution of the $\CRE$ encoding of the correction function $f(d_1,e_1,\ldots,d_p,e_p)$. The value $\hat{f}_{g,\off}$ corresponds to the offline part of the encoding, and the $\ell^v_{z,d_v}, \ell^v_{x,e_v}$ correspond to the labels of the input bits $d_v,e_v$. The privacy guarantee of $\CRE$ is that the density matrix in~\eqref{eq:sec-6} is $(t,\epsilon)$-indistinguishable from
\[
	\puretomixed{\CSim_\topol(f(d_1,e_1,\ldots,d_p,e_p))} = \puretomixed{\CSim_\topol \Big( \gadget{Corr}_1,\ldots,\gadget{Corr}_p \Big)}
\]
where $\topol$ is the topology of the circuit computing $f$, and the topology $\topol$ depends only on the topology $\scr{T}$ of $F$. Thus the density matrix $\qrv{t}$ is $(t,\epsilon)$-indistinguishable from a density matrix $\qrv{\tilde{t}}$, defined as follows:
\[
	\qrv{\tilde{t}} = \tau^{\otimes 2p} \otimes \E_{A^g, (d_v, e_v)_v} \puretomixed{\CSim_\topol \Big( \gadget{Corr}_1,\ldots,\gadget{Corr}_p \Big)} \otimes \puretomixed{0} \, . 
\]
Let $A^g = (A_1,\ldots,A_p)$, and note that $\gadget{Corr}_j$ only depends on $A_j$. Since the $A_j$'s are chosen independently from the randomization group $\scr{R}_{\kappa_j}$ for some label length $\kappa_j$, the distribution of the gadget $\gadget{Corr}_j$ when averaged over $A_j$ is going to be uniform over all correction gadgets corresponding to label length $\kappa_j$ (see the discussion at the end of \Cref{sec:corr-gadget} for more details). Thus, $\E_{A^g} \puretomixed{\CSim_\topol \Big( \gadget{Corr}_1,\ldots,\gadget{Corr}_p \Big)}$ is a density matrix that only depends on the topology $\scr{T}$ of $F$ (and is independent of all other randomness). Thus the density matrix in~\eqref{eq:sec-5b} is $(t,\epsilon)$-indistinguishable from $\qrv{t}$, which only depends on the topology $\scr{T}$.
 
 This 
completes the proof of \Cref{lem:single-eval}.

\appendix

\section{Comparison with Related Cryptographic Notions}
\label{sec:comparisonwithnotions}

We now compare quantum randomized encoding with other cryptographic primitives of a similar nature. While we focus on the quantum variants of these primitives, the distinctions are the same as in the classical versions.

\paragraph{Secure Multiparty Computation.} The goal of secure multiparty computation (MPC) is to allow a number of parties, each with their own private input $x_i$, to jointly compute a function $f(x_1,\ldots,x_n)$ such that no party can learn about others' inputs. There is a close connection between MPC and REs, in that REs can often be used to accomplish secure MPC. Indeed, Yao's garbled circuits scheme from the 80's was presented as a technique for achieving secure $2$-party computation, and its distillation into a separate primitive with concrete properties occurred later. Many protocols for secure quantum MPC have been constructed over the years (see, e.g.,~\cite{ben2006secure,unruh2010universally,DNS10,DNS12,DGJMS19}).

While RE is useful for constructing MPC (and sometimes the other way around), and while the security in both notions is defined using the simulation paradigm, inherently they are very different. While MPC is a communication protocol between multiple parties with inputs, RE is not a protocol and it only considers a single input (one could imagine RE as a single-message protocol where an encoder with an input sends a message to a decoder without an input). Since in the context of RE there is only a single party, there is always a trivial solution of computing the function locally. Therefore usually in RE we are concerned with other properties of the construction beyond its security, such as the complexity of encoding.

\paragraph{Homomorphic Encryption.} Fully homomorphic encryption (FHE) is a method to compute functions on inputs that are encrypted, without having to ever decrypt the information. FHE and RE also share some commonalities, and there are contexts where both techniques are used to accomplish secure computation and delegation of computation. However there are intrinsic differences between these two concepts. (See also \cite[Remark 1.4]{Applebaum17}.)

In FHE, a client encrypts an input $x$ and sends $\Enc(x)$ to a remote evaluator. The evaluator can then compute $\Enc(f(x))$ -- for \emph{any} function $f$ -- from the given ciphertext, without learning anything about $x$ or $f(x)$. 
It sends $\Enc(f(x))$ back to the client, who can use its secret decryption key to recover $f(x)$. 
It is important that the evaluator does not have access to the decryption key, because otherwise it will be able to learn the original input $x$. 

With REs, the client sends an encoding of both a function $f$ and an input $x$, from which the evaluator can compute the value $f(x)$ in the clear. The evaluator doesn't have to send any messages back to the client, and furthermore the evaluator cannot derive an encoding of $g(x)$ for some unrelated function $g$. 

If all we require is decomposable RE, then one can achieve this under minimal assumptions such as the existence of one-way functions (or even unconditionally in some cases, depending on the desired complexity properties). FHE (and homomorphic encryption in general) is only known under stronger assumptions (and cannot be achieved with unconditional security). In particular, candidates for classical and quantum FHE often rely on the hardness of the learning with errors problem (or related problems). Quantum FHE was considered recently in \cite{BJ15,DSS16,M18a,B18}.

\paragraph{Program Obfuscation.} 
In program obfuscation, an obfuscator encodes a function $f$ into an \emph{obfuscated program} $\Enc(f)$, which is sent to an evaluator. Using the obfuscated program, the evaluator can compute $f(x)$ for any choice of input $x$. The security requirement, intuitively, is that nothing about $f$ is revealed beyond its input-output functionality. The most commonly studied notions of obfuscation are \emph{virtual black-box (VBB) obfuscation} and \emph{indistinguishability obfuscation} (iO), which differ in how they hide the function $f$.

While an obfuscated program can be evaluated on multiple inputs (as mant as the user wishes), in RE the encoding fixes both the function and an input, so it can be thought of as a ``one-time obfuscation''. 
As explained in \cite[Section 4.4]{Applebaum17}, the obfuscation of the program which has $x$ hardwired and evaluates $f$ on it, constitutes a RE of $f(x)$. On the other hand, REs can be used to ``bootstrap'' obfuscators to have superior complexity properties.

There has been a formalization of quantum program obfuscation~\cite{alagic2016quantum}, but it is not yet known whether general quantum obfuscation can be achieved assuming only classical obfuscation. Broadbent and Kazmi have recently showed how to achieve indistinguishability obfuscation for quantum circuits with low $T$-gate count (at most logarithmically many)~\cite{broadbent2020indistinguishability}. As we mention in \Cref{sec:otherapps}, a classical RE scheme for quantum circuits can be combined with a classical obfuscator to imply a quantum obfuscator.

\section{Potential Applications of QRE}
\label{sec:otherapps}

Many of the applications of RE in the classical setting seem to carry over to the quantum setting, possibly with some necessary adjustments. As the variety of applications of RE is so vast, we view it as outside the scope of this paper to go over and attempt to re-prove them. We therefore highlight the most immediate ones here. 

\paragraph{PSM and Delegation.}
Two immediate applications that were mentioned above are private simultaneous messages (PSM) and delegation. Indeed, PSM is almost equivalent to decomposable RE, and therefore our results show how to achieve PSM in the quantum setting, using quantum messages and a quantum shared string (or even a classical shared string in the case where the input is completely classical). In terms of delegation, QRE implies that any quantum computation can be delegated (in $2$ messages) using, essentially, a $\qnczf$ verifier.\footnote{One should be careful since the final step of verification requires comparing two long classical strings, which can be done in $\qnczf$ with bounded error \cite{HS05}.}

\paragraph{Two-Party Secure Computation.}
Applications to MPC in the context of round reduction also seem to follow. In particular, one can use our construction to obtain an analogue of Yao's original $2$-message two-party MPC protocol using (classical) oblivious transfer (OT) as a building block. We recall that in OT, we have a \emph{receiver} with a bit $b$, and a \emph{sender} with two strings $r_0, r_1$, and in the end of the execution the receiver learns $r_b$ and the sender learns nothing about $b$.

We can consider two parties $A,B$, each of which holding a quantum input, $\qrv{x}, \qrv{y}$ respectively, and they wish to jointly compute a quantum operation $F$ on their inputs whose output is delivered to $A$.\footnote{Interestingly, contrary to the classical setting, this does not seem to immediately imply a protocol for the setting where both parties receive an output with the same round complexity (if we consider a general quantum operation). One additional round message seems to be needed in this case.} This can be done as follows. Party $A$ encrypts its input with a classical key using a quantum one-time pad (QOTP, \cite{AMTW00}), it sends the encrypted input to $B$ and conducts an OT protocol as a receiver, for each bit of the classical pad for the QOTP. Party $B$ considers a quantum functionality $F'$ that takes as input an encrypted $\qrv{x}$, (unencrypted) $\qrv{y}$, and classical QOTP key. On this input, $F'$ first decrypts $\qrv{x}$ and then applies the original $F$ on $\qrv{x}, \qrv{y}$. Party $B$ creates a decomposable QRE of $F'$, plugging in the encrypted $\qrv{x}$ that was received from $A$, and its own unencrypted input $\qrv{y}$. We recall that for classical input bits, our QRE is classical and decomposable, which means that for each classical input bit we can generate two labels $r_0, r_1$, such that if the value of the bit is $b$ then the part of the encoding that depends on this bit is $r_b$. This means that party $B$ can send the parts of the encoding that it can compute, and complete the OT protocol as a sender with values $r_0, r_1$ for each bit of classical input. This will allow party $A$ to obtain the encoding of $F$, apply the decoding procedure and learn the intended output.

It appears that one should be able to prove security of such a protocol in the \emph{specious} model \cite{DNS10}, which is the mildest model of security in the quantum setting, if the underlying classical OT primitive is also specious secure.\footnote{At a high level, we believe that a proof goes through by requiring all parties to run all functions of the protocol in a purified manner.} However, a formal proof is tangent to the scope of this work and one should only treat this as a candidate until a formal proof is presented.
We also note that protocols with comparable round complexity can be achieved using quantum fully homomorphic encryption \cite{BJ15,DSS16,M18a,B18}.

One could consider further applications in the context of MPC such as improved quantum MPC in the multi-party setting and in the malicious setting \cite{DNS12,DGJMS19}.

\def\sk{\mathsf{sk}}

\paragraph{Functional Encryption.} It was noticed in \cite{SahaiS10} that decomposable RE schemes imply a limited form of \emph{functional encryption} (FE). An encryption scheme where there are multiple secret keys, associated with functions, so that when $\sk_f$ decrypts $\Enc(x)$ the output is $f(x)$. In the classical setting RE implies FE without ``collusion resistance'' (i.e.\ an adversary should not be allowed to obtain more than a single functional key). It was then showed \cite{GorbunovVW12} how to extend this technique to FE with ``bounded collusion''. This construction seem to carry over to the quantum setting using our QRE scheme, under the appropriate definition of FE. However, some definitional work is required in order to formally substantiate the definition and show the connection in the quantum setting.

A more ambitious goal is to construct succinct FE schemes (even with bounded collusion) and so-called ``reusable'' garbled circuits which are function-private symmetric-key FE, analogously to the classical constructions in \cite{goldwasser2013reusable} (but possibly using different technique). One obstacle that seems to prevent direct application is the absence of a quantum attribute-based encryption schemes that are a central building block in that construction.

\paragraph{Classical Garbling and Quantum Obfuscation.} If it is possible to construct a QRE with classical encoding for classical inputs and function descriptions, then it would allow to construct quantum indistinguishability obfuscation (iO) from the classical variant, similarly to how classical RE is leveraged to obtain classical iO~\cite{Applebaum14obf,BitanskyCGHJLPT18}. Constructing iO for quantum circuits is one of the intriguing open problems in the context of quantum cryptography, and with the connections between iO and RE in the classical setting, one would hope that QRE could be a useful tool in establishing it.

\SetAlgoCaptionSeparator{}
\renewcommand{\figurename}{Circuit}

\section{Proof of \Cref{prop:p-h}}
\label{sec:p-h-proof}

Here we give a proof of \Cref{prop:p-h}, restated here for convenience:

\begin{proposition}
\label{prop:p-h-appendix}
For all single-qubit $PX$ gates $R$, $s = (s_z,s_x) \in \{0,1\}^2$, and strings $r \in \binset^\kappa$ there exist $\qnczf$ circuits $C_1(r)$ and $C_3$ and a depth-one Clifford circuit $C_2(R,r,s)$ in the randomization group $\scr{R}_\kappa$ such that the following circuit identity holds:
\begin{equation}
\label{eq:p-h-1}
        \Qcircuit @C=0.6em @R=0.8em {
          \reg{u} & &  & & \gate{R}  & \gate{H} & \ctrl{1} & \gate{X^{s_x}}  & \gate{Z^{s_z}} & \qw &		& & & & \multigate{2}{C_1(r)} & \multigate{2}{C_2(R,r,s)} & \multigate{2}{C_3} & \qw \\
          \reg{z} & & &  & {/} \qw & \qw  & \gate{X(r)} &  \qw & \qw & \qw & =   & & & &  \ghost{C_1(r)} & \ghost{C_2(R,r,s)} & \ghost{C_3} & \qw \\
          \regb && \ket{0}  & & {/} \qw & \qw   & \qw & \qw  & \qw & \qw  &			& & \ket{0}& & 		\ghost{C_1(r)} & \ghost{C_2(R,r,s)} & \ghost{C_3} & \qw         
        }
\end{equation}
Here, the register $\reg{b}$ consists of $(\kappa+1)^2$ qubits. Furthermore, descriptions of the circuits $C_1(r)$ and $C_2(R,r,s)$ can be computed by $\ncz$ circuits of size $O(\kappa^2)$.
\end{proposition}

A general PX group element $R$ can be written (up to a global phase) as $X^x \, Z^z \, P^p$ for some $x, z, p \in \{0,1\}$. We will perform the analysis for the case of $R = X$, $R = Z$, and $R = P$ separately, and then put everything together to deduce the Proposition.

Before doing so, we will transform the right hand side of~\eqref{eq:p-h-1} via sequence of circuit equivalences. First, we observe that the fan-out gate $CX(r)$ which maps $\ket{u,z_1,\ldots,z_\kappa}$ to $\ket{u, z_1 \oplus (u\cdot r_1), \ldots, z_\kappa \oplus (u \cdot r_\kappa)}$, when conjugated by Hadamards, becomes a parity gate $\oplus(r)$ that maps $\ket{u,z_1,\ldots,z_\kappa}$ to $\ket{u \oplus (z \cdot r),z_1,\ldots,z_\kappa}$ where $z \cdot r = \sum_i z_i \cdot r_i$. Next, we observe that we can move the $X^{s_x}$ and $Z^{s_z}$ gates before the Hadamard on register $\reg{u}$, which changes them to $Z^{s_x}$ and $X^{s_z}$. Thus we get that the right hand side of~\eqref{eq:p-h-1} is equivalent to

\begin{figure}[H]
\[
        \Qcircuit @C=0.6em @R=0.8em {
          \reg{u} 	& & 	& 	& \gate{R}  & \qw & \targ &  \qw & \gate{Z^{s_x}} & \gate{X^{s_z}} & \gate{H} & \qw \\
          \reg{z}_1 & & 	&  	& \gate{H} & \qw  &   \multigate{2}{\oplus(r)} \qwx &  \qw & \qw & \qw & \gate{H} & \qw \\
          		\vdots   & & 	&    &  \vdots     &      &    &   	&      &  & \vdots &  \\
          \reg{z}_\kappa & & &  &  \gate{H} & \qw  & \ghost{\oplus(r)} &  \qw & \qw & \qw & \gate{H} & \qw \\					
          \regb && \ket{0}  & & {/} \qw & \qw   & \qw & \qw  & \qw & \qw & \qw & \qw 
        }
\]
\caption{} \label{fig:p-h-1}
\end{figure}

Next, we apply a series of fan-out gates to copy the contents of the registers $\reg{u}, \reg{z}_1,\ldots,\reg{z}_\kappa$ to the ancilla register $\reg{b}$. We label the individual qubits of register $\reg{b}$ as $\reg{b}_{ij}$ where $i,j \in \{0,1,\ldots,\kappa\}$. We first apply a fan-out gate controlled on register $\reg{u}$, with targets on $\{ \regb_{00},\ldots,\regb_{0\kappa} \}$. Then for each $j \in [\kappa]$, we apply a fan-out gate controlled on register $\reg{z}_j$ with targets $\{ \regb_{j0},\ldots,\regb_{j\kappa} \}$. We call these the ``descending fan-outs''. Then, we apply the same series of fan-out gates in reverse, to undo the copying procedure. We call these the ``ascending fan-outs''. Thus Circuit~\ref{fig:p-h-1} is equivalent to Circuit~\ref{fig:p-h-2}. 

\begin{figure}[H]
\[
        \Qcircuit @C=0.6em @R=0.8em {
          \reg{u} 	& & 	& 	& \gate{R}  & \qw & \targ 							&  \ctrl{4} & \qw & \qw & \qw & \qw & \ctrl{4} & \gate{Z^{s_x}} & \gate{X^{s_z}} & \gate{H} & \qw \\
          \reg{z}_1 & & 	&  	& \gate{H} & \qw  &   \multigate{2}{\oplus(r)} \qwx &  \qw & \ctrl{4} & \qw & \qw & \ctrl{4} & \qw & \qw & \qw & \gate{H} & \qw \\
          \vdots		    & &  	&    &   \vdots    &       &   	&   	&      &  &  & & & & & \vdots  \\
          \reg{z}_\kappa & & &  &  \gate{H} & \qw  & \ghost{\oplus(r)} 				&  \qw & \qw & \ctrl{5} & \ctrl{5} & \qw & \qw & \qw & \qw & \gate{H} & \qw \\				
          \regb_0 && \ket{0}  & & {/} \qw 		& \qw   & \qw 						& \targ  & \qw & \qw & \qw & \qw & \targ & \qw & \qw & \qw & \qw \\
          \regb_1 && \ket{0}  & & {/} \qw 		& \qw   & \qw   					& \qw  & \targ & \qw & \qw & \targ & \qw & \qw & \qw & \qw  & \qw \\
          	\vdots	    & &  	&    &       	&  &  \vdots  				&   	&      &    & & & & & & &  \\          
	\\
          \regb_\kappa && \ket{0}  & & {/} \qw 	& \qw   & \qw 						& \qw  & \qw & \targ & \targ & \qw & \qw & \qw & \qw & \qw & \qw \\
        }
\]
\caption{} \label{fig:p-h-2}
\end{figure}

Next, we move the $Z^{s_x}$ and $X^{s_z}$ gates left, before the ascending fan-outs. When moving the $X^{s_z}$ gate past the fan-out controlled on register $\reg{u}$, however, this incurs an $X^{s_z}$ correction on each target of the fan-out, which are the $\{ \regb_{00}, \ldots, \regb_{0\kappa}\}$ registers. Thus Circuit~\ref{fig:p-h-2} is equivalent to Circuit~\ref{fig:p-h-3}.

\begin{figure}[H]
\[
        \Qcircuit @C=0.6em @R=0.8em {
          \reg{u} 	& & 	& 	& \gate{R}  & \qw & \targ 							&  \ctrl{4} & \qw & \qw & \gate{Z^{s_x}} & \gate{X^{s_z}} &  \qw & \qw & \qw & \qw & \ctrl{4} & \gate{H} & \qw \\
          \reg{z}_1 & & 	&  	& \gate{H} & \qw  &   \multigate{2}{\oplus(r)} \qwx &  \qw & \ctrl{4} & \qw & \qw & \qw & \qw & \qw & \qw &\ctrl{4} &  \qw &  \gate{H} & \qw \\
          \vdots		    & &  	&    &   \vdots    &       &   	&   	&      &  &  & & & & & & & \vdots  \\
          \reg{z}_\kappa & & &  &  \gate{H} & \qw  & \ghost{\oplus(r)} 				&  \qw & \qw & \ctrl{5} & \qw &\qw & \qw & \qw & \ctrl{5} & \qw & \qw  & \gate{H} & \qw \\				
          \regb_0 && \ket{0}  & & {/} \qw 		& \qw   & \qw 						& \targ  & \qw & \qw & \qw & \gate{(X^{s_z})^{\otimes \kappa}} & \qw &\qw & \qw &  \qw & \targ & \qw & \qw \\
          \regb_1 && \ket{0}  & & {/} \qw 		& \qw   & \qw   					& \qw  & \targ & \qw & \qw & \qw & \qw &  \qw & \qw & \targ & \qw & \qw & \qw \\
          	\vdots	    & &  	&    &       	&  &  \vdots  				&   	&      &    & & & & &&  & & &  \\          
	\\
          \regb_\kappa && \ket{0}  & & {/} \qw 	& \qw   & \qw 						& \qw  & \qw & \targ & \qw & \qw & \qw & \qw & \targ & \qw & \qw & \qw & \qw \\
        }
\]
\caption{} \label{fig:p-h-3}
\end{figure}

For now we focus on the part of the circuit depicted in Circuit~\ref{fig:p-h-4} -- we omit the first layer of Hadamard gates, and everything past the descending fan-out gates. We show, for different gates $R$, how to ``push'' the $R$ correction past the descending fan-out gates.

\begin{figure}[H]
\[
        \Qcircuit @C=0.6em @R=0.8em {
          \reg{u} 	& & 	& 	& \gate{R}  & \qw & \targ 							&  \ctrl{4} & \qw & \qw & \qw & \qw &  \qw & \qw \\
          \reg{z}_1 & & 	&  	& \qw & \qw  &   \multigate{2}{\oplus(r)} \qwx &  \qw & \ctrl{4} & \qw & \qw & \qw & \qw  & \qw \\
          \vdots		    & &  	&    &       &       &     	&   	&      &  &  & & \\
          \reg{z}_\kappa & & &  &  \qw & \qw  & \ghost{\oplus(r)} 				&  \qw & \qw & \ctrl{5} & \qw &\qw & \qw & \qw  \\				
          \regb_0 && \ket{0}  & & {/} \qw 		& \qw   & \qw 						& \targ  & \qw & \qw & \qw & \qw & \qw  & \qw \\
          \regb_1 && \ket{0}  & & {/} \qw 		& \qw   & \qw   					& \qw  & \targ & \qw & \qw & \qw & \qw &  \qw  \\
          	\vdots	    & &  	&    &       	&  &  \vdots  				&   	&      &    & & & &   \\          
	\\
          \regb_\kappa && \ket{0}  & & {/} \qw 	& \qw   & \qw 						& \qw  & \qw & \targ & \qw & \qw & \qw & \qw \\
        }
\]
\caption{} \label{fig:p-h-4}
\end{figure}

\paragraph{Case 1.} Suppose that $R = X^x$. First, the $X^x$ gate commutes with the parity gate. Moving it past the first fan-out gate that is controlled on the register $\reg{u}$ incurs a $X^x$ correction on each of the $\{ \regb_{0j} \}_j$ registers. Thus Circuit~\ref{fig:p-h-4} is equivalent to Circuit~\ref{fig:p-h-5}.

\begin{figure}[H]
\[
        \Qcircuit @C=0.6em @R=0.8em {
          \reg{u} 	& & 	& 	& \qw  & \qw & \targ 							&  \ctrl{4} & \qw & \qw & \gate{X^x} & \qw &  \qw & \qw \\
          \reg{z}_1 & & 	&  	& \qw & \qw  &   \multigate{2}{\oplus(r)} \qwx &  \qw & \ctrl{4} & \qw & \qw & \qw & \qw  & \qw \\
          \vdots		    & &  	&    &       &       &    	&   	&      &  &  & & \\
          \reg{z}_\kappa & & &  &  \qw & \qw  & \ghost{\oplus(r)} 				&  \qw & \qw & \ctrl{5} & \qw &\qw & \qw & \qw  \\				
          \regb_0 && \ket{0}  & & {/} \qw 		& \qw   & \qw 						& \targ  & \qw & \qw & \gate{(X^x)^{\otimes \kappa}} & \qw & \qw  & \qw \\
          \regb_1 && \ket{0}  & & {/} \qw 		& \qw   & \qw   					& \qw  & \targ & \qw & \qw & \qw & \qw &  \qw  \\
          	\vdots	    & &  	&    &       	&  &  \vdots  				&   	&      &    & & & &   \\          
	\\
          \regb_\kappa && \ket{0}  & & {/} \qw 	& \qw   & \qw 						& \qw  & \qw & \targ & \qw & \qw & \qw & \qw \\
        }
\]
\caption{} \label{fig:p-h-5}
\end{figure}

\paragraph{Case 2.} Suppose now that $R = Z^z$. Moving the $Z^z$ gate past the parity gate incurs a $Z^{z \cdot r_j}$ correction on the register $\reg{z}_j$. This can be seen by repeatedly applying the identity $CNOT_{a,b} \cdot (I_a \otimes Z_b) = (Z_a \otimes Z_b) CNOT_{a,b}$, where $CNOT_{a,b}$ denotes a CNOT with control $a$ and target $b$.  Furthermore, the $Z$ gates all commute with the fan-out gates. Thus Circuit~\ref{fig:p-h-4} is equivalent to Circuit~\ref{fig:p-h-6}.

\begin{figure}[H]
\[
        \Qcircuit @C=0.6em @R=0.8em {
          \reg{u} 	& & 	& 	& \qw  & \qw & \targ 							&  \ctrl{4} & \qw & \qw & \gate{Z^z} & \qw &  \qw & \qw \\
          \reg{z}_1 & & 	&  	& \qw & \qw  &   \multigate{2}{\oplus(r)} \qwx &  \qw & \ctrl{4} & \qw & \gate{Z^{z \cdot r_1}} & \qw & \qw  & \qw \\
          \vdots		    & &  	&    &       &       &    	&   	&      &  &  \vdots & & \\
          \reg{z}_\kappa & & &  &  \qw & \qw  & \ghost{\oplus(r)} 				&  \qw & \qw & \ctrl{5} & \gate{Z^{z \cdot r_\kappa}} &\qw & \qw & \qw  \\				
          \regb_0 && \ket{0}  & & {/} \qw 		& \qw   & \qw 						& \targ  & \qw & \qw & \qw & \qw & \qw  & \qw \\
          \regb_1 && \ket{0}  & & {/} \qw 		& \qw   & \qw   					& \qw  & \targ & \qw & \qw & \qw & \qw &  \qw  \\
          	\vdots	    & &  	&    &       	&  &  \vdots  				&   	&      &    & & & &   \\          
	\\
          \regb_\kappa && \ket{0}  & & {/} \qw 	& \qw   & \qw 						& \qw  & \qw & \targ & \qw & \qw & \qw & \qw \\
        }
\]
\caption{} \label{fig:p-h-6}
\end{figure}

\paragraph{Case 3.} Suppose now that $R = P^p$. Then we claim that Circuit~\ref{fig:p-h-4} is equivalent to Circuit~\ref{fig:p-h-7}. Here, if $p = 0$, then the gate $\Gamma_{p,r}$ is the identity. Otherwise, the gate $\Gamma_{p,r}$ stands for the following tensor-product of $CZ$ gates:\footnote{Recall that a $CZ$ gate is a controlled-$Z$ gate; it maps $\ket{a,b}$ to $(-1)^{ab} \ket{a,b}$.}
\begin{enumerate}
	\item $CZ(\regb_{0j},\regb_{j0})$ for all $j$ such that $r_j = 1$, and
	\item $CZ(\regb_{ij}, \regb_{ji})$ for all $i < j$ such that $r_i = r_j = 1$. 
\end{enumerate}
Notice that all of these $CZ$ gates are disjoint, and thus can be implemented in a single layer.

\begin{figure}[H]
\[
        \Qcircuit @C=0.6em @R=0.8em {
          \reg{u} 	& & 	& 	& \qw  & \qw & \targ 							&  \ctrl{4} & \qw & \qw & \gate{P^p} & \qw &  \qw & \qw \\
          \reg{z}_1 & & 	&  	& \qw & \qw  &   \multigate{2}{\oplus(r)} \qwx &  \qw & \ctrl{4} & \qw & \gate{P^{p \cdot r_1}} & \qw & \qw  & \qw \\
          \vdots		    & &  	&    &       &       &   	&   	&      &  &  \vdots & & \\
          \reg{z}_\kappa & & &  &  \qw & \qw  & \ghost{\oplus(r)} 				&  \qw & \qw & \ctrl{4} & \gate{P^{p \cdot r_\kappa}} &\qw & \qw & \qw  \\				
          \regb_0 && \ket{0}  & & {/} \qw 		& \qw   & \qw 						& \targ  & \qw & \qw & \multigate{3}{\Gamma_{p,r}} & \qw & \qw  & \qw \\
          \regb_1 && \ket{0}  & & {/} \qw 		& \qw   & \qw   					& \qw  & \targ & \qw & \ghost{\Gamma_{p,r}} & \qw & \qw &  \qw  \\
          	\vdots & &  	   &  &              	&  &  \vdots  				    &    	&      &   	& 	&  &    &   \\          
          \regb_\kappa && \ket{0}  & & {/} \qw 	& \qw   & \qw 						& \qw  & \qw & \targ & \ghost{\Gamma_{p,r}} & \qw & \qw & \qw \\
        }
\]
\caption{} \label{fig:p-h-7}
\end{figure}

This can be verified by calculating the behavior of both circuits. For simplicity assume that $p = 1$. Circuit~\ref{fig:p-h-4} on input $\ket{u,z_1,\ldots,z_\kappa}$ produces an output state that has the following structure:
\begin{itemize}
	\item It is pre-multiplied by a phase factor $i^{u}$
	\item Registers $\reg{u}_0$ and $\reg{b}_{0j}$ are in the state $\ket{u \oplus (z \cdot r)}$ for all $j$.
	\item Registers $\reg{z}_i$ and $\reg{b}_{ij}$ are in the state $\ket{z_i}$ for all $i > 0$ and all $j$.
\end{itemize}
On the other hand, the state of Circuit~\ref{fig:p-h-7} is exactly the same except the phase factor in front is the product of
\begin{itemize}
	\item $i$ raised to the power $u + (z \cdot r) \mod 2$ (this comes from the $P$ gate on register $\reg{u}$).
	\item $i$ raised to the power $\sum_i z_i $ (this comes from the $P$ gates on registers $\{\reg{z}_i\}$).
	\item $-1$ raised to the power $(u +1)\sum_{j : r_j = 1} z_j $ (this comes from the CZ gates applied to $(\regb_{0j},\regb_{j0})$ for all $j$ such that $r_j = 1$).
	\item $-1$ raised to the power $\sum_{i < j : r_i = r_j = 1} z_i z_j $ (this comes from the CZ gates applied to $(\regb_{ij},\regb_{ji})$ for all $i < j$ such that $r_i = r_j = 1$).
\end{itemize}
Using the identity that for bits $c_1,\ldots,c_n$,
\[
	c_1 + \cdots + c_n \mod 2 = \sum_i c_i - 2\sum_{i < j} c_i c_j \mod 4,
\]
we get that these phase factors are equivalently
\begin{itemize}
	\item $i$ raised to the power $u + \sum_{i: r_i = 1} z_i - 2 u \sum_{j : r_j = 1} z_j - 2\sum_{i < j : r_i = r_j = 1} z_i z_j$
	\item $i$ raised to the power $\sum_{i : r_i = 1} z_i$
	\item $i$ raised to the power $2 (u +1)\sum_{j : r_j = 1} z_j $
	\item $i$ raised to the power $2  (\sum_{i < j : r_i = r_j = 1} z_i z_j)$
\end{itemize}
Summing all of these exponents modulo $4$, we get that the phase factor is $i^{u}$, as desired.

\paragraph{General case.}  Suppose that $R$ is a general $PX$ group element, so it can be written (up to a global phase) as $R =  X^x Z^z P^p$ for $x,z,p \in \{0,1\}$. Then by combining Cases 1, 2 and 3 we determine that Circuit~\ref{fig:p-h-1} is equivalent to Circuit~\ref{fig:p-h-8}.

\begin{figure}[H]
\[
        \Qcircuit @C=0.6em @R=0.8em {
          \reg{u} 	& & 	& 	& \qw  & \qw & \targ 							&  \ctrl{4} & \qw & \qw & \qw & \gate{P^p} & \gate{Z^z} & \gate{X^x} & \gate{Z^{s_x}} & \gate{X^{s_z}} &  \qw & \qw & \qw & \qw & \ctrl{4} & \gate{H} & \qw \\
          \reg{z}_1 & & 	&  	& \gate{H} & \qw  &   \multigate{2}{\oplus(r)} \qwx &  \qw & \ctrl{4} & \qw & \qw & \gate{P^{p \cdot r_1}} & \gate{Z^{z \cdot r_1}} & \qw & \qw & \qw & \qw & \qw & \qw &\ctrl{4} &  \qw &  \gate{H} & \qw \\
          \vdots		    & &  	&    &   \vdots    &       &  \ghost{\oplus(r)}   	&   	&      &  &  & \vdots  & \vdots & & & & &  & & & & \vdots  \\
          \reg{z}_\kappa & & &  &  \gate{H} & \qw  & \ghost{\oplus(r)} 				&  \qw & \qw & \ctrl{5} & \qw & \gate{P^{p \cdot r_\kappa}} & \gate{Z^{z \cdot r_\kappa}} & \qw &\qw & \qw & \qw & \qw  & \ctrl{5} & \qw & \qw  & \gate{H} & \qw \\				
          \regb_0 && \ket{0}  & & {/} \qw 		& \qw   & \qw 						& \targ  & \qw & \qw & \qw & \multigate{4}{\Gamma_{p,r}} & \gate{(X^x)^{\otimes \kappa}} & \qw & \qw & \gate{(X^{s_z})^{\otimes \kappa}} & \qw &\qw & \qw &  \qw & \targ & \qw & \qw \\
          \regb_1 && \ket{0}  & & {/} \qw 		& \qw   & \qw   					& \qw  & \targ & \qw & \qw & \ghost{\Gamma_{p,r}} & \qw & \qw & \qw & \qw & \qw &  \qw & \qw & \targ & \qw & \qw & \qw \\
          	\vdots	    & &  	&    &       	& 	 &  				&   	&      &    &  & & & && &  & & &  \\          
	          		    & &  	&    &       	& 	 &  				&   	&      &    &  & & & && &  & & &  \\        
          \regb_\kappa && \ket{0}  & & {/} \qw 	& \qw   & \qw 						& \qw  & \qw & \targ & \qw & \ghost{\Gamma_{p,r}} & \qw & \qw & \qw & \qw & \qw& \qw& \targ & \qw  & \qw & \qw & \qw \\
        }
\]
\caption{} \label{fig:p-h-8}
\end{figure}

Note that Circuit~\ref{fig:p-h-8} has the desired structure:
\begin{itemize}
	\item Letting $C_1(r)$ denote the the subcircuit up to and including the descending fan-out gates, we see that $C_1$ only depends on $r$ and is independent of $R$ and $s$. Furthermore, $C_1(r)$ can be implemented as a $\qnczf$ circuit. 
	\item Letting $C_2(R,r,s)$ denote the subcircuit in between the descending and ascending fan-out gates,  we see that $C_2$ depends on $R, r, s$, and can be implemented as a tensor product of single-qubit Clifford unitaries on all registers except for the pairs $(\reg{b}_{0j},\reg{b}_{j0})$ for $j > 0$, which have two-qubit Clifford unitaries acting on them (namely, products of $X$ and $CZ$ gates).
	\item Letting $C_3$ denote the subcircuit that includes the ascending fan-out gates as well as the final layer of Hadamard gates, we see that $C_3$ is independent of $R, r, s$, and can be implemented as a $\qnczf$ circuit.
\end{itemize}

Furthermore, the description of the circuit $C_2(R,r,s)$ in Lines 7 through 12 shows that each of the single- and two-qubit gates are simple functions of the inputs $(R,r,s)$, so this description can be computed by a classical $\ncz$ circuit of size $O(\kappa^2)$. 

This concludes the proof of the Proposition.

\bibliographystyle{alpha}
\bibliography{qgc}

\end{document}